\documentclass[10pt,a4paper]{article}

\usepackage[utf8]{inputenc}
\usepackage{amsmath}
\usepackage{amsfonts}
\usepackage{amssymb}
\usepackage{amsthm}
\usepackage[left=2cm,right=2cm,top=2cm,bottom=2cm]{geometry}
\usepackage{url}
\usepackage{graphicx}
\usepackage{algorithm}
\usepackage{algorithmic}
\usepackage{wrapfig}
\usepackage{enumitem}
\usepackage{booktabs}
\usepackage{color}
\usepackage{multirow}
\usepackage{xfrac}
\usepackage[usenames, dvipsnames, table]{xcolor}

\newtheorem{theorem}{Theorem}[section]
\newtheorem{corollary}[theorem]{Corollary}
\newtheorem{lemma}[theorem]{Lemma}

\newtheorem{definition}[theorem]{Definition}
\newtheorem{remark}[theorem]{Remark}

\newcommand{\ra}[1]{\renewcommand{\arraystretch}{#1}}

\def\Ph{\mathbf{X}}

\def\Gm{\Gamma}

\def\RR{\mathbb{R}}

\newcommand{\support}{\mathcal S}

\newcommand{\model}{\mathcal M}

\newcommand{\obs}{\mathbf y}
\newcommand{\x}{\boldsymbol{\beta}}
\newcommand{\z}{\mathbf z}
\newcommand{\w}{\mathbf w}
\newcommand{\X}{\mathbf X}
\newcommand{\bestx}{\x^{\star}} 

\newcommand{\numsam}{n}
\def \dim {p}      

\newcommand{\noise}{\boldsymbol{\xi}}
\newcommand{\sparsity}{k}

\newcommand{\EE}{\mathcal E}

\newcommand{\Real}{\mathbb{R}}

\newcommand{\bPhi}{\X}

\newcommand{\bitem}{\begin{itemize}}
\newcommand{\eitem}{\end{itemize}}

\DeclareMathOperator*{\argmax}{argmax}        
\DeclareMathOperator*{\argmin}{argmin}

\newcommand{\beqn}{\begin{equation}}
\newcommand{\eeqn}{\end{equation}}
\newcommand{\balign}{\begin{align}}
\newcommand{\ealign}{\end{align}}

\newcommand{\R}{ \mathbb{R} }     

\title{Convex block-sparse linear regression with expanders -- provably}
\author{Anastasios Kyrillidis \\ UT Austin \and 
 Bubacarr Bah \\ UT Austin \and Rouzbeh Hasheminezhad \\ Sharif University of Technology \and Quoc Tran-Dinh \\ University of North Carolina \and Luca Baldassarre \\ EPFL \and Volkan Cevher \\ EPFL}

\begin{document}
\maketitle

\begin{abstract}
Sparse matrices are favorable objects in machine learning and optimization. 
When such matrices are used, in place of dense ones, the overall complexity requirements in optimization can be significantly reduced in practice, both in terms of space and run-time. 
Prompted by this observation, we study a convex optimization scheme for block-sparse recovery from linear measurements. 
To obtain linear sketches, we use expander matrices, \textit{i.e.}, sparse matrices containing only few non-zeros per column. Hitherto, to the best of our knowledge, such algorithmic solutions have been only studied from a non-convex perspective. Our aim here is to theoretically characterize the performance of convex approaches under such setting. 

Our key novelty is the expression of the recovery error in terms of the model-based norm, while assuring that solution lives in the model. 
To achieve this, we show that sparse model-based matrices satisfy a \emph{group} version of the null-space property. 
Our experimental findings on synthetic and real applications support our claims for faster recovery in the convex setting -- as opposed to using dense sensing matrices, while showing a competitive recovery performance.  
\end{abstract}

\section{Introduction} \label{sec:intro}

Consider the {\em sparse recovery problem} in the linear setting: given a measurement matrix $\Ph \in \R^{\numsam \times \dim}$ ($\numsam < \dim$), we seek a sparse vector $\x^\star \in \R^{\dim}$ that satisfies a set of measurements $\obs = \Ph \x^\star$ or finds a solution in the feasible set $\left\{\x ~:~\|\obs - \bPhi \x\|_2 \leq \|\noise\|_2\right\}$, when noise $\noise \in \R^\numsam$ is present. 
This problem is closely related to compressive sensing \cite{donoho2006compressed, candes08introduction} and the subset selection problem \cite{tibshirani1996regression}, as well as to graph sketching \cite{ahn2012graph} and data streaming \cite{nelson2012deterministic}. 

Typically, the analysis included in such task focuses on investigating conditions on $\Ph$ and the recovery algorithm $\Delta$ to obtain estimate $\widehat{\x} = \Delta(\Ph\x^\star + \noise)$ that yields a guarantee:
\begin{equation}
\label{eqn:cs_error_bnd}
\|\x^\star - \widehat{\x}\|_{q_1} \leq C_1 \|\x^\star - \x_s^\star\|_{q_2} + C_2\|\noise\|_{q_1}.
\end{equation}
Here, $\x_s^\star$ denotes the best $s$-sparse approximation to $\x^\star$ for constants $C_1, C_2>0$ and $1\leq q_2 \leq q_1 \leq 2 $ \cite{candes2006stable,berinde2008combining}. This type of guarantee is known as the $\ell_{q_1}/\ell_{q_2}$ error guarantee.

Without further assumptions, there is a range of recovery algorithms that achieve the above, both from a non-convex \cite{needell2009cosamp, kyrillidis2011recipes, kyrillidis2012combinatorial} and a convex perspective \cite{chen1998atomic, berinde2008combining}. \emph{In this work, we focus on the latter case.} 
Within this context, we highlight $\ell_1$-minimization method, also known as {\em Basis pursuit} (BP) \cite{chen1998atomic}, which is one of the prevalent schemes for compressive sensing: 
\begin{equation}
\min_{\x} \|\x\|_{1} \quad \mbox{subject to: }  \obs = \Ph \x. \label{eq:BP}
\end{equation} 
Next, we discuss common configurations for such problem settings, as in \eqref{eq:BP}, and how a priori knowledge can be incorporated in optimization. We conclude this section with our contributions.

\definecolor{maroon}{cmyk}{1,0.4,0,0}

\noindent \textbf{Measurement matrices in sparse recovery:} 
Conventional wisdom states that the best compression performance, \textit{i.e.}, using the least number of measurements to guarantee \eqref{eqn:cs_error_bnd}, is achieved by random, independent and identically distributed \emph{dense} sub-gaussian matrices \cite{candes2005decoding}. Such matrices are known to satisfy the \emph{$\ell_q$-norm restricted isometry property} (RIP-$q$) with high-probability. That is, for some $\delta_s \in (0, 1)$ and $\forall ~s\text{-sparse}~\x \in \R^\dim$ \cite{candes2006robust, berinde2008combining}: 
\begin{equation}
\label{eqn:rip}
(1-\delta_s)\|\x\|_q^q \leq \|\Ph\x\|_q^q \leq (1+\delta_s)\|\x\|_q^q.
\end{equation} In this case, we can tractably and provably approximate sparse signals from $\numsam = \mathcal{O}\left(s \log\left(\dim/s\right)\right)$ sketches for generic $s$-sparse signals. 
\emph{Unfortunately, dense matrices require high time- and storage-complexity when applied in algorithmic solutions.}\footnote{There exist \emph{structured} dense matrices that achieve better time and storage complexity, as compared to random Gaussian matrices; see \cite{price2013sparse} for discussion on subsampled Fourier matrices for sparse recovery. However, sparse design matrices have been proved to be  more advantageous over such structured dense matrices \cite{berinde2008combining}.}

On the other hand, in several areas, such as data stream computing and combinatorial group testing \cite{nelson2012deterministic}, it is vital to use a {\em sparse} sensing matrix $\Ph$, containing only very few non-zero elements per column \cite{indyk2008near,berinde2008combining}. In this work, we consider random sparse matrices $\bPhi $ that are adjacency matrices of expander graphs; a precise definition is given below.
\begin{definition}
\label{def:mdlexpander}
Let a graph $\left(\dim, \numsam, \EE \right)$ be a left-regular bipartite graph with $\dim$ left (variable) nodes, $\numsam$ right (check) nodes, a set of edges $\EE$ and left degree $d$.
If, any $\support \subseteq \left\{1, \dots, \dim\right\}$ with $|\support| = s \ll \dim$ and for some $\epsilon_s \in (0,1/2)$, we have that the set of neighbours of $\support$, $\Gamma(\support)$, satisfy $|\Gamma(\support)| \geq (1-\epsilon_s)d|\support|$, then the graph is called a {\em $(s,~d,~\epsilon_s)$-lossless expander graph}. The adjacency matrix of an expander graph is denoted as $\bPhi \in \left\{0, 1\right\}^{\numsam \times \dim}$ and is called as \emph{expander matrix}.
\end{definition}

\noindent In the above, the attribute $|\Gamma(\support)| \geq (1-\epsilon_s)d|\support|$ is known as the \emph{expansion} property, where $\epsilon_s \ll 1$ \cite{jafarpour2009efficient}. As \cite{berinde2008combining} states, expander matrices satisfy the RIP-1 condition if $\forall s\mbox{-sparse vector } \x \in \R^\dim$, we have:
\begin{equation}
\label{eqn:mrip1}
(1-2\epsilon_s)d\|\x\|_1 \leq \|\Ph\x\|_1 \leq d\|\x\|_1.
\end{equation} 
Since the recovery algorithms typically use $\Ph$ and its adjoint $\Ph^T$ as subroutines over vectors, \emph{sparse matrices have low computational complexity, even in the convex case \eqref{eq:BP}}, without any loss in sample complexity $\numsam = \mathcal{O}\left(s \log\left(\dim/s\right)\right)$; see Section \ref{sec:experiments} for some illustrative results on this matter. Along with computational benefits, \cite{berinde2008combining} provides recovery guarantees of the form \eqref{eqn:cs_error_bnd} for BP recovery where $q_1 = q_2 = 1$. 

\noindent \textbf{Model-based sparse recovery:} 
Nevertheless, 
sparsity is merely a first-order description of $\x^\star$ and in many applications we have considerably more information \emph{a priori}. In this work, we consider the \emph{$\sparsity$-sparse block model} \cite{yuan2006model, baldassarre2013group, kyrillidis2015structured}: 


\medskip
\begin{definition}{\label{def:block}}
We denote a {\em block-sparse structure} by $\model := \left\{ \mathcal{G}_1, \dots, \mathcal{G}_M\right\}$ where $\mathcal{G}_i \subseteq [\dim], |\mathcal{G}_i| = g$ for $i = 1, \ldots, M$, $\mathcal{G}_i \cap \mathcal{G}_j = \emptyset, ~i \neq j$, and $M$ is the total number of groups. Given a group structure $\model$, the {\em $\sparsity$-sparse block} model is defined as the collection of sets $\model_{\sparsity} := \left\{ \mathcal{G}_{i_1}, \dots, \mathcal{G}_{i_{\sparsity}}\right\}$ of $\sparsity$ groups from $\model$.
\end{definition} 

We consider block-sparse models such that $\displaystyle \cup_{\mathcal{G} \in \model}\mathcal{G} \equiv \left\{1, \dots, \dim\right\}$ and $g = \dim/M$ where $M$ is the number of groups. 
The idea behind group models is the identification of group of variates that should be either selected or discarded (\textit{i.e.}, set to zero) together. Such settings naturally appear in applications such as gene expression data \cite{zhou2010association} and neuroimaging \cite{gramfort2009improving}.

As an indicator of what can be achieved, Model-based Compressive Sensing \cite{baraniuk2010model}, leverages such models with \emph{dense sensing matrices} to provably reduce the number of measurements for stable recovery from $\mathcal{O}\left(\sparsity g \log \left(\dim / (\sparsity g\right)\right)$ to $\mathcal{O}\left(\sparsity g + \sparsity \log (M/\sparsity)\right)$ using \emph{non-convex} schemes\footnote{Most of the non-convex approaches known heretofore consider a (block) sparse constrained optimization criterion, where one is minimizing a data fidelity term (\textit{e.g.}, a least-squares function) over cardinality constraints.}; see \cite{hegde2015nearly} for a more recent discussion and extensions on this matter. Along this research direction, \cite{indyk2013model} and \cite{bah2014model} propose non-convex algorithms for equally-sized and variable-sized block-sparse signal recovery, respectively, using sparse matrices. To show reduction in sampling complexity, $\bPhi$ is assumed to satisfy the \emph{model RIP-1} \cite{indyk2013model}:\footnote{While probabilistic constructions of expander matrices satisfying model RIP-1 coincide with those of regular RIP-1 expander matrices, the model-based assumption on signals $\x$ in \eqref{eqn:mrip2} results into constructions of expanders $\Ph$ with less number of rows and guaranteed signal recovery.}
\begin{equation}
\label{eqn:mrip2}
(1-2\epsilon_{\model_{\sparsity}})d\|\x\|_1 \leq \|\Ph\x\|_1 \leq d\|\x\|_1, 
\end{equation} $\forall \sparsity\mbox{-block-sparse vector	}~ \x \in \R^\dim$, where $\epsilon_{\model_\sparsity}$ represents the \emph{model-based} expansion constant. 




\subsection{Contributions}
Restricting error guarantees to only variations of standard $\ell_{q_1} / \ell_{q_2}$-norms might be inconvenient in some applications. Here, we broaden the results in \cite{indyk2011k} towards having non-standard $\ell_q$ distances in approximation guarantees.
To the best of our knowledge, this work is the first attempt for provable convex recovery with approximate guarantees in the appropriate norm and using sparse matrices. Similar results -- but for a different model -- for the case of dense (Gaussian) sensing matrices are presented in \cite{eldar2009robust}.

In particular, in Section \ref{sec:results}, we provide provable error guarantees for the convex criterion:
\begin{equation}{\label{eq:modelBP}}
\min_{\x} \|\x\|_{2,1} \quad \mbox{subject to: }  \obs = \Ph \x,
\end{equation} where $\|\cdot\|_{2,1}$ is the $\ell_{2,1}$-norm; see next Section for details. Our key novelty is to provide $\ell_{2,1}/\ell_{2,1}$ approximation guarantees using sparse matrices for sensing; see Section \ref{sec:proof}.
In practice, we show the merits of using sparse matrices in convex approaches for \eqref{eq:modelBP}, both on synthetic and real data cases; see Section \ref{sec:experiments}.

%
%

\section{Preliminaries} \label{sec:prelim}

Scalars are mostly denoted by plain letters (\emph{e.g.} $\sparsity$, $\dim$), vectors by lowercase boldface letters (\emph{e.g.}, ${\bf x}$), matrices by uppercase boldface letters (\emph{e.g.}, ${\bf A}$) and sets by uppercase calligraphic letters (\emph{e.g.}, $\mathcal{S}$), with the exception of $[\dim]$ which denotes the index set $\{1, \ldots, \dim\}$.
Given a set $\mathcal{S}$ in item space $\mathcal{I}$ such that $\mathcal{S} \subseteq \mathcal{I}$, we denote its complement by $\mathcal{I} \setminus \mathcal{S}$; we also use the notation $\mathcal{S}^c$, when the item space is clearly described by context. For vector $\x\in\RR^\dim$, $\x_\mathcal{S}$ denotes the restriction of $\x$ onto $\mathcal{S}$, \emph{i.e.}, $(\x_\mathcal{S})_i = \x_i$ if $i \in \mathcal{S}$ and $0$ otherwise.
We use $|\mathcal{S}|$ to denote the cardinality of a set $\mathcal{S}$.
The $\ell_p$ norm of a vector ${\bf x} \in \Real^\dim$ is defined as $\|{\bf x}\|_p := \left ( \sum_{i=1}^\dim |x_i|^p \right )^{1/p}$.

\noindent \textbf{The $\ell_{2,1}$-norm:} The first convex relaxation for block-sparse approximation is due to \cite{yuan2006model}, who propose group sparsity-inducing $\ell_{2,1}$ convex norm:
\begin{equation*}
\|\x\|_{2,1} := \sum_{\mathcal{G} \in \model} w_i \|\x_{\mathcal{G}}\|_2,~~w_i > 0.
\end{equation*} The $\ell_{2,1}$-norm construction follows group supports according to a \emph{predefined} model $\mathcal{M}$ and promotes sparsity in the group level. 
We assume all the $w_i = 1$, but may be otherwise for a more general setting \cite{jacob2009group}.

Given a vector $\w \in \R^\dim$, we define $\model_{\sparsity}^\star$ as the \emph{best $\sparsity$-block sparse index set} according to:
\begin{align*}
\model_{\sparsity}^\star \in \argmin_{\model_{\sparsity} \subseteq \model} \|\w - \w_{\model_{\sparsity}}\|_{2,1}
\end{align*}

\section{Problem statement}

To properly set up the problem and our results, consider the following question:

\medskip
\noindent \textsc{Question:} {\it Let $\model$ be a predefined and known a priori block sparse model of $M$ groups of size $g$ and $\sparsity > 0$ be a user-defined group-sparsity level. Consider $\Ph \in \left\{0, 1\right\}^{\numsam \times \dim} (\numsam < \dim$) be a known $(\sparsity\cdot g, ~d, ~\epsilon_{\model_{\sparsity}})$ expander \emph{sparse} matrix, satisfying model-RIP-1 in \eqref{eqn:mrip2} for some degree $d > 0$ and $\epsilon_{\model_{\sparsity}} \in (0, 1/2)$. Assume an unknown vector $\x^\star \in \R^\dim$ is observed through $\obs = \Ph \x^\star$. Using criterion \eqref{eq:modelBP} to obtain a solution $\widehat{\x}$ such that $\obs = \Ph \widehat{\x}$, can we achieve a constant approximation of the form:}
\begin{align*}
\|\widehat{\x} - \x^\star\|_{2,1} \leq C_1 \cdot \|\x^\star - \x^\star_{\model_\sparsity^\star}\|_{2,1}, ~~C_1 > 0,
\end{align*} {\it where $\model_{\sparsity}^\star$ contains the groups of the best $\sparsity$-block sparse approximation of $\x^\star$?
}

\section{Main result}{\label{sec:results}}

In this section, we answer affirmatively to the above question with the following result:

\begin{theorem}{\label{sec_results:thm1}}
Let $\Ph \in \left\{ 0, 1\right\}^{\numsam \times \dim}$ be a $(k\cdot g, ~d, ~\epsilon_{\model_{\sparsity}})$ expander, satisfying model-RIP-1 \eqref{eqn:mrip2} for block sparse model $\model$, as in Definition \ref{def:block}. 
Consider two vectors $\x^\star, \widehat{\x} \in \R^{\dim}$ such that $\Ph \x^\star = \Ph \widehat{\x}$, where $\x^\star$ is unknown and $\widehat{\x}$ is the solution to \eqref{eq:modelBP}. Then, $\|\widehat{\x}\|_{2,1} \leq \|\x^\star\|_{2,1}$ and:
\begin{align*}
\|\x^\star - \widehat{\x}\|_{2,1} \leq \frac{2}{1 - \frac{4\epsilon_{\model_\sparsity}\cdot g}{1-2\epsilon_{\model_\sparsity}}} \cdot \|\x^\star - \x^\star_{\model_{\sparsity}^\star}\|_{2,1}.
\end{align*} 
\end{theorem}

In other words, Theorem \ref{sec_results:thm1} shows that, given a proper sparse matrix $\Ph$ satisfying the model-RIP-1 \eqref{eqn:mrip2} for structure $\model$, the $\ell_{2,1}$-convex program in \eqref{eq:modelBP} provides a good block-sparse solution $\widehat{\x} \in \R^{\dim}$. Next, we also state a corollary to the theorem above for the more realistic noisy case.

\begin{corollary}{\label{sec_results:cor1}}
Assume the setting in Theorem \ref{sec_results:thm1} and let two vectors $\x^\star, \widehat{\x} \in \R^{\dim}$ such that $\|\Ph (\x^\star - \widehat{\x})\|_1 = \gamma \geq 0$, where $\x^\star$ is unknown and $\widehat{\x}$ is the solution to \eqref{eq:modelBP}. Then, $\|\widehat{\x}\|_{2,1} \leq \|\x^\star\|_{2,1}$ and:
\begin{align*}
\|\x^\star - \widehat{\x}\|_{2,1} &\leq \frac{2}{1 - \frac{4\epsilon_{\model_\sparsity}\cdot g}{1-2\epsilon_{\model_\sparsity}}} \cdot \|\x^\star - \x^\star_{\model_{\sparsity}^\star}\|_{2,1} \nonumber + \frac{\gamma}{1 - \frac{4\epsilon_{\model_\sparsity}\cdot g}{1-2\epsilon_{\model_\sparsity}}}.
\end{align*}
\end{corollary}

\section{Proof of Theorem \ref{sec_results:thm1}}{\label{sec:proof}}

Let us consider two solutions to the set of linear equations $\Ph \x^\star = \obs = \Ph \widehat{\x}$, where $\x^\star, \widehat{\x} \in \R^\dim$ and let $\Ph \in \left\{ 0, 1\right\}^{\numsam \times \dim}$ be a given expander satisfying \eqref{eqn:mrip2}. We define $\z = \widehat{\x} - \x^\star \in \text{ker}(\Ph)$ since $\Ph \widehat{\x} = \Ph \x^\star \Longrightarrow \Ph\left(\widehat{\x} - \x^\star\right) = \mathbf{0}$. 

Using the group model formulation, each vector $\widehat{\x}, \x^\star$ can be decomposed as:
$\widehat{\x} = \sum_{\mathcal{G} \in \model} \widehat{\x}_{\mathcal{G}} ~~\text{ and }~~ \x^\star = \sum_{\mathcal{G} \in \model} \x^{\star}_{\mathcal{G}},$
where $\text{supp}(\widehat{\x}_{\mathcal{G}}) = \mathcal{G},~\text{supp}(\x^{\star}_{\mathcal{G}}) = \mathcal{G},$ $\forall \mathcal{G} \in \model$; moreover, assume that $\widehat{\x} \neq \x^\star$. Thus, given the above decompositions, we also have:
\begin{align}
\z = \widehat{\x} - \x^\star = \sum_{\mathcal{G} \in \model} \left(\widehat{\x}_{\mathcal{G}} - \x^{\star}_{\mathcal{G}} \right) =: \sum_{\mathcal{G} \in \model} \z_{\mathcal{G}},
\end{align} where $\text{supp}(\z_{\mathcal{G}}) = \mathcal{G}$. By assumption, we know that $\|\x^\star\|_{2,1} \geq \|\widehat{\x}\|_{2,1}$. Using the definition of the $\|\cdot\|_{2,1}$, we have:
\begin{align}
\sum_{\mathcal{G} \in \model} \|\x^{\star}_{\mathcal{G}}\|_2 \geq \sum_{\mathcal{G} \in \model} \|\widehat{\x}_{\mathcal{G}}\|_2 = \sum_{\mathcal{G} \in \model} \|\z_{\mathcal{G}} + \x^{\star}_{\mathcal{G}}\|_2
\end{align} Denote the set of $k$ groups in the best $\sparsity$-block sparse approximation of $\x^\star$ as $\model_{\sparsity}^\star$. Then, 
\begin{align}
\sum_{\mathcal{G} \in \model} \|\x^{\star}_{\mathcal{G}}\|_2 \nonumber &\geq \sum_{\mathcal{G} \in \model_{\sparsity}^\star} \|\z_{\mathcal{G}} + \x^{\star}_{\mathcal{G}}\|_2 + \sum_{\mathcal{G} \in \model \setminus \model_{\sparsity}^\star} \|\z_{\mathcal{G}} + \x^{\star}_{\mathcal{G}}\|_2 \nonumber \\
&\geq \sum_{\mathcal{G} \in \model_{\sparsity}^\star} \|\x^{\star}_{\mathcal{G}}\|_2 - \sum_{\mathcal{G} \in \model_{\sparsity}^\star} \|\z_{\mathcal{G}}\|_2 \nonumber \\ 
&+ \sum_{\mathcal{G} \in \model \setminus \model_{\sparsity}^\star} \|\z_{\mathcal{G}}\|_2 - \sum_{\mathcal{G} \in \model \setminus \model_{\sparsity}^\star} \|\x^{\star}_{\mathcal{G}}\|_2 \nonumber \\
&= \sum_{\mathcal{G} \in \model} \|\x^{\star}_{\mathcal{G}}\|_2 - 2\sum_{\mathcal{G} \in \model \setminus \model_{\sparsity}^\star} \|\x^{\star}_{\mathcal{G}}\|_2 \nonumber \\ 
&+\sum_{\mathcal{G} \in \model} \|\z_{\mathcal{G}}\|_2 - 2\sum_{\mathcal{G} \in \model_{\sparsity}^\star} \|\z_{\mathcal{G}}\|_2 \label{sec_proof:eq1}
\end{align} where the first equality is due to the block-sparse model and the first inequality is due to the triangle inequality. To proceed, we require the following Lemma; the proof is based on \cite{berinde2008combining}.

\begin{lemma}{\label{sec_proof:lem1}}
Let $\z \in \R^\dim$ such that $\Ph \z = \mathbf{0}$ and $\z = \sum_{\mathcal{G} \in \model} \z_{\mathcal{G}}$, where $\text{supp}(\z_\mathcal{G}) \subseteq \mathcal{G}, \forall \mathcal{G} \in \model$. Then,
\begin{align}
\sum_{\mathcal{G} \in \model_{\sparsity}^\star} \|\z_{\mathcal{G}}\|_2 \leq \frac{2\epsilon_{\model_\sparsity}\cdot g}{1-2\epsilon_{\model_\sparsity}} \cdot \sum_{\mathcal{G} \in \model} \|\z_{\mathcal{G}}\|_2, \label{sec_proof:eq2}
\end{align} where $\epsilon_{\model_\sparsity} \in (0, 1/2)$ denotes the model-based expansion parameter of expander matrix $\Ph \in \{0, 1 \}^{\numsam \times \dim}$. 
\end{lemma}

\begin{proof}
Following \cite{berinde2008combining}, we assume the following decomposition of $\dim$ indices, according to $\model$: Due to the non-overlapping nature of $\model$, we split coordinates in $\dim$ into $\sparsity$-block sparse index sets $\model_{\sparsity}^0 \equiv \model_{\sparsity}^\star,~ \model_{\sparsity}^1, ~\model_{\sparsity}^2, \dots, ~\model_{\sparsity}^t $ such that $(i)$ each $\model_{\sparsity}^l, \forall l$, has $\sparsity$ groups (except probably $\model_{\sparsity}^t$), $(ii)$ each group $\mathcal{G} \in \model_{\sparsity}^l$ has $g$ indices and $(iii)$ there is ordering on groups such that:
\begin{align*}
\|\z_{\mathcal{G}}\|_2 \geq \|\z_{\mathcal{G}'}\|_2, \quad \forall \mathcal{G} \in \model_{\sparsity}^l,~\forall \mathcal{G}' \in \model_{\sparsity}^q \text{~~s.t.~~} l \leq q.
\end{align*}

Since $\z \in \text{ker}(\Ph)$, we have $0 = \|\Ph \z\|_1$. Moreover, we denote as $\Gm(\model_{\sparsity}^\star)$ the set of indices of the rows of $\Ph$ that correspond to the neighbours of left-nodes in $\model_{\sparsity}^\star$; see Figure \ref{sec_proof:fig1}. 
Thus, without loss of generality, we reorder the rows of $\bPhi$ such that the top $|\cup_{\mathcal{G} \in \model_{\sparsity}^\star} \mathcal{G}|$ rows are indexed by $\Gm\left(\model_{\sparsity}^\star\right)$---the {\em union} of the set of neighbors of all $\mathcal{G} \in \model_{\sparsity}^\star$. Given the above partition, we denote the submatrix of $\bPhi$ composed of these rows as $\bPhi_{\Gm}$, such that 
\begin{align*}
\bPhi = \begin{pmatrix}
\rule[.5ex]{3.5em}{0.4pt} & \bPhi_{\Gm} & \rule[.5ex]{3.5em}{0.4pt}\\
& & \\
\rule[.5ex]{3.5em}{0.4pt} & \bPhi_{\Gm^c} & \rule[.5ex]{3.5em}{0.4pt}\\
\end{pmatrix}.
\end{align*} 
\begin{figure}[!tb]
	\begin{center}
		\includegraphics[width=0.3\columnwidth]{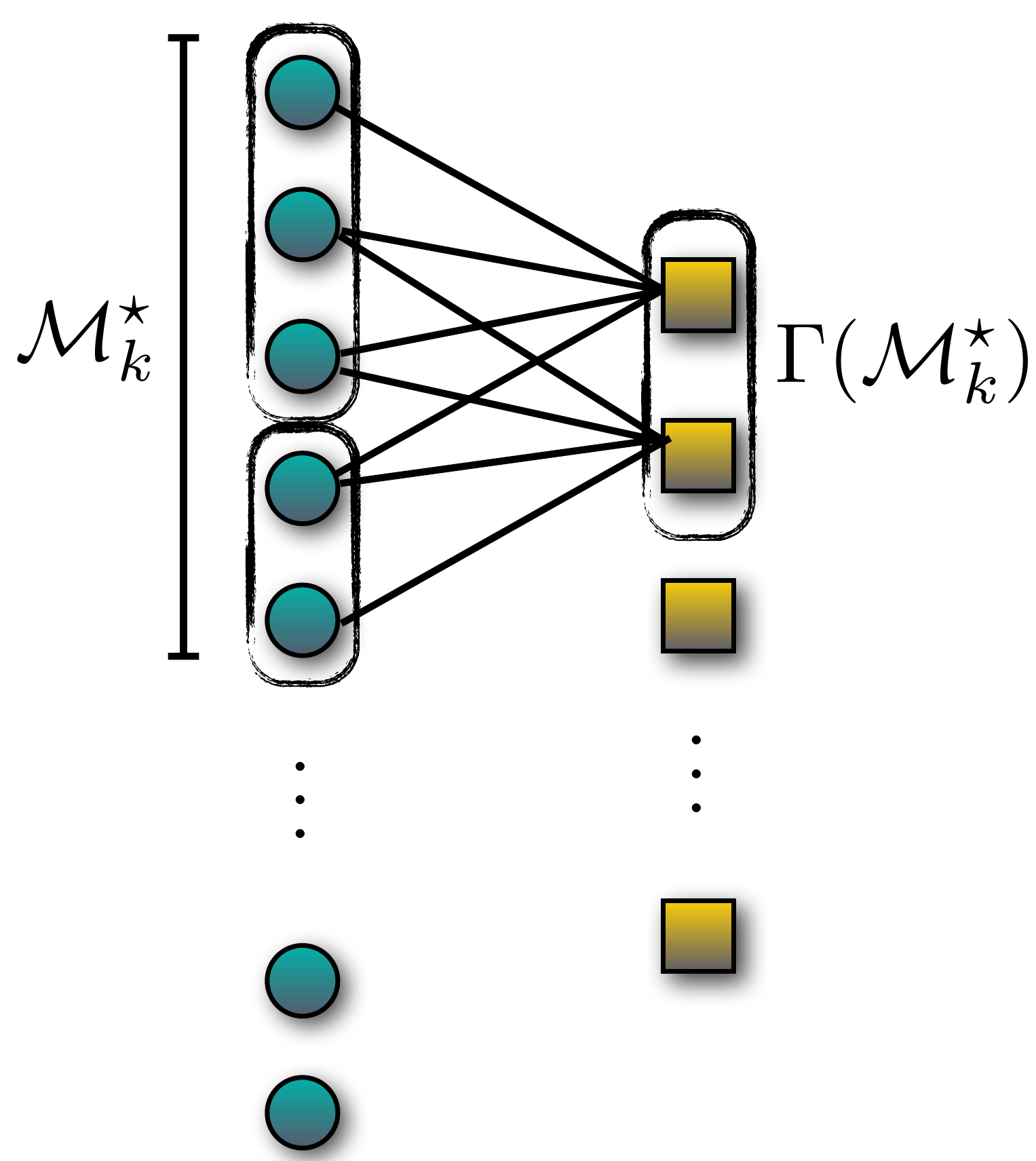}
	\end{center}
	\caption{Representation of neighbours of left-nodes, indexed by the $k$-group sparse set $\model_{\sparsity}^\star$. }{\label{sec_proof:fig1}}
\end{figure}
Thus, we have:
\begin{align*}
0 &= \|\Ph \z\|_1 = \|\Ph_{\Gm} \z\|_1 = \|\Ph_{\Gm} \cdot \sum_{\mathcal{G} \in \model} \z_{\mathcal{G}} \|_1 \nonumber \\
  &= \|\Ph_{\Gm} \cdot \left(\sum_{\mathcal{G} \in \model_{\sparsity}^\star}\z_{\mathcal{G}} + \sum_{\mathcal{G} \in \model \setminus \model_{\sparsity}^\star} \z_{\mathcal{G}}\right) \|_1 \nonumber \\
  &= \|\Ph_{\Gm} \cdot \left(\sum_{\mathcal{G} \in \model_{\sparsity}^0}\z_{\mathcal{G}} + \sum_{\mathcal{G} \in \model \setminus \model_{\sparsity}^0} \z_{\mathcal{G}}\right) \|_1 \nonumber \\
  &\geq \|\Ph_{\Gm} \cdot \sum_{\mathcal{G} \in \model_{\sparsity}^0}\z_{\mathcal{G}}\|_1 - \|\Ph_{\Gm}\sum_{\mathcal{G} \in \model \setminus \model_{\sparsity}^0} \z_{\mathcal{G}}\|_1 \nonumber  
\end{align*} By using the model-RIP-1 property \eqref{eqn:mrip2} of $\Ph$ on the input vector $\sum_{\mathcal{G} \in \model_{\sparsity}^0}\z_{\mathcal{G}}$ we have:
\begin{align*}
\|\Ph_{\Gm} \cdot \sum_{\mathcal{G} \in \model_{\sparsity}^0}\z_{\mathcal{G}}\|_1 &= \|\Ph \cdot \sum_{\mathcal{G} \in \model_{\sparsity}^0}\z_{\mathcal{G}}\|_1 \nonumber \\ &\geq (1-2\epsilon_{\model_\sparsity})\cdot d \cdot \|\sum_{\mathcal{G} \in \model_{\sparsity}^0}\z_{\mathcal{G}}\|_1  \\
&= (1-2\epsilon_{\model_\sparsity})\cdot d \cdot \sum_{\mathcal{G} \in \model_{\sparsity}^0}\|\z_{\mathcal{G}}\|_1 \\
&\geq (1-2\epsilon_{\model_\sparsity})\cdot d \cdot \sum_{\mathcal{G} \in \model_{\sparsity}^0}\|\z_{\mathcal{G}}\|_2
\end{align*} where the first equality is due to $\Ph_{\Gm^c} \cdot \sum_{\mathcal{G} \in \model_{\sparsity}^0} \z_{\mathcal{G}} = \mathbf{0}$, the second equality due to the non-overlapping groups and the last inequality since $\|\x\|_1 \geq \|\x\|_2,~\forall \x \in \R^\dim$. 

Therefore, we have:
\begin{small}
\begin{align*}
0 &\geq (1-2\epsilon_{\model_\sparsity})\cdot d \cdot \sum_{\mathcal{G} \in \model_{\sparsity}^0}\|\z_{\mathcal{G}}\|_2 - \|\Ph_{\Gm}\sum_{\mathcal{G} \in \model \setminus \model_{\sparsity}^0} \z_{\mathcal{G}}\|_1 \nonumber \\
  &\geq (1-2\epsilon_{\model_\sparsity})\cdot d \cdot \sum_{\mathcal{G} \in \model_{\sparsity}^0}\|\z_{\mathcal{G}}\|_2 - \sum_{\mathcal{G} \in \model \setminus \model_{\sparsity}^0} \|\Ph_{\Gm}\z_{\mathcal{G}}\|_1 \nonumber \\
  &= (1-2\epsilon_{\model_\sparsity})\cdot d \cdot \sum_{\mathcal{G} \in \model_{\sparsity}^0}\|\z_{\mathcal{G}}\|_2 - \sum_{l \geq 1} \sum_{\mathcal{G} \in \model_{\sparsity}^l} \sum_{\substack{(i, j) \in \mathcal{E} \\ i \in \mathcal{G} \\ j \in \Gamma}} |(\z_{\mathcal{G}})_i| \\
  &\geq (1-2\epsilon_{\model_\sparsity})\cdot d \cdot \sum_{\mathcal{G} \in \model_{\sparsity}^0}\|\z_{\mathcal{G}}\|_2 - \sum_{l \geq 1} |\text{\# edges}~(i,j)~:~ i\in \mathcal{G}, j \in \Gm, \mathcal{G} \in \model_{\sparsity}^l| 
  \cdot\max_{\mathcal{G} \in \model_{\sparsity}^l} \max_{i \in \mathcal{G}} |(\z_{\mathcal{G}})_i|
\end{align*}
\end{small} 
The quantity $|\text{\# edges}~(i,j)~:~ i\in \mathcal{G}, j \in \Gm, \mathcal{G} \in \model_{\sparsity}^l|$ further satisfies:
\begin{align*}
|\text{\# edges}~(i,j)~:~ i\in \mathcal{G}, j \in \Gm, \mathcal{G} \in \model_{\sparsity}^l| \nonumber
&= |\Gm(\model_{\sparsity}^0) \cap \Gm(\model_{\sparsity}^l)| \nonumber \\
&\stackrel{(i)}{=} |\Gm(\model_{\sparsity}^0)| + |\Gm(\model_{\sparsity}^l)| - |\Gm(\model_{\sparsity}^0) \cup \Gm(\model_{\sparsity}^l)| \\
&\leq d\cdot|\model_{\sparsity}^0| + d\cdot|\model_{\sparsity}^l| - |\Gm(\model_{\sparsity}^0) \cup \Gm(\model_{\sparsity}^l)| \\
&\leq 2d\cdot\sparsity\cdot g - |\Gm(\model_{\sparsity}^0) \cup \Gm(\model_{\sparsity}^l)| \\
&\stackrel{(ii)}{\leq} 2d\cdot\sparsity\cdot g - d(1 - \epsilon_{\model_\sparsity})|\model_{\sparsity}^0 \cup \model_{\sparsity}^l| \\
&\stackrel{(iii)}{\leq} 2d\cdot\sparsity\cdot \epsilon_{\model_\sparsity} \cdot g
\end{align*} where $(i)$ is due to the inclusion-exclusion principle, $(ii)$ is due to the expansion property and, $(iii)$ is due to $|\mathcal{G}| = g,~\forall \mathcal{G} \in \model$.

Thus, the above inequality becomes:
\begin{small}
\begin{align*}
0  &\geq (1-2\epsilon_{\model_\sparsity})\cdot d \cdot \sum_{\mathcal{G} \in \model_{\sparsity}^0}\|\z_{\mathcal{G}}\|_2 - 2d\cdot\sparsity\cdot \epsilon_{\model_\sparsity}\cdot g \sum_{l \geq 1} \max_{\mathcal{G} \in \model_{\sparsity}^l} \max_{i \in \mathcal{G}} |(\z_{\mathcal{G}})_i| 
\end{align*}
\end{small} Let $\mathcal{G}_{\max}^l$ denote the group that contains the maximizing index $i_{\max}$ such that: 
\begin{align*}
i_{\max} \in \argmax_{i \in \mathcal{G} \text{~s.t.~} \mathcal{G} \in \model_{\sparsity}^l} |(\z_{\mathcal{G}})_i|.
\end{align*} Then:
\begin{small}
\begin{align*}
0 &\geq (1-2\epsilon_{\model_\sparsity})\cdot d \cdot \sum_{\mathcal{G} \in \model_{\sparsity}^0}\|\z_{\mathcal{G}}\|_2 - 2d\cdot\sparsity\cdot \epsilon_{\model_\sparsity}\cdot g \sum_{l \geq 1} \|\z_{\mathcal{G}_{\max}^l}\|_{\infty}
\end{align*}
\end{small} However, due to the ordering of $\ell_2$-norm of groups:
\begin{small}
\begin{align*}
\|\z_{\mathcal{G}_{\max}^l}\|_{\infty} \leq \|\z_{\mathcal{G}_{\max}^l}\|_{2} \leq \min_{\mathcal{G} \in \model_{\sparsity}^{l-1}} \|\z_{\mathcal{G}}\|_{2} \leq \frac{1}{\sparsity} \sum_{\mathcal{G} \in \model_{\sparsity}^{l-1}} \|\z_{\mathcal{G}}\|_{2},
\end{align*}
\end{small} we further have:
\begin{small}
\begin{align*}
0 &\geq (1-2\epsilon_{\model_\sparsity})\cdot d \cdot \sum_{\mathcal{G} \in \model_{\sparsity}^0}\|\z_{\mathcal{G}}\|_2 - 2d\cdot\sparsity\cdot \epsilon_{\model_\sparsity}\cdot g \sum_{l \geq 1} \frac{1}{\sparsity} \sum_{\mathcal{G} \in \model_{\sparsity}^{l-1}} \|\z_{\mathcal{G}}\|_{2} \\
  &\geq (1-2\epsilon_{\model_\sparsity})\cdot d \cdot \sum_{\mathcal{G} \in \model_{\sparsity}^0}\|\z_{\mathcal{G}}\|_2 - 2d\cdot \epsilon_{\model_\sparsity} \cdot g \sum_{\mathcal{G} \in \model} \|\z_{\mathcal{G}}\|_{2}
\end{align*}
\end{small} which leads to:
\begin{align*}
\sum_{\mathcal{G} \in \model_{\sparsity}^0} \|\z_{\mathcal{G}}\|_2 \leq \frac{2 \epsilon_{\model_\sparsity} \cdot g}{1-2\epsilon_{\model_\sparsity}} \cdot \sum_{\mathcal{G} \in \model} \|\z_{\mathcal{G}}\|_2 \end{align*}
\end{proof}
Using \eqref{sec_proof:eq2} in \eqref{sec_proof:eq1}, we further have:
\begin{small}
\begin{align*}
0 &\geq - 2\sum_{\mathcal{G} \in \model \setminus \model_{\sparsity}^0} \|\x^{\star}_{\mathcal{G}}\|_2 + \sum_{\mathcal{G} \in \model} \|\z_{\mathcal{G}}\|_2 \nonumber - \frac{4 \epsilon_{\model_\sparsity} \cdot g}{1-2\epsilon_{\model_\sparsity}} \cdot \sum_{\mathcal{G} \in \model} \|\z_{\mathcal{G}}\|_2 \Longrightarrow \\
\sum_{\mathcal{G} \in \model} \|\z_{\mathcal{G}}\|_2 &\leq \frac{2}{1 - \frac{4 \epsilon_{\model_\sparsity} \cdot g}{1-2\epsilon_{\model_\sparsity}}} \sum_{\mathcal{G} \in \model \setminus \model_{\sparsity}^0} \|\x^{\star}_{\mathcal{G}}\|_2,
\end{align*}
\end{small} which is the desired result. 

\section{Comparison to state-of-the-art}

To justify our theoretical results, here we compare with the only known results on \emph{convex sparse recovery using expander matrices} of \cite{berinde2008combining}. We note that, in \cite{berinde2008combining}, no apriori knowledge is assumed, beyond plain sparsity. 

We start our discussion with the following remark. 

\begin{remark}
In the extreme case of $g = 1$, Lemma \ref{sec_proof:lem1} is analogous to Lemma 16 of \cite{berinde2008combining}. To see this, observe that $\|\z_{\mathcal{G}}\|_2 = |(\z)_{i}|$ for $g = 1$, where $|\mathcal{G}| = 1$ and $\mathcal{G} \equiv i \in \model_{\sparsity}^0$ or $\in \model$. 
\end{remark}

We highlight that, as $g$ grows, feasible values of $\epsilon_{\model_\sparsity} \rightarrow 0$, \textit{i.e.}, 
we require more rows to construct an expander matrix $\Ph$ with such expansion property. 
For simplicity, in the discussion below we use $\epsilon \equiv \epsilon_{\model_\sparsity}$ interchangeably, where the type of $\epsilon$ used is apparent from context.

In the case where we are \emph{oblivious to any, a-priori known, structured sparsity model}, \cite{berinde2008combining} prove the following error guarantees for the vanilla BP formulation \eqref{eq:BP}, using expander matrices:
\begin{align}
\|\x^\star - \widehat{\x}\|_{1} \leq \frac{2}{1 - 4\cdot \frac{\epsilon}{1-2\epsilon}} \cdot \|\x^\star - \x^\star_{\mathcal{S}}\|_{1}, \label{eq:l1_error}
\end{align} where $\mathcal{S} \subseteq [\dim]$ such that $|\mathcal{S}| = |\cup_{\mathcal{G} \in \model_{\sparsity}^\star} \cup_{i \in \mathcal{G}} i|$; \textit{i.e.}, we are looking for a solution of the same sparsity as the union of groups in the block-sparse case. In order to compare \eqref{eq:l1_error} with our result, a naive transformation of \eqref{eq:l1_error} into $\ell_{2,1}$ terms leads to:
\begin{align}
\|\x^\star - \widehat{\x}\|_{2,1} \leq \frac{2\sqrt{g}}{1 - 4\cdot \frac{\epsilon}{1-2\epsilon}} \cdot \|\x^\star - \x^\star_{\model_{\sparsity}^\star}\|_{2,1}. \label{eq:l21_error2}
\end{align} 


To define the conditions under which \eqref{eq:l21_error2} is valid, we require:
\begin{align*}
1 - 4\cdot \frac{\epsilon}{1 - 2\epsilon} > 0 \Rightarrow \epsilon < 1/6,
\end{align*} \textit{i.e.}, $\epsilon$ is independent of $g$, while in our case, we have:
\begin{align*}
1 - 4\cdot \frac{\epsilon \cdot g}{1 - 2\epsilon} > 0 \Rightarrow \epsilon < \frac{1}{2(1 + 2g)} ~~\left(\stackrel{g = 1}{=} 1/6\right);
\end{align*} \textit{i.e.}, our analysis provides weaker bounds with respect to the range of $\epsilon$ values such that the error guarantees is meaningful; see also Figure \ref{sec_comp:fig1}.

\begin{figure}[!tb]
	\begin{center}
		\includegraphics[width=0.4\columnwidth]{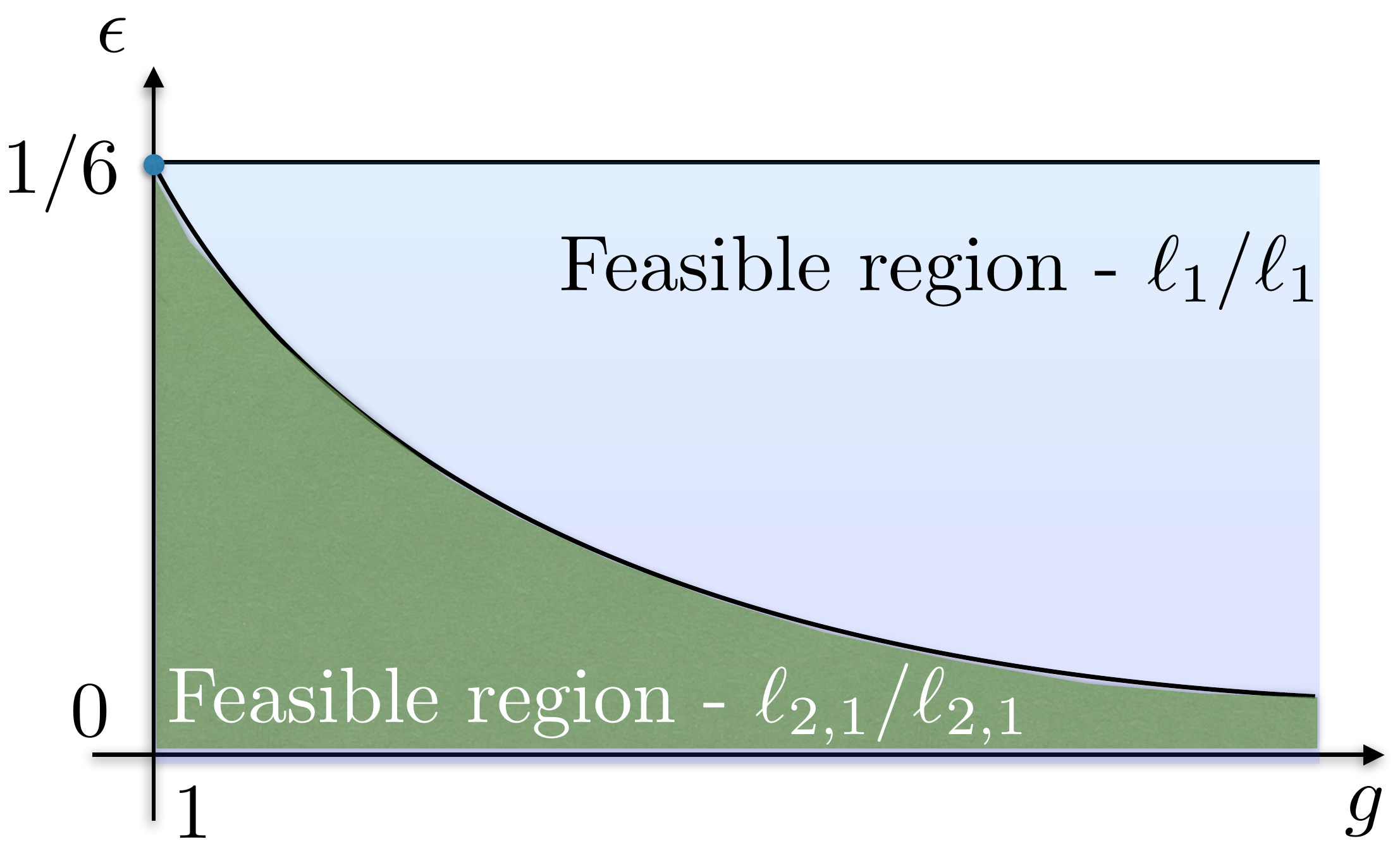}
	\end{center}
	\caption{Feasibility regions for our approach and \cite{berinde2008combining} as a function of $\epsilon$ and $g$ variables.}{\label{sec_comp:fig1}}
\end{figure}

However, as already mentioned, \eqref{eq:l21_error2} is \emph{oblivious} to any model $\model$: \emph{the solution $\widehat{\x}$ does not necessarily belong to $\model$}. This fact provides degrees of freedom to obtain better approximation constants. Nevertheless, this does not guarantee $\widehat{\x} \in \model$, considering only simple sparsity. Section \ref{sec:experiments} includes examples that highlight the superiority of $\ell_{2,1}$-norm in optimization.

%
%
\section{Experiments}{\label{sec:experiments}}

We start our discussion $(i)$ with a comparison between $\ell_1$- and $\ell_{2,1}$-norm convex approaches, when the ground truth is known to be block sparse, and $(ii)$ with a comparison between dense sub-Gaussian and sparse sensing matrices in \eqref{eq:modelBP}, both w.r.t. sampling complexity and computational complexity requirements. We conclude with the task of recovering 2D images from compressed measurements, using block sparsity. 

\paragraph{Solver.} To solve both $\ell_1$- and $\ell_{2,1}$-norm instances \eqref{eq:BP} and \eqref{eq:modelBP}, we use the primal-dual convex optimization framework in \cite{tran2014constrained}, that solves \eqref{eq:BP} and \eqref{eq:modelBP} -- among other objectives, by using ideas from the alternating direction methods of multipliers.
Using the same framework for solution, we obtain more accurate and credible comparison results between different problem settings.

\begin{figure}[!ht]
	\begin{center}
		\includegraphics[width=0.5\columnwidth]{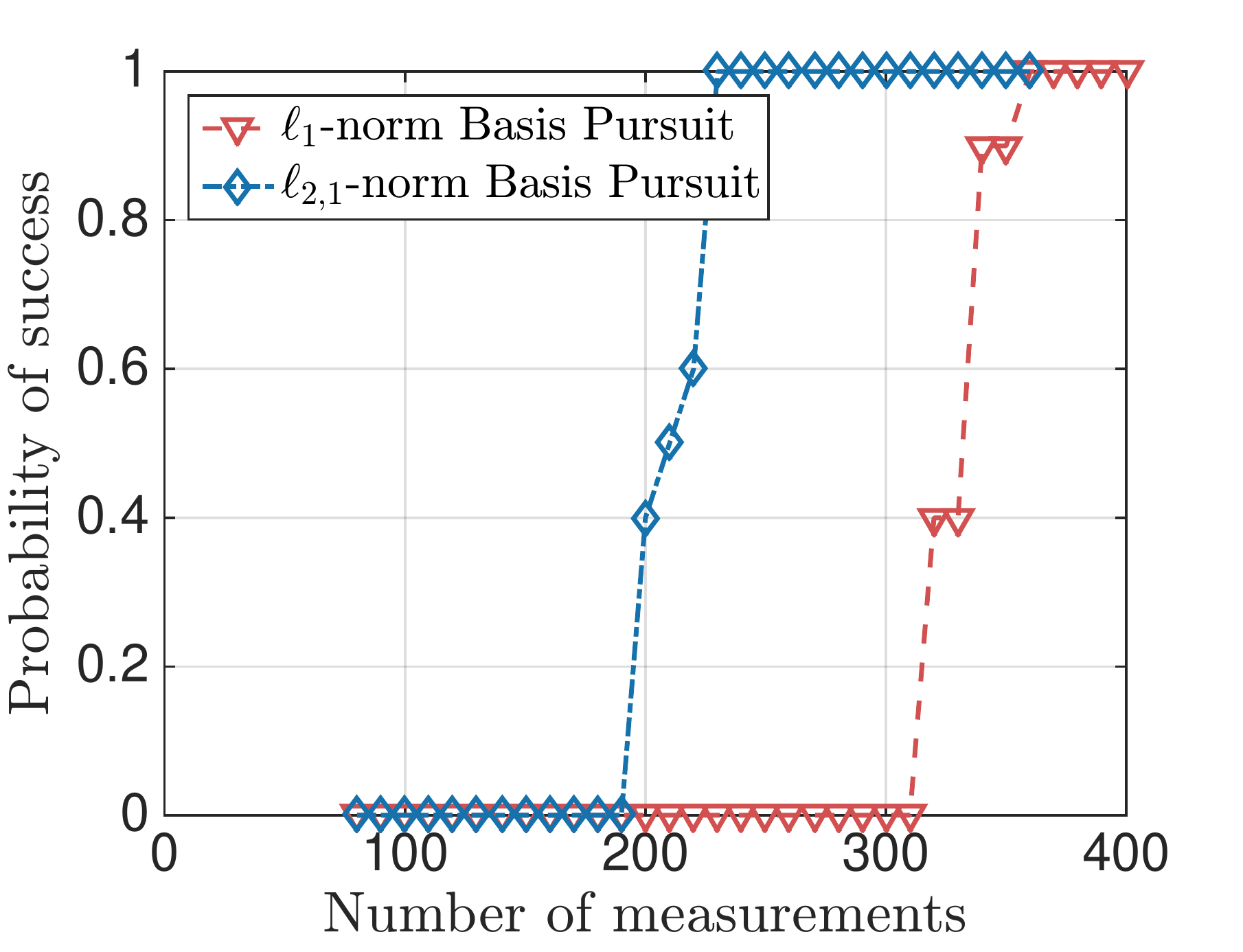} \hspace{-0.5cm}
		\includegraphics[width=0.5\columnwidth]{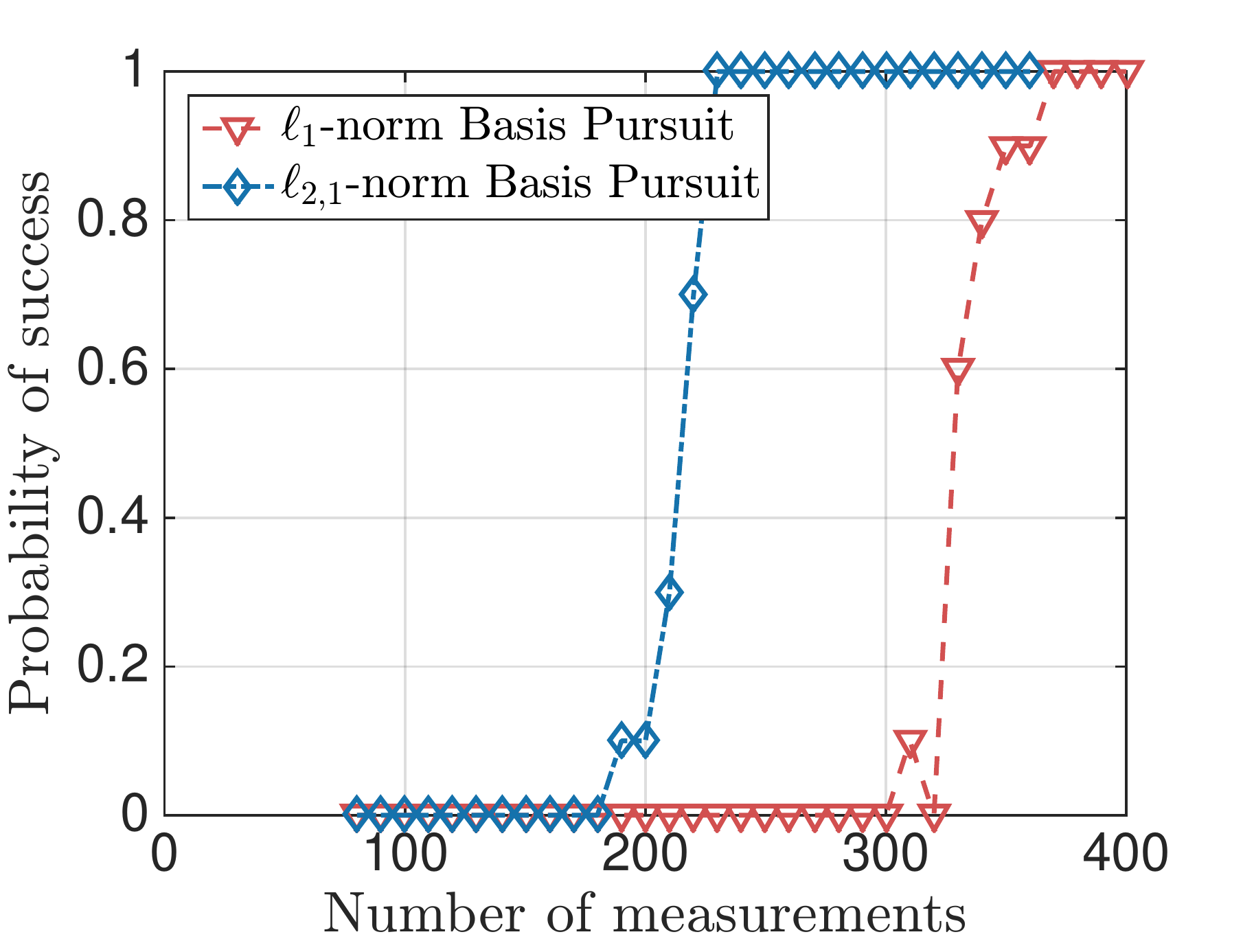}
	\end{center}
	\caption{Average probability of successful recovery, as a function of the total number of observations, over $10$ Monte Carlo iterations. \textit{Top panel}: $\bestx$ is generated from a normal distribution. \textit{Bottom panel}: $\bestx$ is generated from a Bernoulli distribution. $\ell_{2,1}$-norm solver requires much less number of measurements for successful recovery, compared to $\ell_{1}$-norm solver.}{\label{fig:000}}
\end{figure}

\paragraph{$\ell_1$- vs. $\ell_{2,1}$-norm.} In this experiment, we verify that \emph{a priori structure knowledge helps in recovery.} To show this in the convex domain, consider the following artificial example. Let $\obs = \bPhi \bestx$ be the set of observations, where $\bestx \in \R^{10^3}$ is a $k$-block sparse vector, for $k = 8$. Here, we assume a block-sparse model $\model$ with $M = 100$ non-overlapping groups. Observe that $g = \dim/M = 10$.

Both for $\ell_1$- and $\ell_{2,1}$-norm cases, $\bPhi \in \R^{\numsam \times \dim}$ is designed as an expander, with degree $d = \left \lceil 22 \cdot \sfrac{\log \left(M\right)}{g}\right \rceil = 11$, complying with our theory. Further, we make the convention $\bPhi := \sfrac{1}{d} \cdot \bPhi$.

We consider two cases: $(i)$ $\bestx$ is generated from a Gaussian distribution and $(ii)$ $\bestx$ is generated as a binary signal $\in \left\{\pm 1\right\}$. In both cases, $\bestx$ is normalized such that $\|\bestx\|_2 = 1$. For all cases, we used the same algorithmic framework \cite{tran2014constrained} and performed $10$ Monte Carlo realizations. Figure \ref{fig:000} shows the average probability for \emph{successful recovery} as a function of the total number of measurements observed; we declare a success when $\|\widehat{\x} - \bestx\|_2 \leq 10^{-5}$. It is apparent that knowing the model a priori can guarantee recovery with less number of observations.

\paragraph{Sub-Gaussian vs. sparse sensing matrices in $\ell_{2,1}$-norm recovery.} Let us focus now on the $\ell_{2,1}$-norm case. For this setting, we perform two experiments. First, we consider a similar setting as previously; the only difference lies in the way we generate the sensing matrix $\X$. We consider two cases: $(i)$ $\X \sim \mathcal{N}\left(\mathbf{0}, \sfrac{1}{\numsam} \cdot \mathbf{I}\right)$ and $(ii)$ $\X$ is a sparse expander matrix, again with degree $d = \left \lceil 22 \cdot \sfrac{\log \left(M\right)}{g}\right \rceil = 11$. Here, we only use the $\ell_{2,1}$-norm solver. Figure \ref{fig:001} depicts some results. We observe that expander matrices perform worse -- as a function of the total number of measurements required for successful recovery. Nevertheless, using sparse matrices is still considerably competitive to dense Gaussian matrices. 

\begin{figure}[!hb]
	\begin{center}
		\includegraphics[width=0.5\columnwidth]{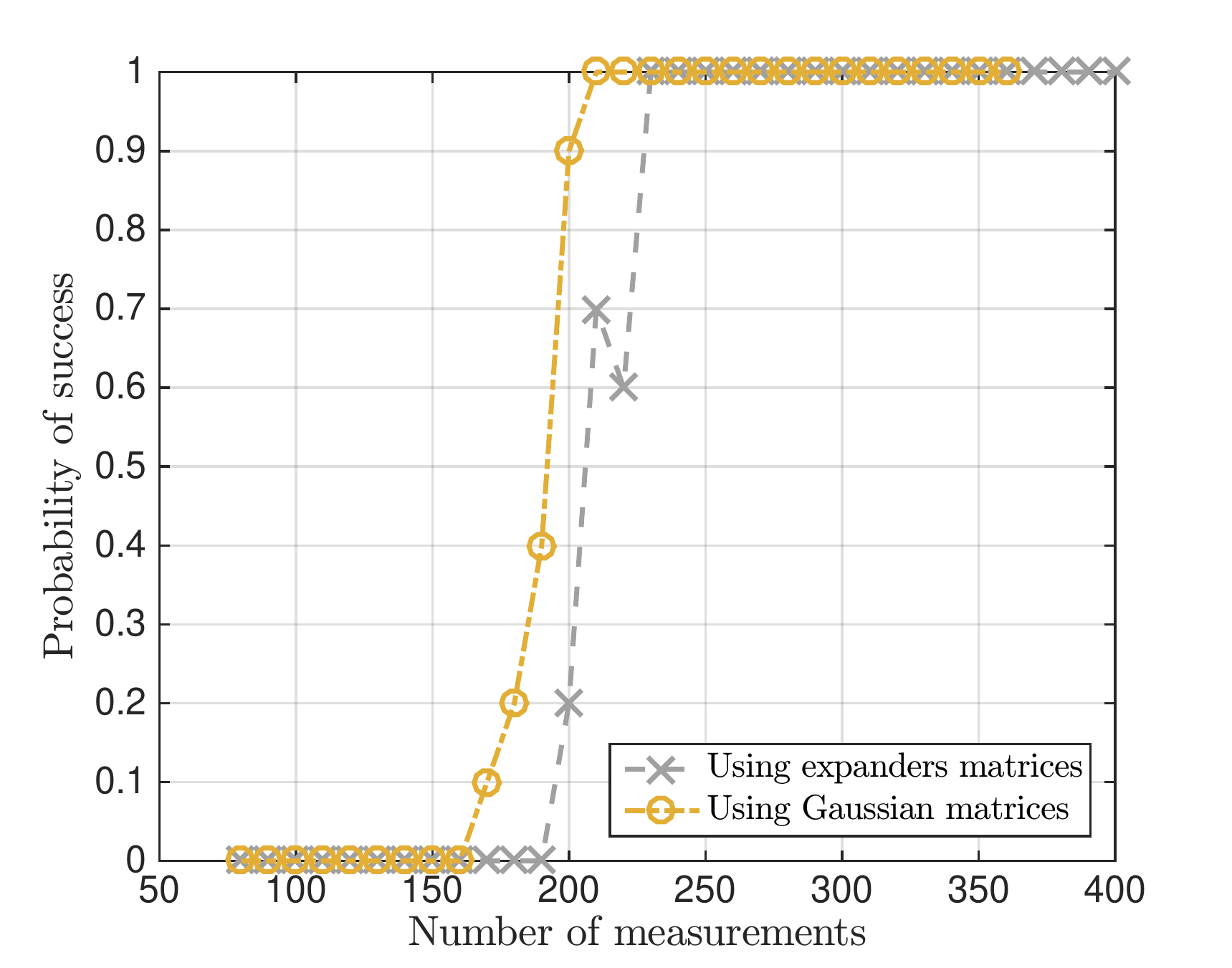}
	\end{center} 
	\caption{Average probability of successful recovery, as a function of the total number of observations, over $10$ Monte Carlo iterations. Expander matrices perform worse but competitive to dense Gaussian matrices.}{\label{fig:001}}
\end{figure}


\begin{table}[!h]
\centering
\ra{1.3}
\begin{scriptsize}
\begin{tabular}{ll c cc c cc} \toprule
\multicolumn{2}{c}{Model} & \phantom{a} & \multicolumn{2}{c}{$\|\widehat{\x} - \bestx\|_2$} &  \phantom{a}  & \multicolumn{2}{c}{Time (sec)} \\
\cmidrule{1-2} \cmidrule{4-5} \cmidrule{7-8} 
\multicolumn{1}{c}{$\dim$} & \multicolumn{1}{c}{$k\cdot g$} & \phantom{a} & Gaus. & Exp. & \phantom{a}  & Gaus. & Exp. \\ \midrule
\multirow{3}{*}{$10000$} & $300$ &  & 8.6e-07  & 3.3e-06 &  & 24.3  & \textcolor[rgb]{0.4,0.1,0}{\textbf{5.8}} \\
 & $400$ &  & 8.2e-06  & 3.4e-06 &  & 27.5  & \textcolor[rgb]{0.4,0.1,0}{\textbf{6.2}} \\
 & $500$ &  & 8.6e-06  & 3.2e-06 &  & 27.8  & \textcolor[rgb]{0.4,0.1,0}{\textbf{6.0}} \\
 & $600$ &  & 8.6e-06  & 3.4e-06 &  & 31.2  & \textcolor[rgb]{0.4,0.1,0}{\textbf{7.6}} \\
 \midrule 
\multirow{3}{*}{$20000$} & $300$ &  & 8.1e-07  & 3.4e-06 &  & 95.5  & \textcolor[rgb]{0.4,0.1,0}{\textbf{18.5}} \\
 & $400$ &  & 8.1e-06  & 3.3e-06 &  & 79.4  & \textcolor[rgb]{0.4,0.1,0}{\textbf{16.4}} \\
 & $500$ &  & 8.5e-06  & 3.4e-06 &  & 83.9  & \textcolor[rgb]{0.4,0.1,0}{\textbf{15.8}} \\
 & $600$ &  & 8.5e-06  & 3.5e-06 &  & 91.3  & \textcolor[rgb]{0.4,0.1,0}{\textbf{18.9}} \\
 \midrule 
\multirow{3}{*}{$50000$} & $300$ &  & 8.2e-06  & 3.3e-06 &  & 419.3  & \textcolor[rgb]{0.4,0.1,0}{\textbf{49.6}} \\
 & $400$ &  & 8.1e-06  & 3.4e-06 &  & 432.8  & \textcolor[rgb]{0.4,0.1,0}{\textbf{45.8}} \\
 & $500$ &  & 8.5e-06  & 3.6e-06 &  & 436.0  & \textcolor[rgb]{0.4,0.1,0}{\textbf{52.9}} \\
 & $600$ &  & 8.4e-06  & 3.5e-06 &  & 435.4  & \textcolor[rgb]{0.4,0.1,0}{\textbf{51.1}} \\
 \midrule 
\multirow{3}{*}{$100000$} & $600$ &  & 8.1e-06  & 9.4e-06 &  & 1585.5  & \textcolor[rgb]{0.4,0.1,0}{\textbf{55.1}} \\
 & $800$ &  & 8.1e-06  & 9.5e-06 &  & 1598.2  & \textcolor[rgb]{0.4,0.1,0}{\textbf{54.5}} \\
 & $1000$ &  & 8.4e-06  & 9.4e-06 &  & 1600.6  & \textcolor[rgb]{0.4,0.1,0}{\textbf{56.2}} \\
 & $1200$ &  & 8.1e-06  & 9.3e-06 &  & 1648.0  & \textcolor[rgb]{0.4,0.1,0}{\textbf{55.6}} \\
\bottomrule
\end{tabular}
\end{scriptsize}
\caption{Summary of comparison results for reconstruction and efficiency. Median values are reported. As a stopping criterion, we used $\sfrac{\|\x_{i+1} - \x_{i}\|_2}{\|\x_{i+1}\|_2} \leq 10^{-6}$, where $\x_i$ is the estimate at the $i$-th iteration. In all cases, $n = \lceil 0.4 \cdot p \rceil$.} \label{tbl:CompSVP}
\end{table}

Now, to grasp the full picture in this setting, we scale our experiments to higher dimensions. Table \ref{tbl:CompSVP} summons up the results. All experiments were repeated for 10 independent realizations and the table contains the median scores. For $\dim = \left\{10^4, 2\cdot 10^4, 5\cdot 10^4, 10^5 \right\}$, the total number of non-overlapping groups in $\model$ was $M = \left\{10^3, 2\cdot 10^3, 5\cdot 10^3, 5\cdot 10^3 \right\}$, respectively. The column cardinality is selected as $d = \left \lceil 22 \cdot \sfrac{\log \left(M\right)}{g}\right \rceil$. For each $\dim$, the block sparsity ranges as $k \in \left\{3, \dots, 6\right\}$. In all cases, $n = \lceil 0.4 \cdot p \rceil$.

One can observe that using sparse matrices in this setting results into \emph{faster convergence} -- as in total time required for stopping criterion to be met. Meanwhile, the solution quality is at least at the same levels, compared to that when dense Gaussian matrices are used. 

Finally, we highlight that, for $\dim = 10^5$, since $M$ does not increase, the number of non-zeros $d$ per column decreases (\textit{i.e.}, $d = 7$ while $d > 11$ in all other cases). This results into a small deterioration of the recovery quality; though, still comparable to the convex counterpart.\footnote{Observe that in most configurations, expander matrices find a solution closer to $\bestx$, compared to the dense setting, except for the case of $\dim = 10^5$, where we decrease the number of zeros per column.} Meanwhile, the time required for convergence in the expander case remains at the same levels as when $\dim = 5 \cdot 10^4$; however, the same does not apply for the dense counterpart. This constitutes the use of expander matrices appealing in real applications.

\paragraph{Block sparsity in image processing.} For this experiment, we use the real background-subtracted image dataset, presented in \cite{huang2011learning}. Out of 205 frames, we randomly selected 100 frames to process. Each frame is rescaled to be of size $2^8 \times 2^8$ pixels. Each pixel takes values in $[0, 1]$. 
We observe $\obs = \bPhi \bestx$ where $\bestx \in [0,1]^{\dim}, ~\dim = 2^{16}, $ is the ground-truth vectorized frame and, $\bPhi$ is either sparse or dense sensing matrix, designed as in the previous experiments. 

\begin{figure*}
\begin{minipage}{0.75\paperwidth}
	\begin{minipage}{0.105\linewidth}
	\centering \scriptsize{Original}
	\end{minipage}
	\begin{minipage}{0.105\linewidth}
	\centering \scriptsize{$\ell_1$ + Exp.}
	\end{minipage}
	\begin{minipage}{0.105\linewidth}
	\centering \scriptsize{$\ell_1$ + Gaus.}
	\end{minipage}
	\begin{minipage}{0.105\linewidth}
	\centering \scriptsize{$\ell_{2,1}$ + Exp.}
	\end{minipage}
	\begin{minipage}{0.105\linewidth}
	\centering \scriptsize{$\ell_{2,1}$ + Gaus.}
	\end{minipage}\\
	\includegraphics[width = 0.105\linewidth]{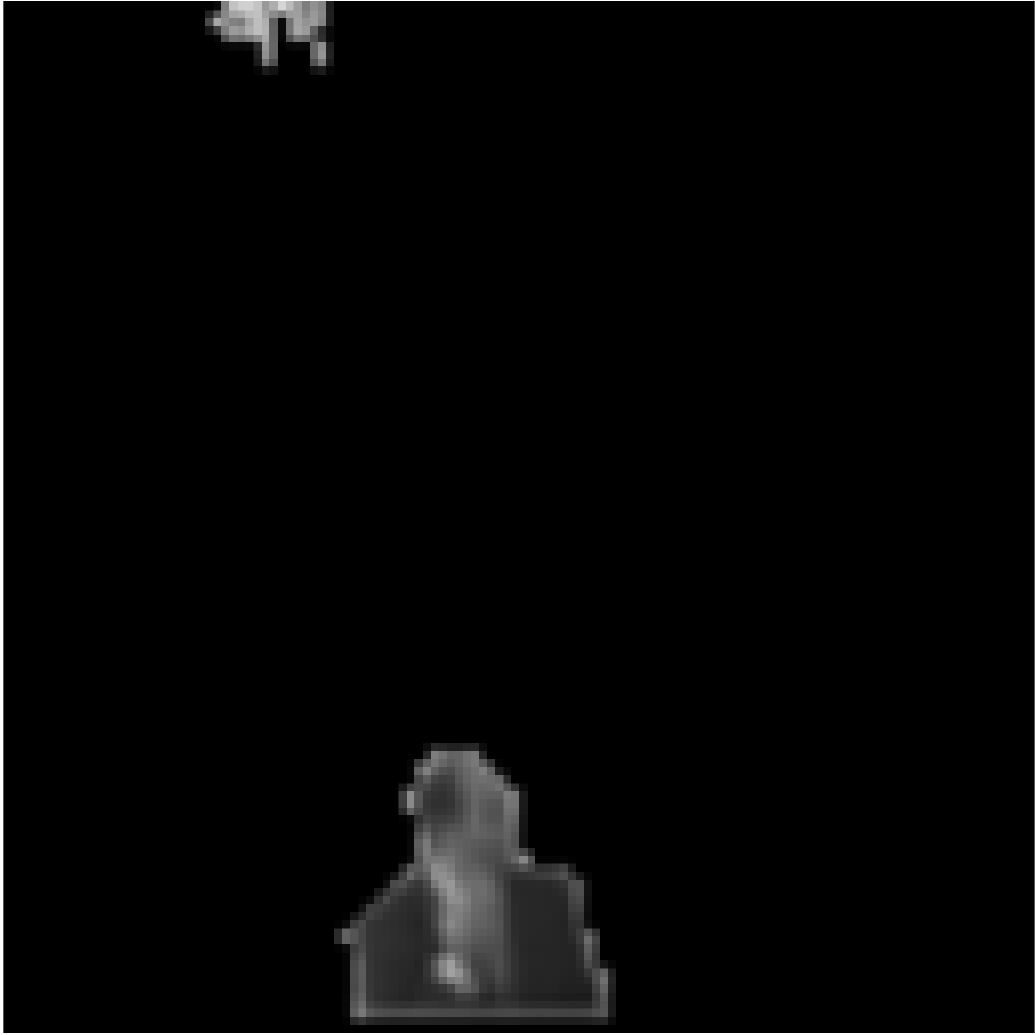}
	\includegraphics[width = 0.105\linewidth]{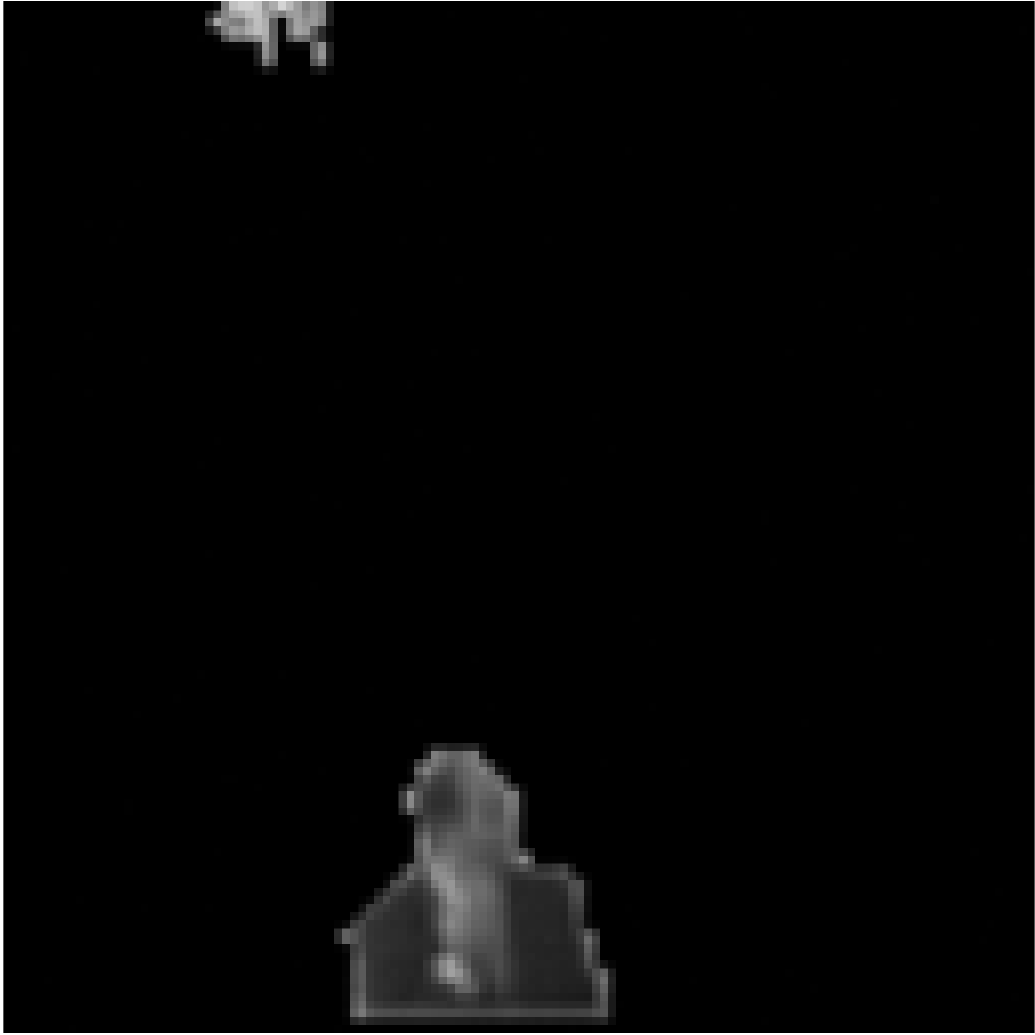}
	\includegraphics[width = 0.105\linewidth]{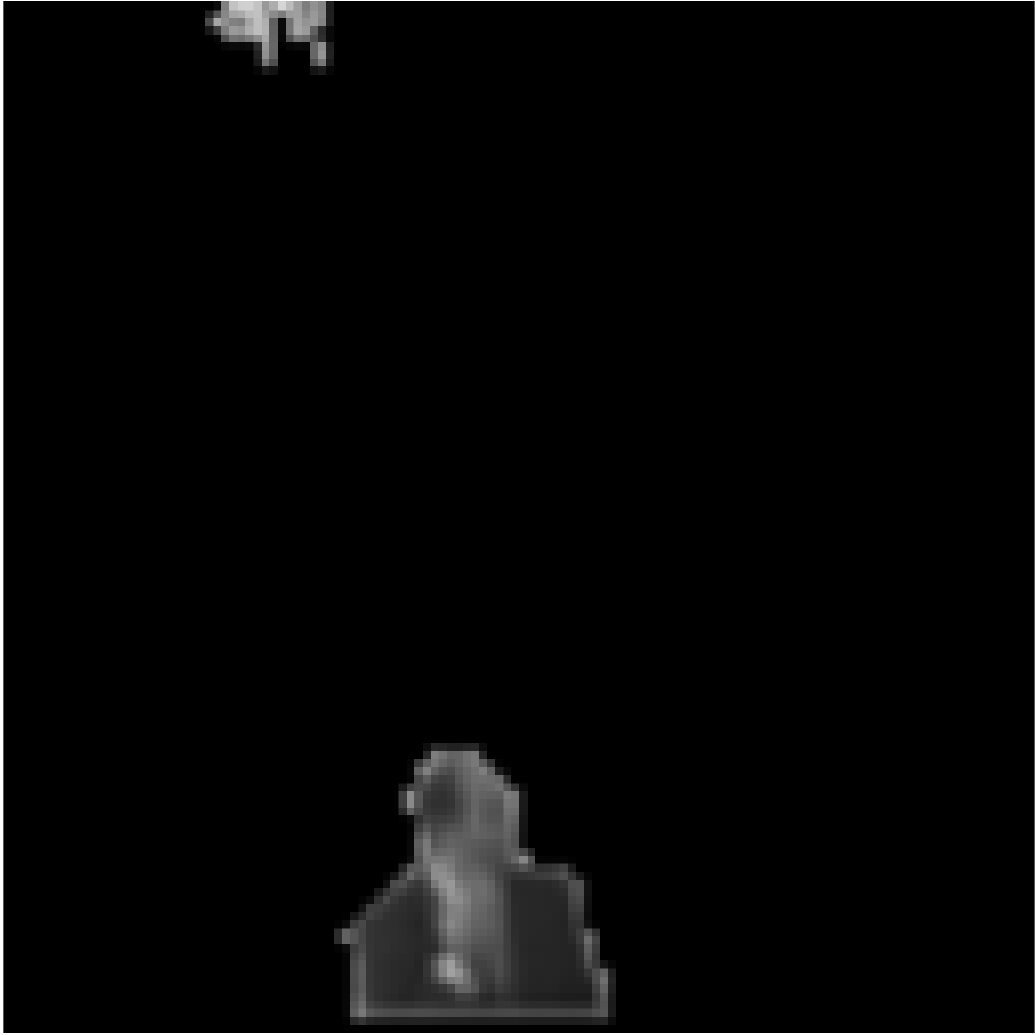}
	\includegraphics[width = 0.105\linewidth]{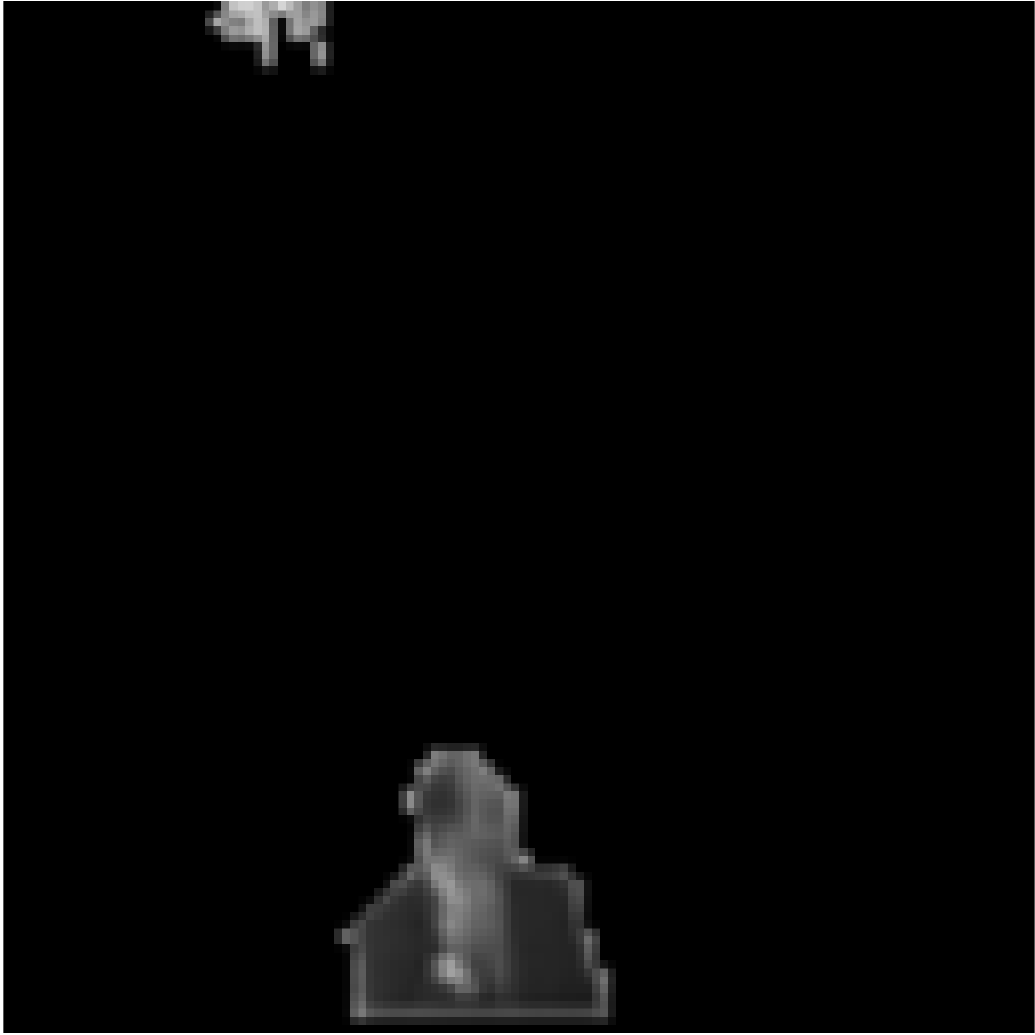}
	\includegraphics[width = 0.105\linewidth]{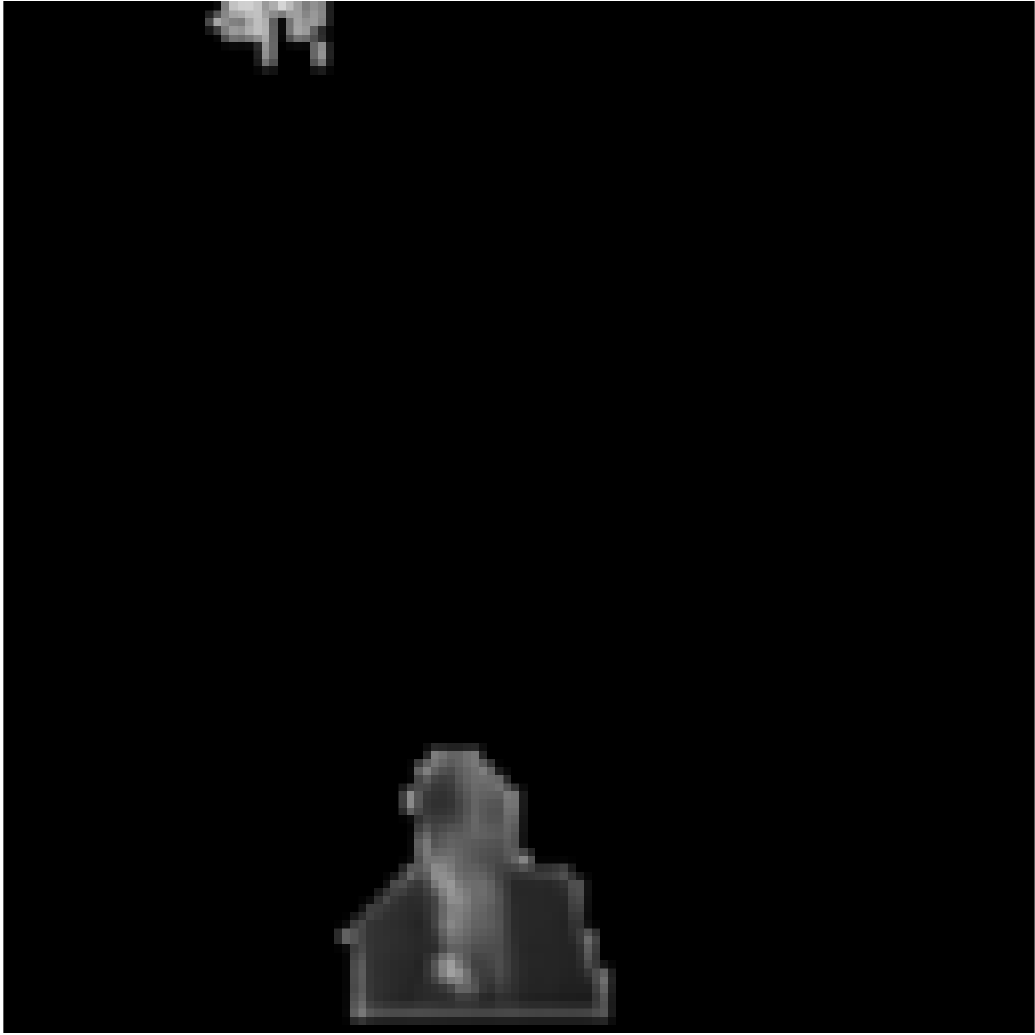} \vspace{-2mm} \\	
	\begin{minipage}{0.105\linewidth}
	\centering  \scriptsize{$\phantom{a}$}
	\end{minipage}	
	\begin{minipage}{0.105\linewidth}
	\centering  \scriptsize{$57.3 ~\rm{dB}$}
	\end{minipage}
	\begin{minipage}{0.105\linewidth}
	\centering \scriptsize{$138.6 ~\rm{dB}$}
	\end{minipage}
	\begin{minipage}{0.105\linewidth}
	\centering \scriptsize{$128.8 ~\rm{dB}$}
	\end{minipage}
	\begin{minipage}{0.105\linewidth}
	\centering \scriptsize{$142.6 ~\rm{dB}$}
	\end{minipage}\\
	\includegraphics[width = 0.105\linewidth]{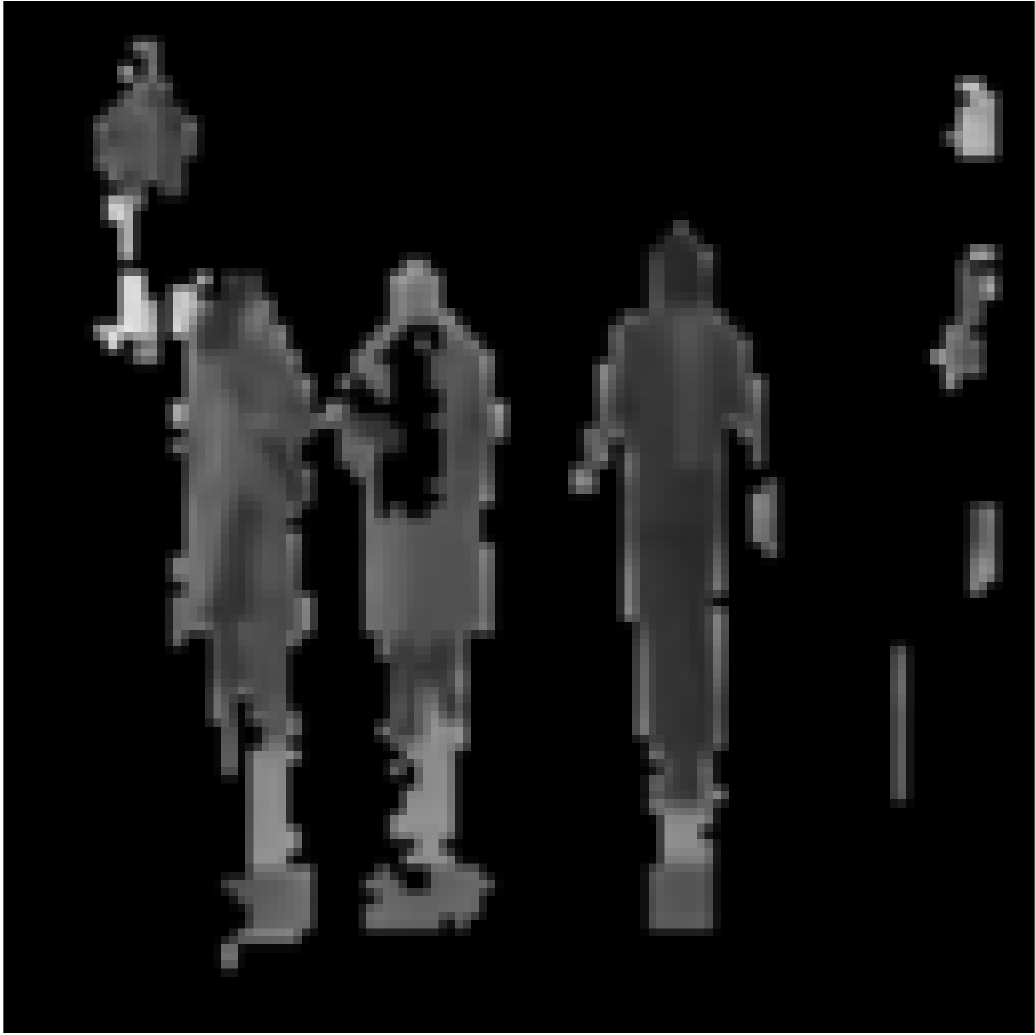}
	\includegraphics[width = 0.105\linewidth]{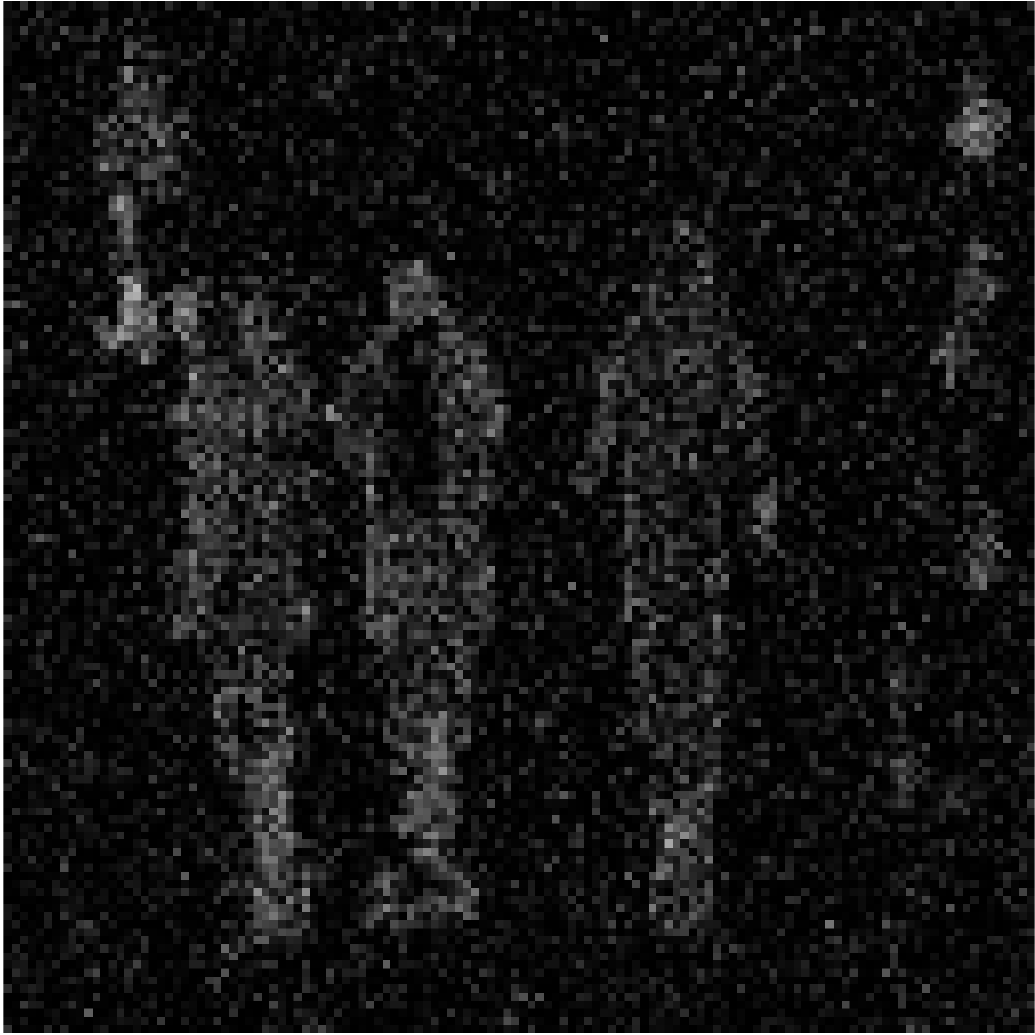}
	\includegraphics[width = 0.105\linewidth]{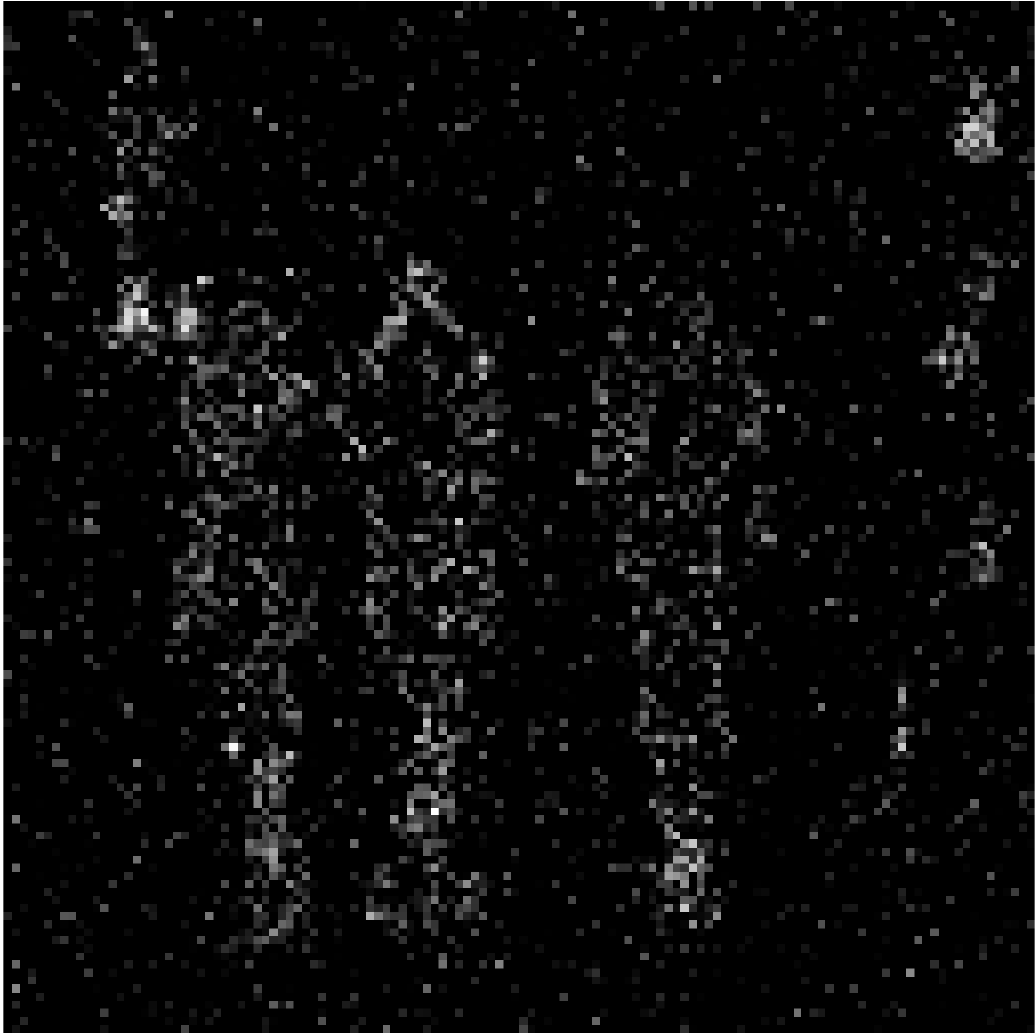}
	\includegraphics[width = 0.105\linewidth]{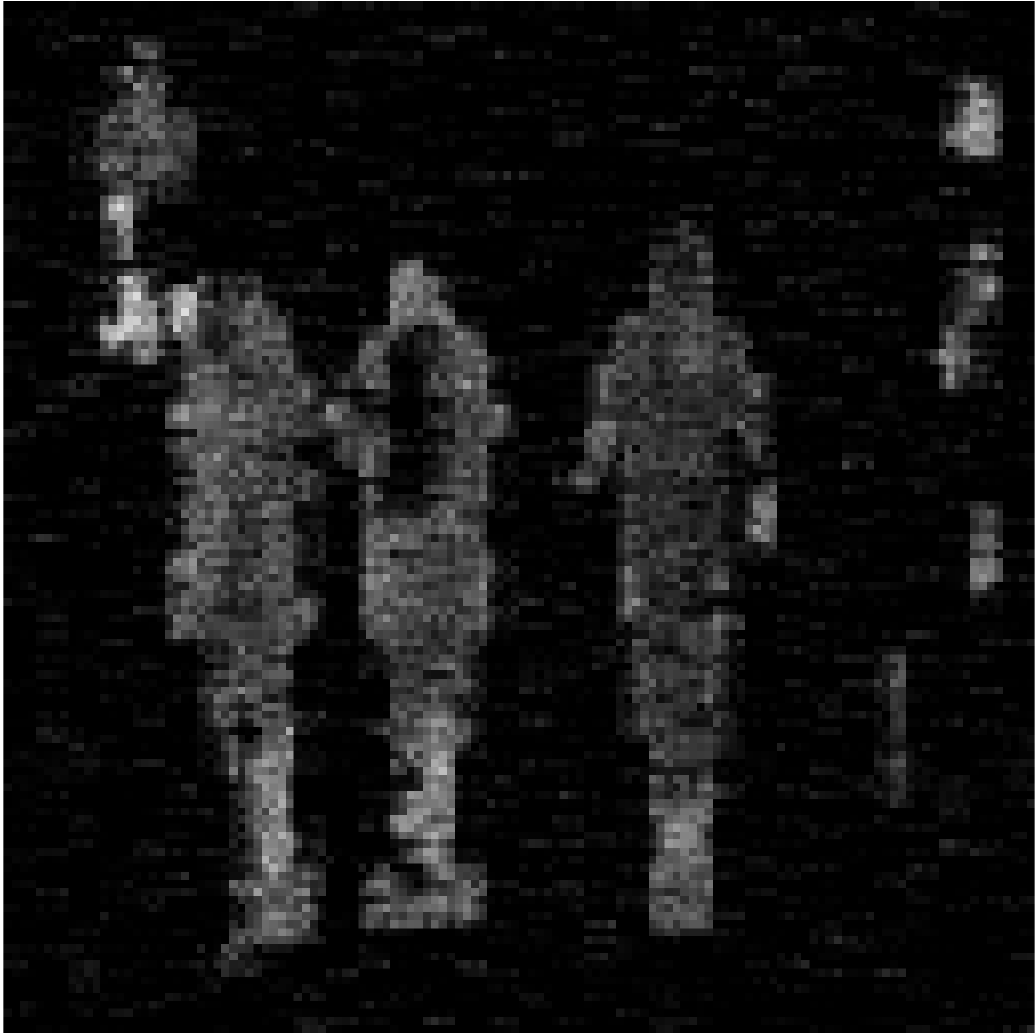}
	\includegraphics[width = 0.105\linewidth]{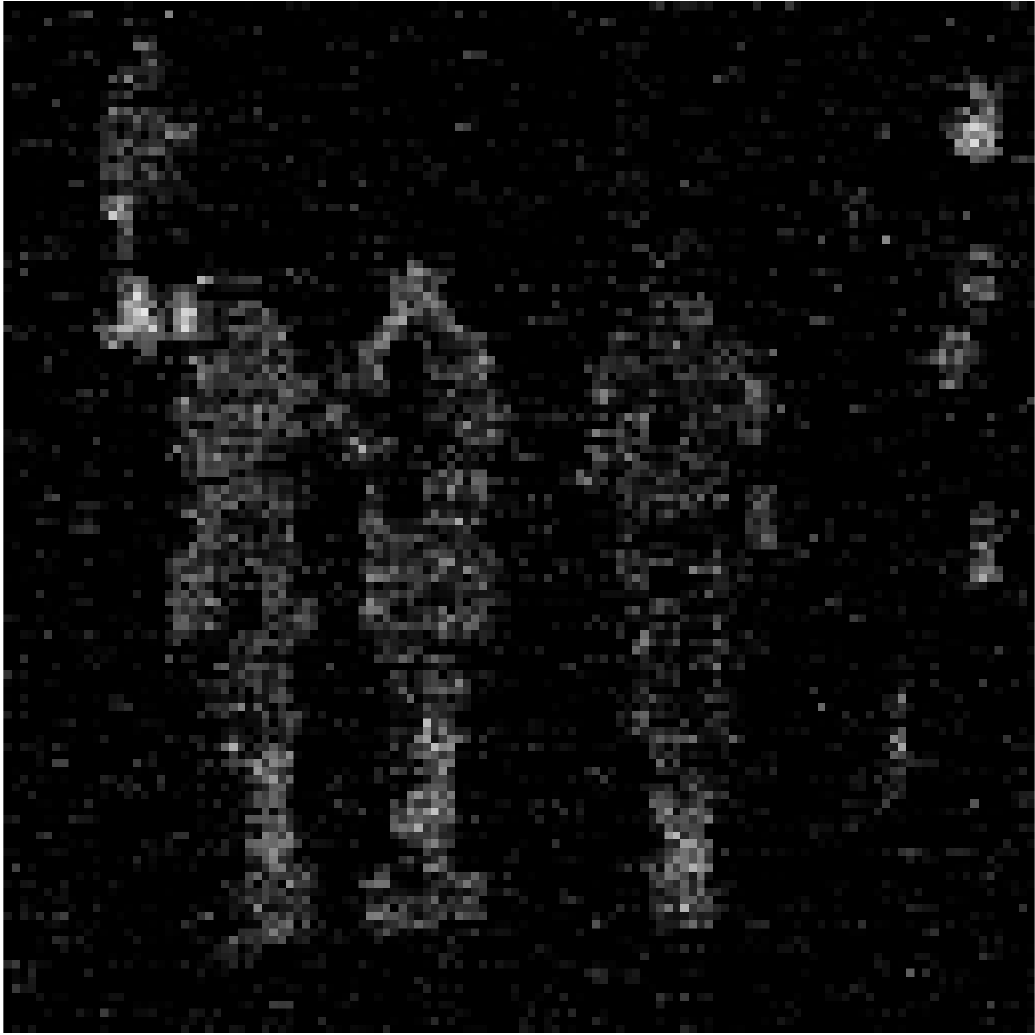} \vspace{-2mm} \\
	\begin{minipage}{0.105\linewidth}
	\centering  \scriptsize{$\phantom{a}$}
	\end{minipage}	
	\begin{minipage}{0.105\linewidth}
	\centering  \scriptsize{$19.2 ~\rm{dB}$}
	\end{minipage}
	\begin{minipage}{0.105\linewidth}
	\centering \scriptsize{$16.9 ~\rm{dB}$}
	\end{minipage}
	\begin{minipage}{0.105\linewidth}
	\centering \scriptsize{$25.8 ~\rm{dB}$}
	\end{minipage}
	\begin{minipage}{0.105\linewidth}
	\centering \scriptsize{$19.4 ~\rm{dB}$}
	\end{minipage}\\
	\includegraphics[width = 0.105\linewidth]{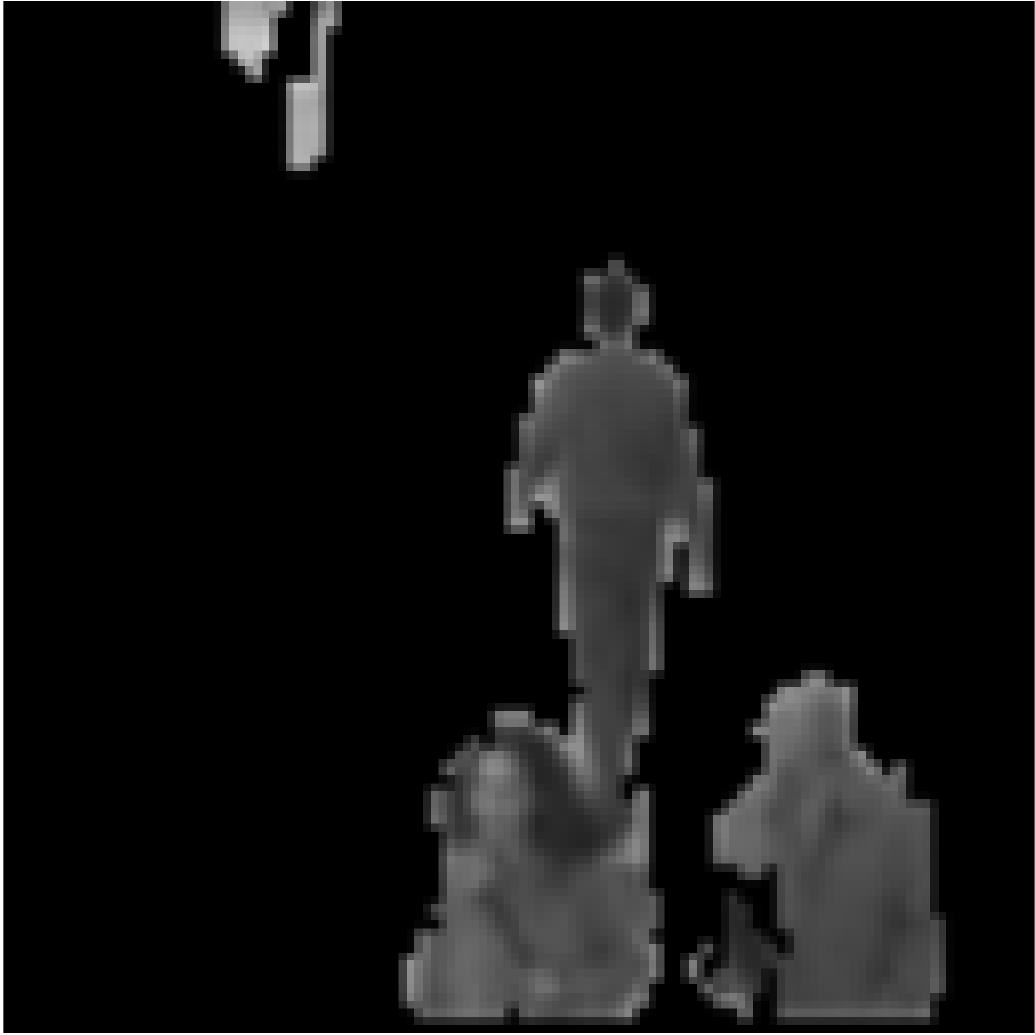}
	\includegraphics[width = 0.105\linewidth]{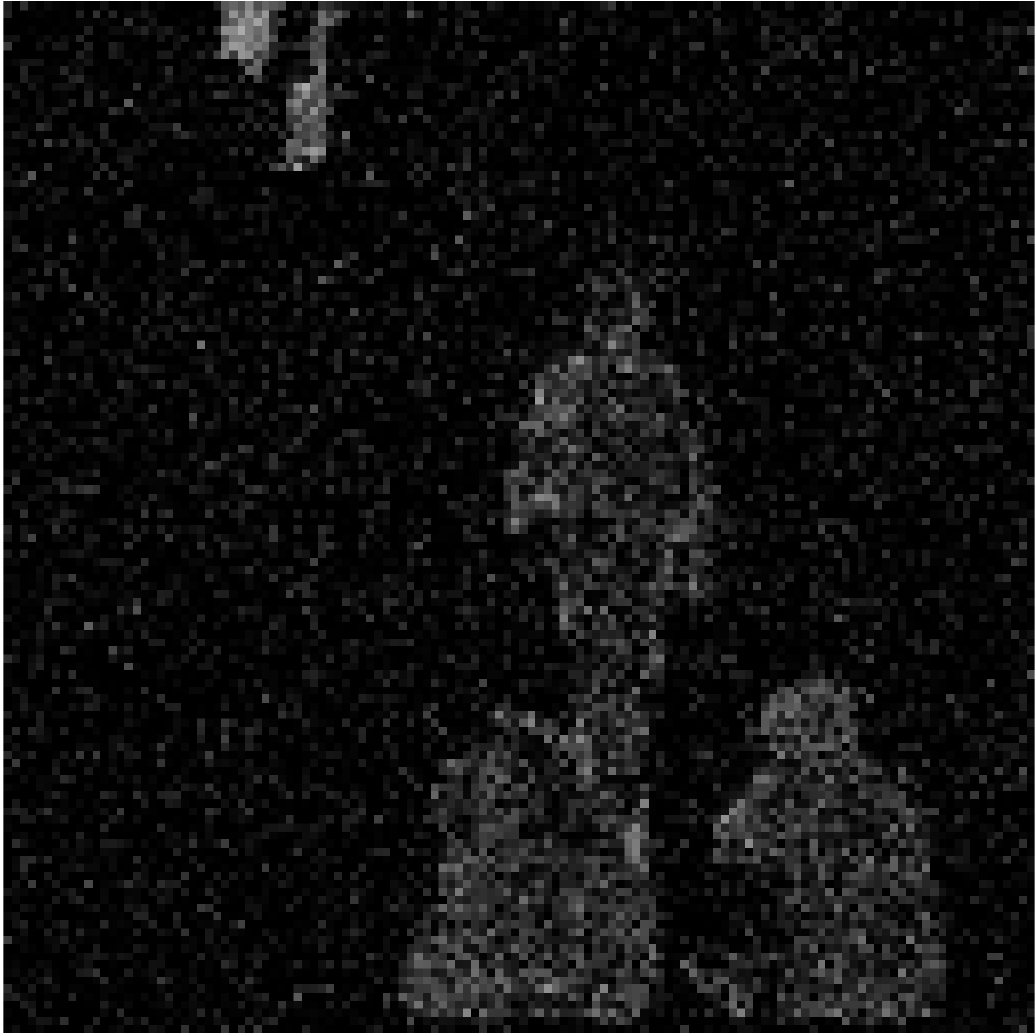}
	\includegraphics[width = 0.105\linewidth]{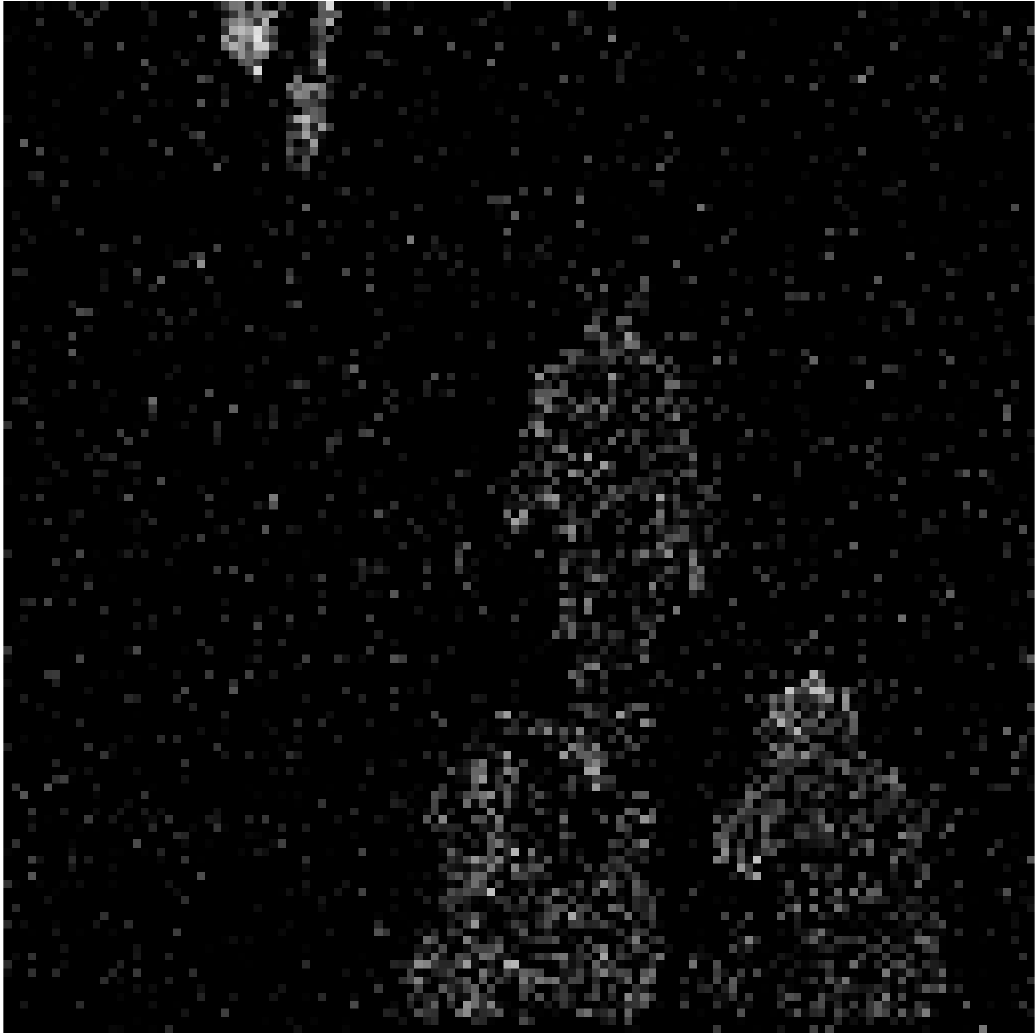}
	\includegraphics[width = 0.105\linewidth]{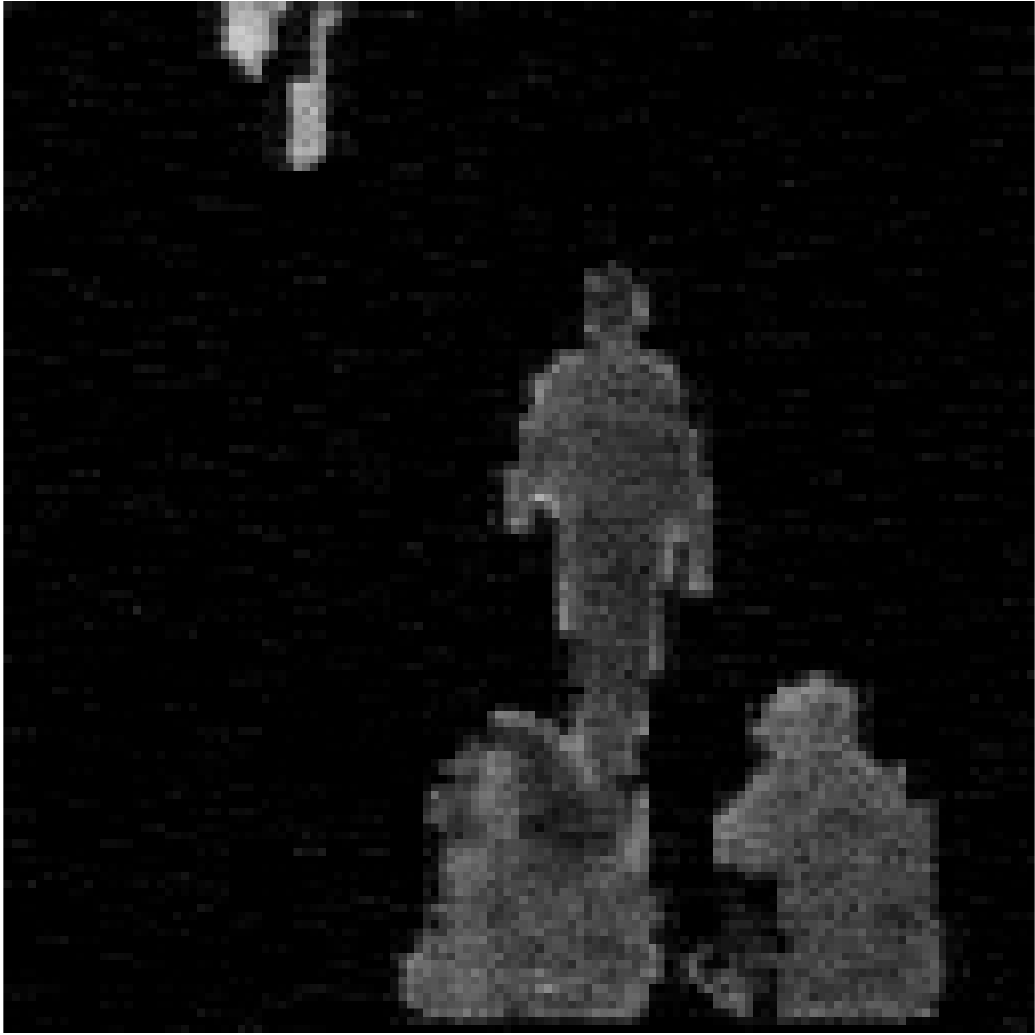}
	\includegraphics[width = 0.105\linewidth]{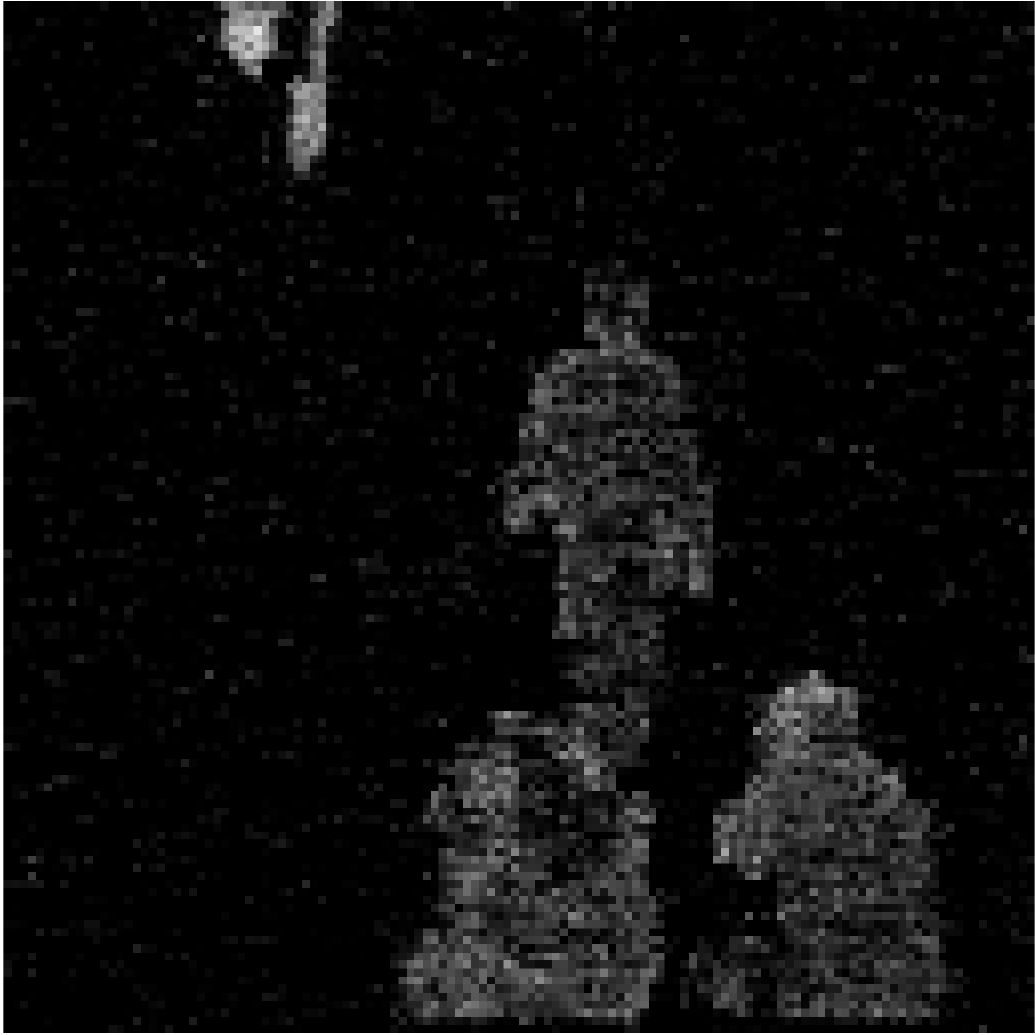} \vspace{-2mm} \\
	\begin{minipage}{0.105\linewidth}
	\centering  \scriptsize{$\phantom{a}$}
	\end{minipage}	
	\begin{minipage}{0.105\linewidth}
	\centering  \scriptsize{$20.6 ~\rm{dB}$}
	\end{minipage}
	\begin{minipage}{0.105\linewidth}
	\centering \scriptsize{$18.5 ~\rm{dB}$}
	\end{minipage}
	\begin{minipage}{0.105\linewidth}
	\centering \scriptsize{$31.5 ~\rm{dB}$}
	\end{minipage}
	\begin{minipage}{0.105\linewidth}
	\centering \scriptsize{$24.8 ~\rm{dB}$}
	\end{minipage} 
\end{minipage} \hspace{-6.5cm}
\begin{minipage}{0.1\paperwidth}
	\includegraphics[width=0.36\paperwidth]{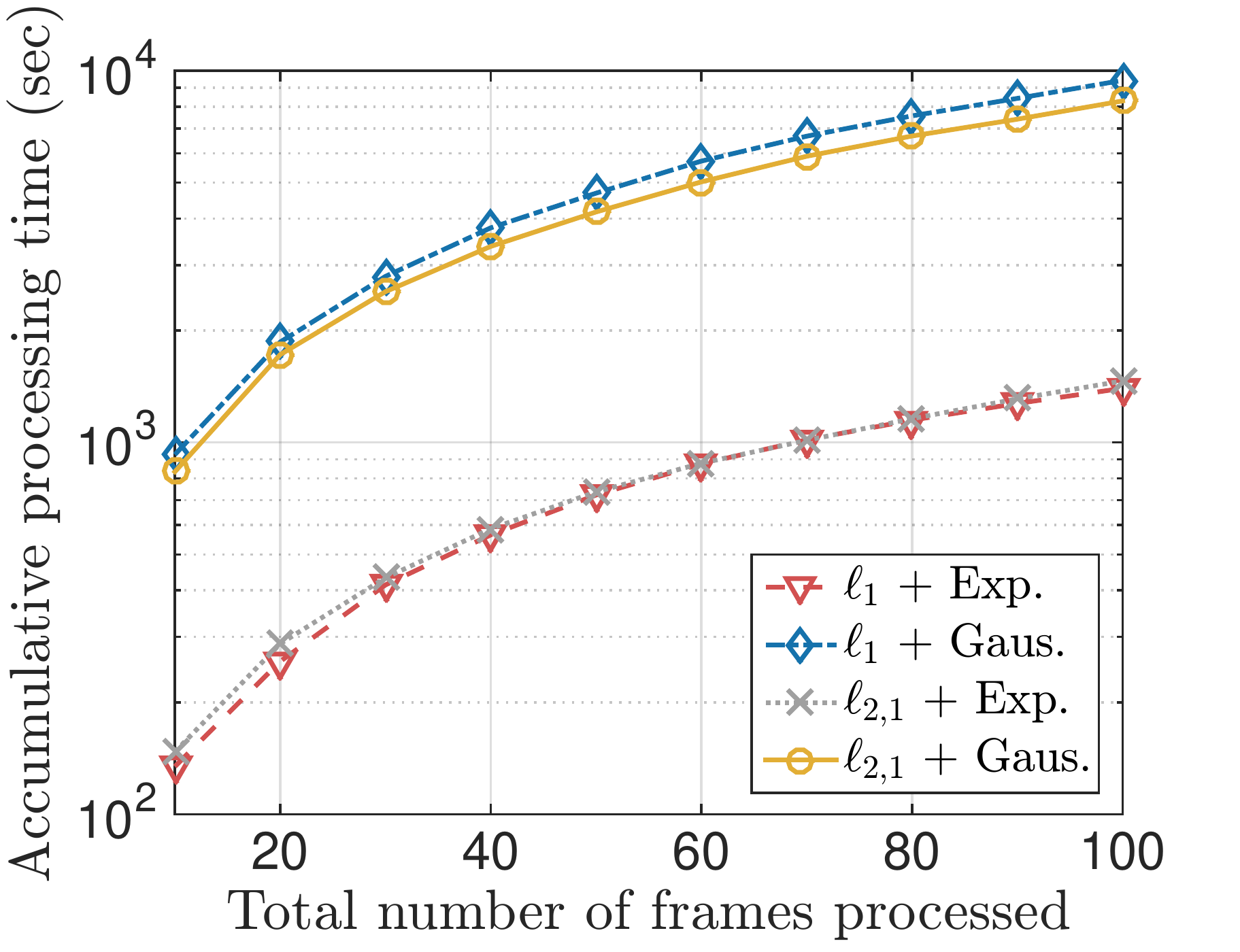}
\end{minipage}
\caption{\textit{Left panel}: representative examples of subtracted frame recovery from compressed measurements. Here,  $\numsam = \lceil 0.3 \cdot \dim \rceil $ measurements are observed for $\dim = 2^{16}$. Block sparse model $\model$ contains groups of consecutive indices where $g = 4$. \textit{Right panel}: Accumulative computational time required to process 100 frames. Overall, using Gaussian matrices in the $\ell_{2,1}$-norm case, \textsc{DecOpt} required almost 2.8 hours (upper bound), as compared to 0.55 hours when $\bPhi$ is a sparse expander matrix. Thus, while Gaussian matrices is known to lead to better recovery results if no time restrictions apply, sparse sensing matrices constitute an appealing choice in practice.} {\label{fig:002}}
\end{figure*}

For the purpose of this experiment, we set up an upper wall time of $10^4$ seconds (\textit{i.e.}, $~2.8$ hours) to process $100$ frames for each solver. This translates into $100$ seconds per frame. 

Due to the nature of the dataset, we can safely assume that nonzeros are clustered together. Thus, we assume group models $\model$ where groups are constituted of consecutive column pixels and the $\dim$ indices are divided in consecutive groups of equal size $g = \left\{4, 8, 16\right\}$. No other parameters are required -- this is an advantage over non-convex approaches, where a sparsity level is usually required to be known a priori. All experiments are repeated 10 independent times for different $\bPhi$'s.

Figure \ref{fig:002} shows some representative results. Left panel illustrates the recovery performance for different settings -- $\ell_1$- vs. $\ell_{2,1}$-norm and Gaussian vs. sparse matrices $\bPhi$; results for other configurations are presented in the appendix. The first row considers a ``simple'' image with a small number of non-zeros; the other two rows show two less sparse cases. While for the ``simple'' case, solving $\ell_{2,1}$-norm minimization with Gaussian matrices lead to better recovery results -- within the time constraints, the same does not apply for the more ``complex'' cases. Overall, we observe that, \emph{given such time restrictions per frame}, by using expander matrices one can achieve a better solution in terms of PSNR \emph{faster}. 
This is shown in more detail in Figure \ref{fig:002} (right panel); see also the appendix for more results. 

\section{Conclusions}

Sparse matrices are favorable objects in machine learning and optimization. When such matrices can be applied in favor of dense ones, the computational requirements can be significantly reduced in practice, both in terms of space and runtime complexity. In this work, we both show theoretically and experimentally that such selection is advantageous for the case of group-based basis pursuit recovery from linear measurements. As future work, one can consider other sparsity models, as well as different objective criteria. 

\begin{small}
\bibliographystyle{plain}
\bibliography{BlockSparseExpanders}
\end{small}
\section{Appendix}

Here, we report further results on the 2D image recovery problem. We remind that, for the purpose of this experiment, we set up an upper wall time of $10^4$ seconds (\textit{i.e.}, $~2.8$ hours) to process $100$ frames for each solver. This translates into $100$ seconds per frame.

\subsection{Varying group size $g$}
For this case, we focus on a single frame. Due to its higher number of non-zeros, we have selected the frame shown in Figure \ref{fig:003}. For this case, we consider a roughly sufficient number of measurements is acquired where $\numsam = \lceil 0.3 \cdot \dim \rceil$. By varying the group size $g$, we obtain the results in Figure \ref{fig:003}. 

\subsection{Varying number of measurements}
Here, let $g = 4$ as this group selection performs better, as shown in the previous subection. Here, we consider $\numsam$ take values from $\numsam \in \left\lceil \left\{ 0.25, 0.3, 0.35, 0.4 \right\} \cdot \dim \right \rceil$. The results, are shown in Figure \ref{fig:004}.

\begin{figure*}[t]
\centering
	\begin{minipage}{0.15\linewidth}
	\centering \small{Original}
	\end{minipage}
	\begin{minipage}{0.15\linewidth}
	\centering \small{$\ell_1$ + Exp.}
	\end{minipage}
	\begin{minipage}{0.15\linewidth}
	\centering \small{$\ell_1$ + Gaus.}
	\end{minipage}
	\begin{minipage}{0.15\linewidth}
	\centering \small{$\ell_{2,1}$ + Exp.}
	\end{minipage}
	\begin{minipage}{0.15\linewidth}
	\centering \small{$\ell_{2,1}$ + Gaus.}
	\end{minipage}\\
	\includegraphics[width = 0.18\linewidth]{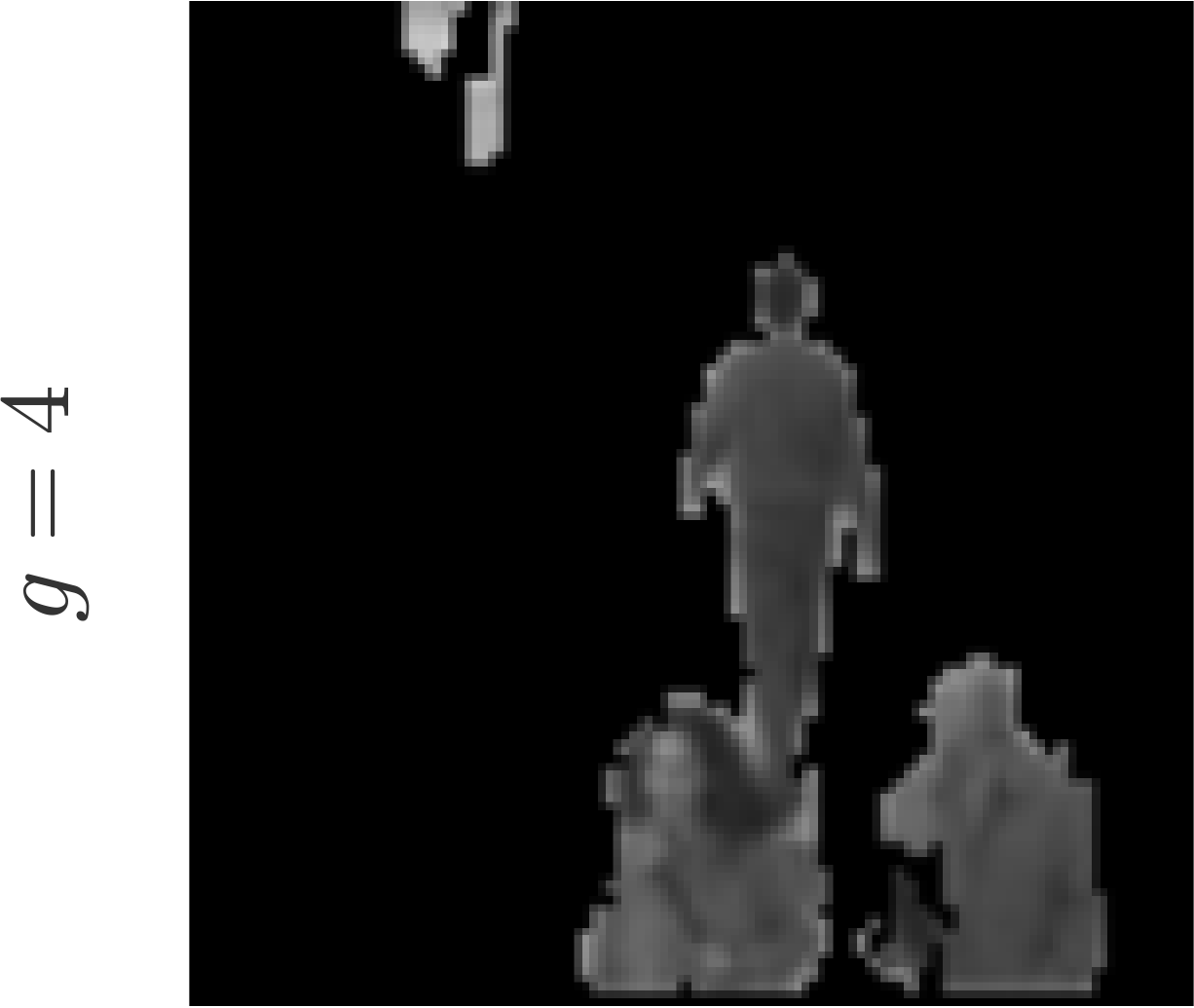}
	\includegraphics[width = 0.15\linewidth]{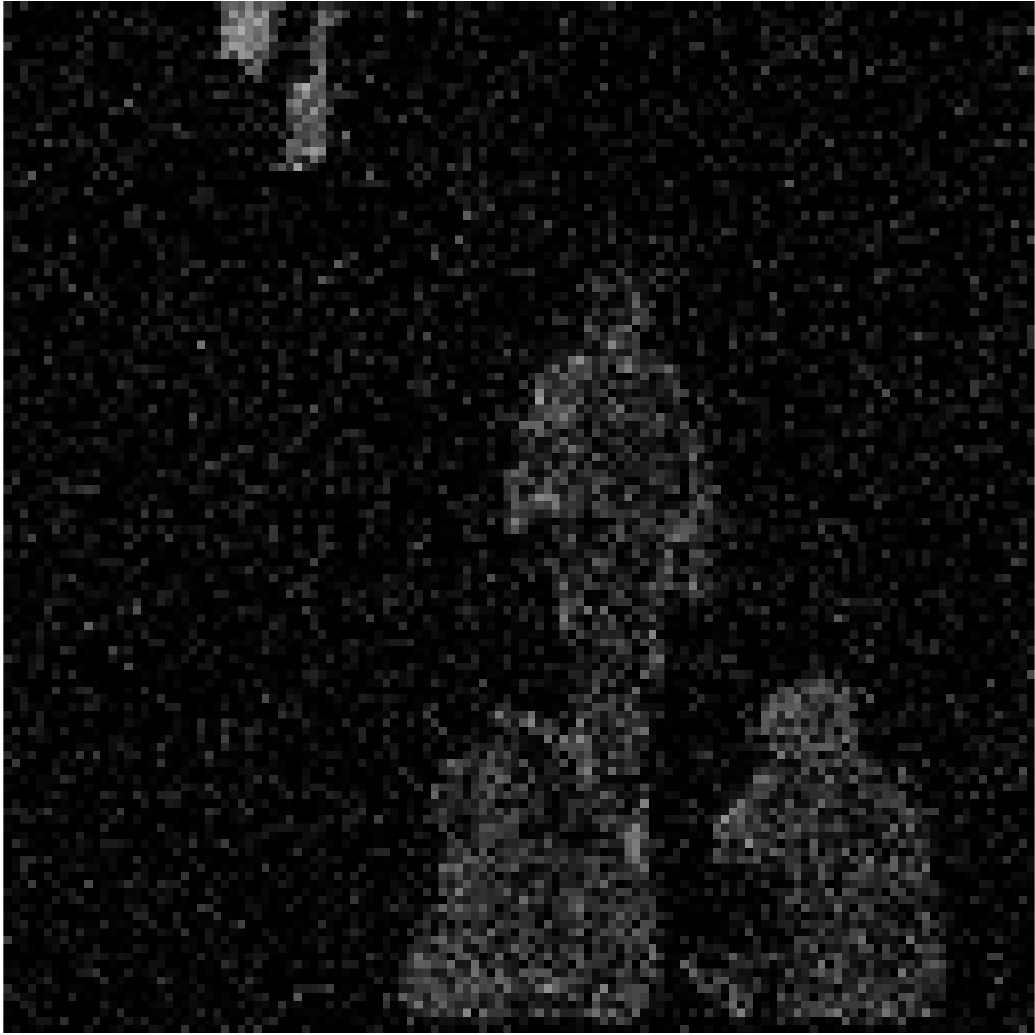}
	\includegraphics[width = 0.15\linewidth]{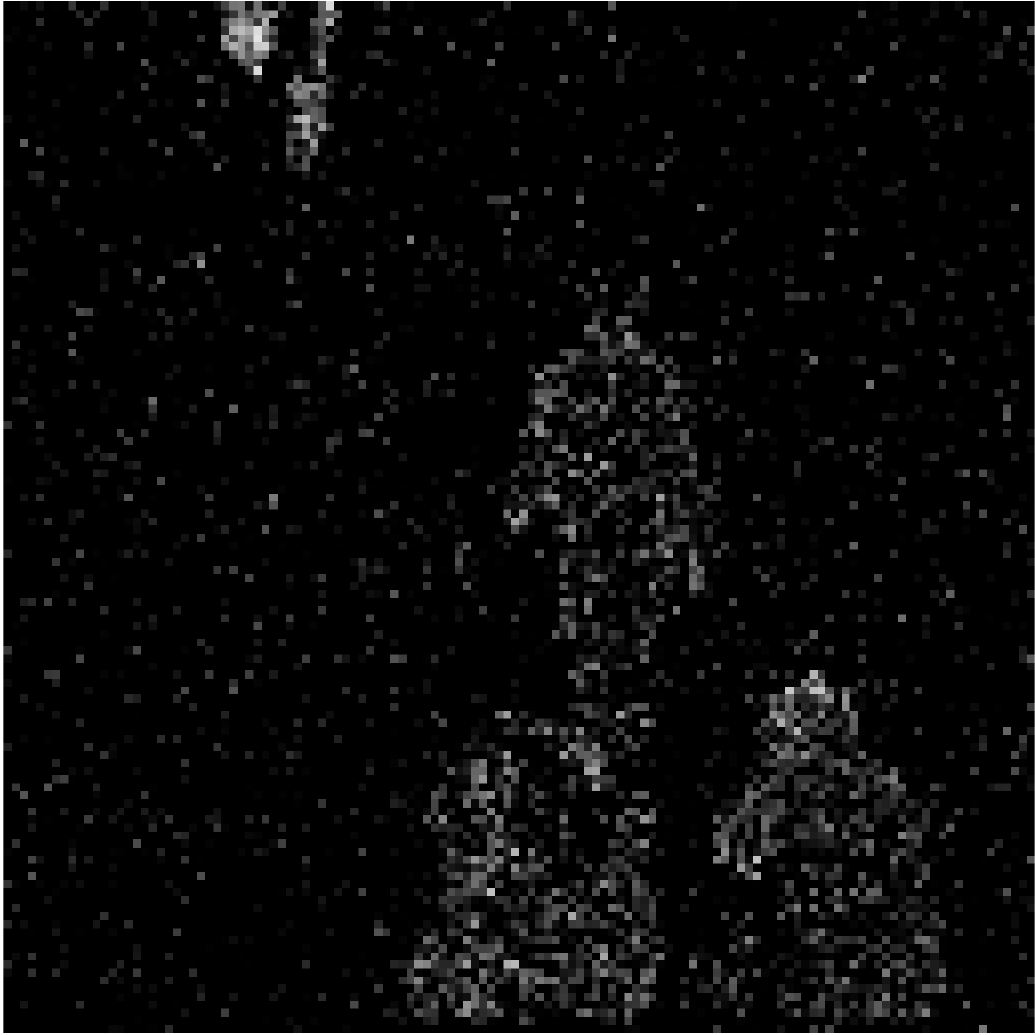}
	\includegraphics[width = 0.15\linewidth]{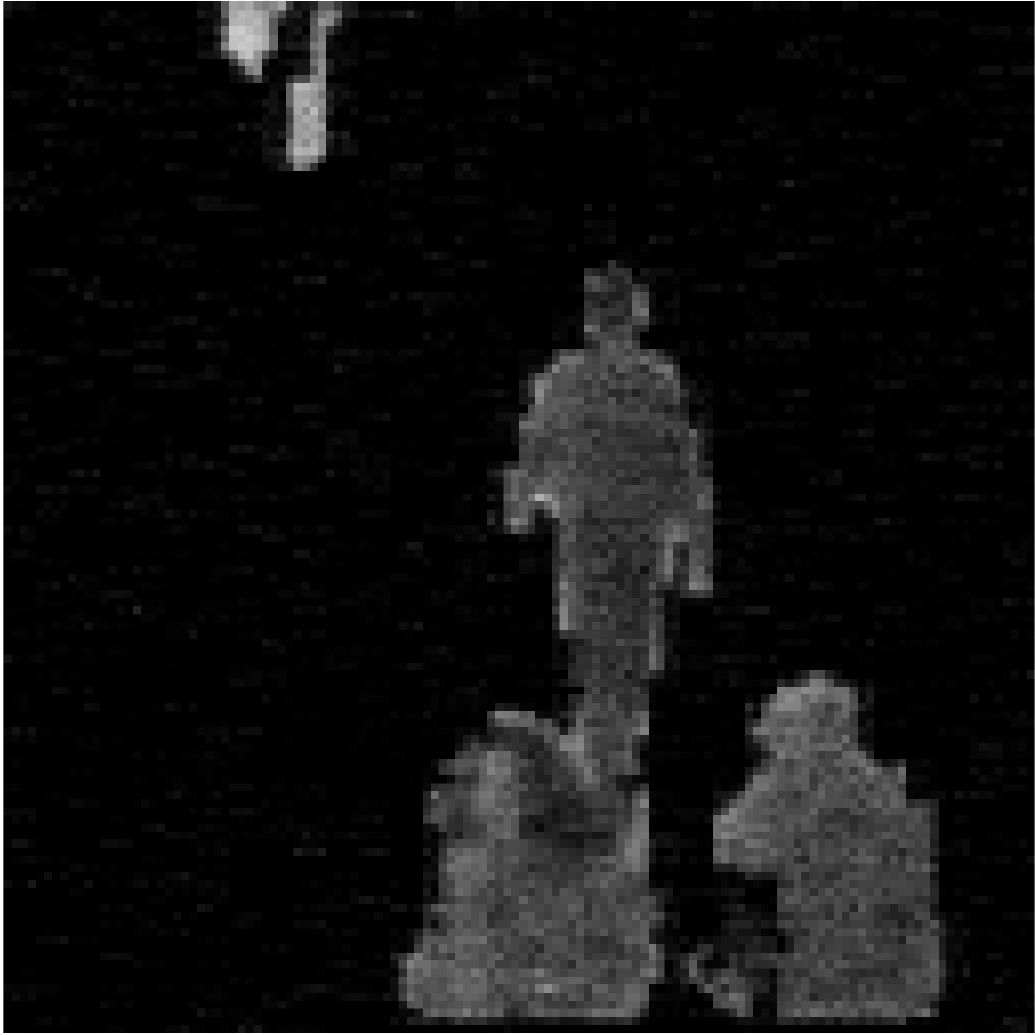}
	\includegraphics[width = 0.15\linewidth]{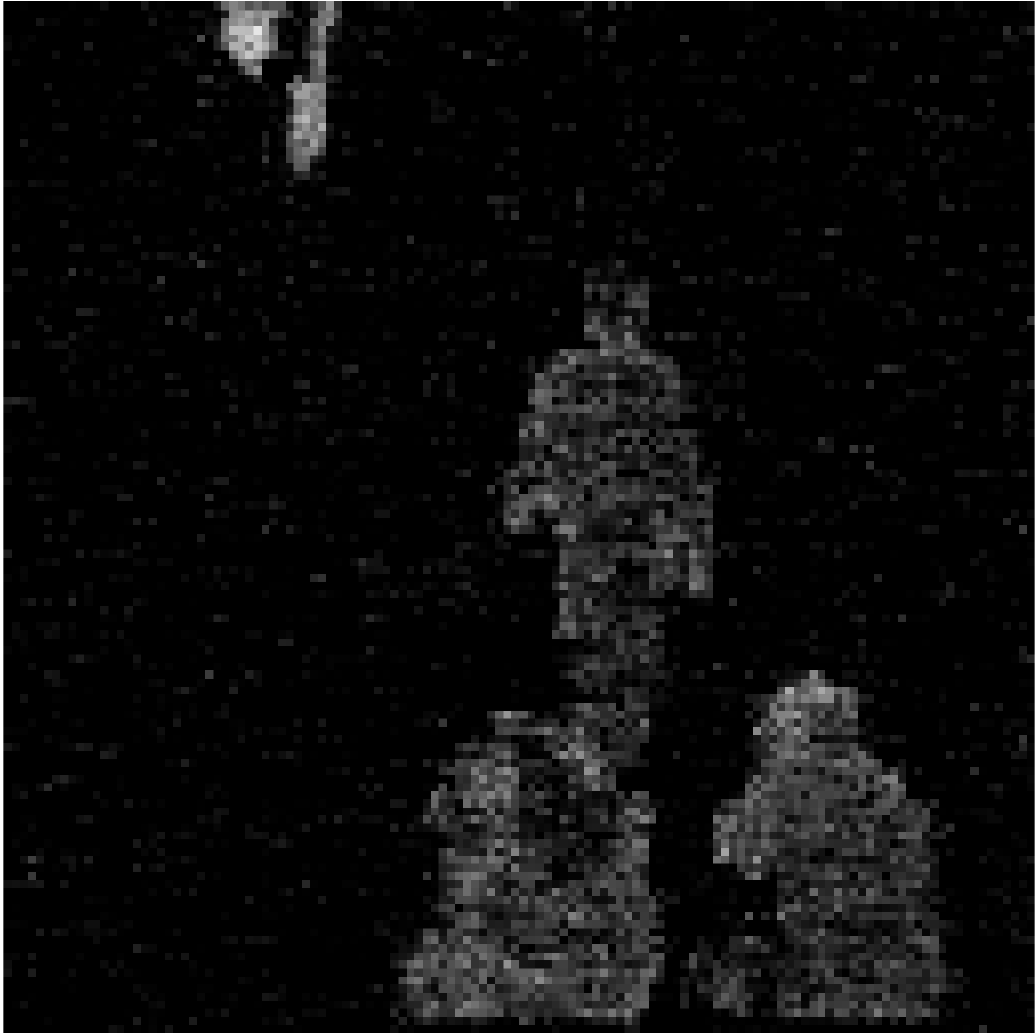} \vspace{-2mm} \\	
	\begin{minipage}{0.15\linewidth}
	\centering  \small{$\phantom{a}$}
	\end{minipage}	
	\begin{minipage}{0.15\linewidth}
	\centering  \small{$20.6 ~\rm{dB}$}
	\end{minipage}
	\begin{minipage}{0.15\linewidth}
	\centering \small{$18.5 ~\rm{dB}$}
	\end{minipage}
	\begin{minipage}{0.15\linewidth}
	\centering \small{$31.5 ~\rm{dB}$}
	\end{minipage}
	\begin{minipage}{0.15\linewidth}
	\centering \small{$24.8 ~\rm{dB}$}
	\end{minipage}\\
	\includegraphics[width = 0.18\linewidth]{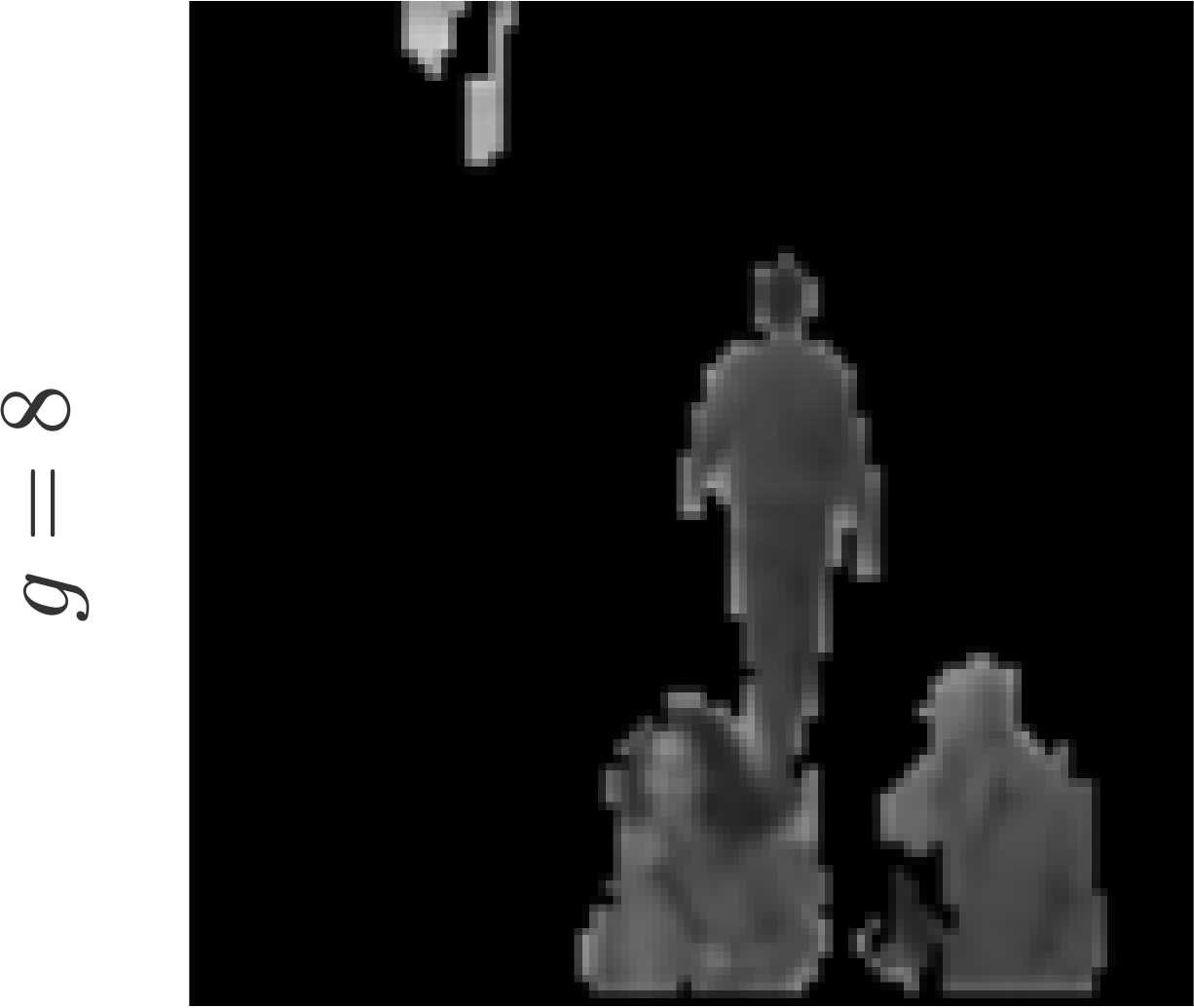}
	\includegraphics[width = 0.15\linewidth]{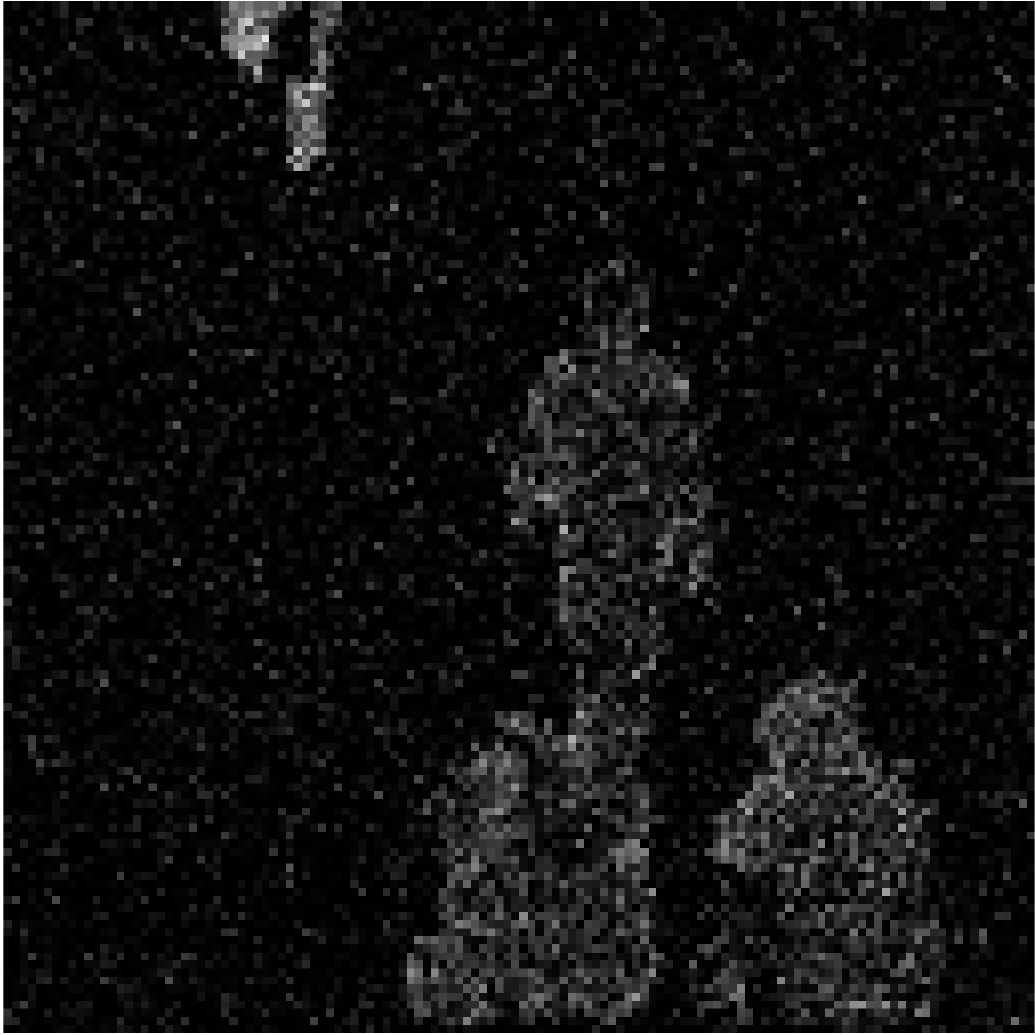}
	\includegraphics[width = 0.15\linewidth]{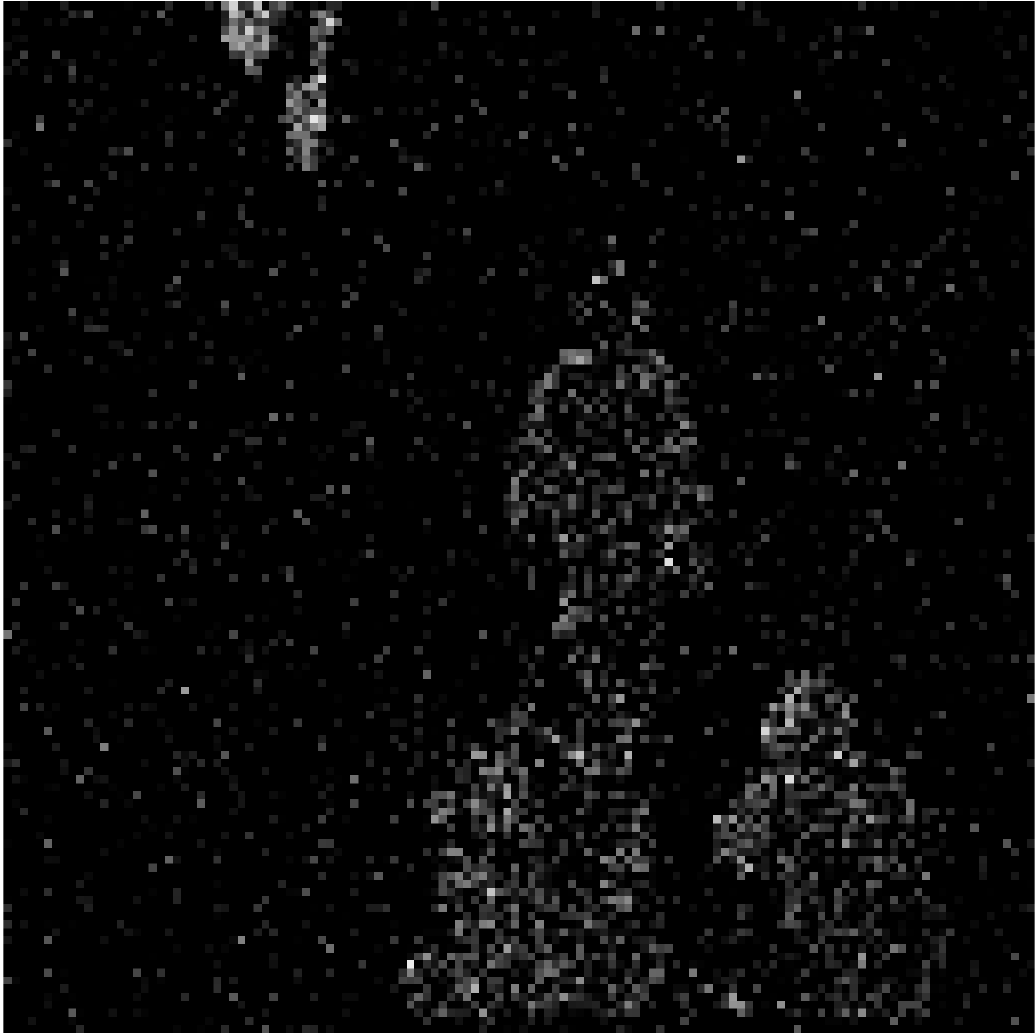}
	\includegraphics[width = 0.15\linewidth]{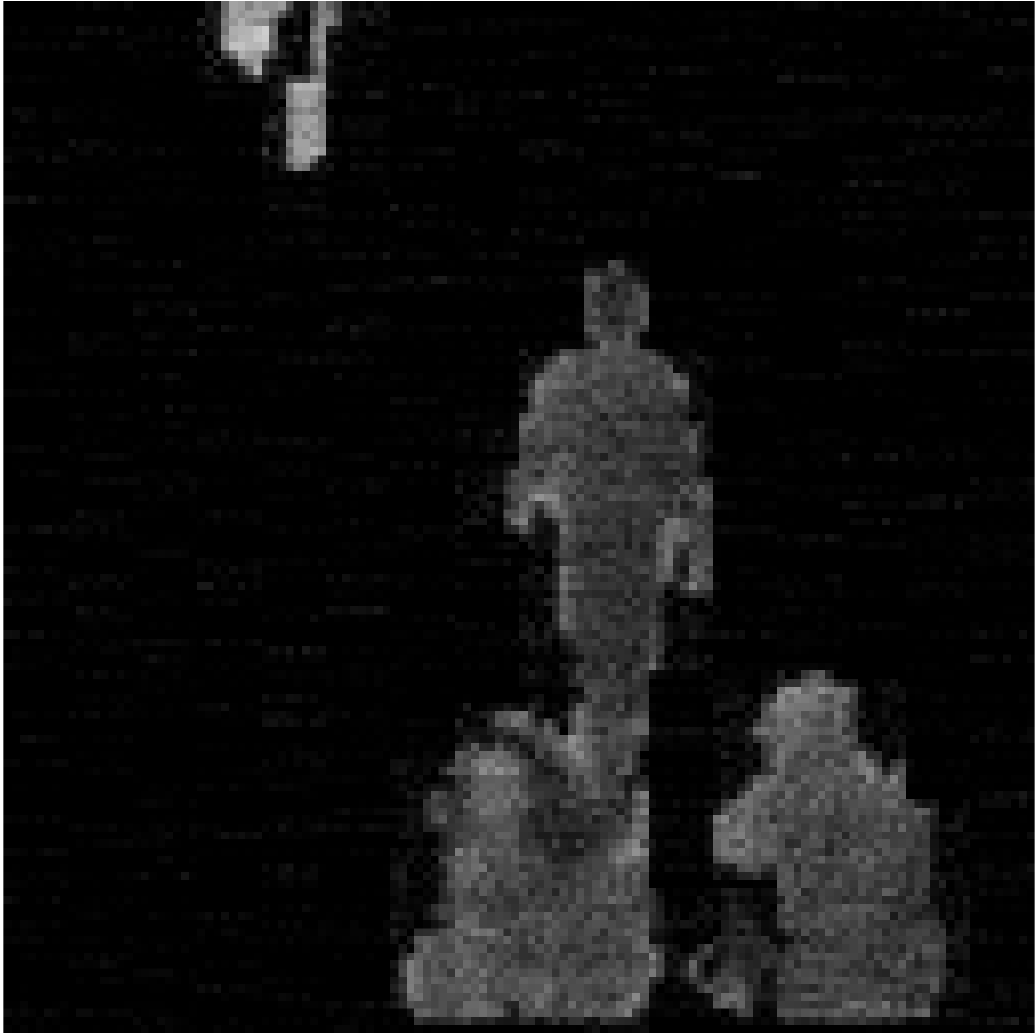}
	\includegraphics[width = 0.15\linewidth]{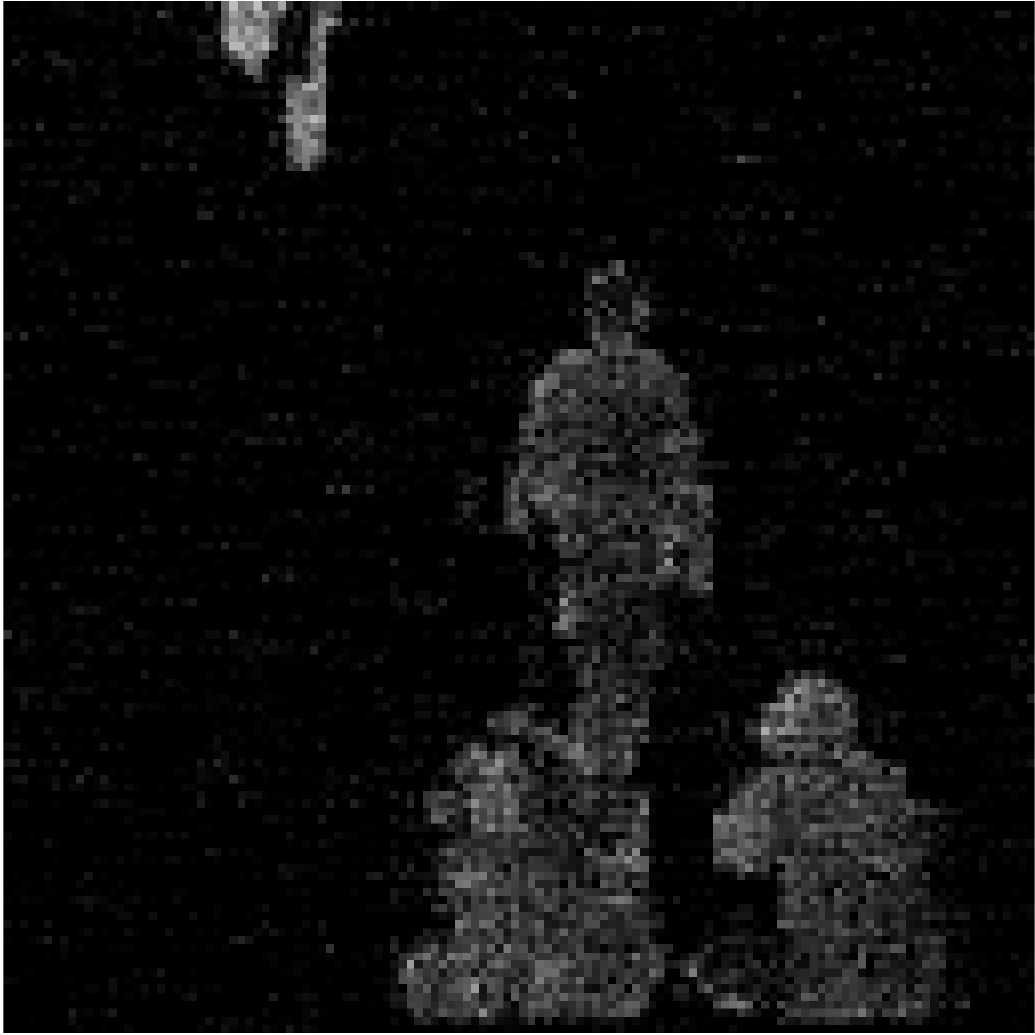} \vspace{-2mm} \\
	\begin{minipage}{0.15\linewidth}
	\centering  \small{$\phantom{a}$}
	\end{minipage}	
	\begin{minipage}{0.15\linewidth}
	\centering  \small{$20.6 ~\rm{dB}$}
	\end{minipage}
	\begin{minipage}{0.15\linewidth}
	\centering \small{$18.4 ~\rm{dB}$}
	\end{minipage}
	\begin{minipage}{0.15\linewidth}
	\centering \small{$30.8 ~\rm{dB}$}
	\end{minipage}
	\begin{minipage}{0.15\linewidth}
	\centering \small{$23.3 ~\rm{dB}$}
	\end{minipage}\\
	\includegraphics[width = 0.18\linewidth]{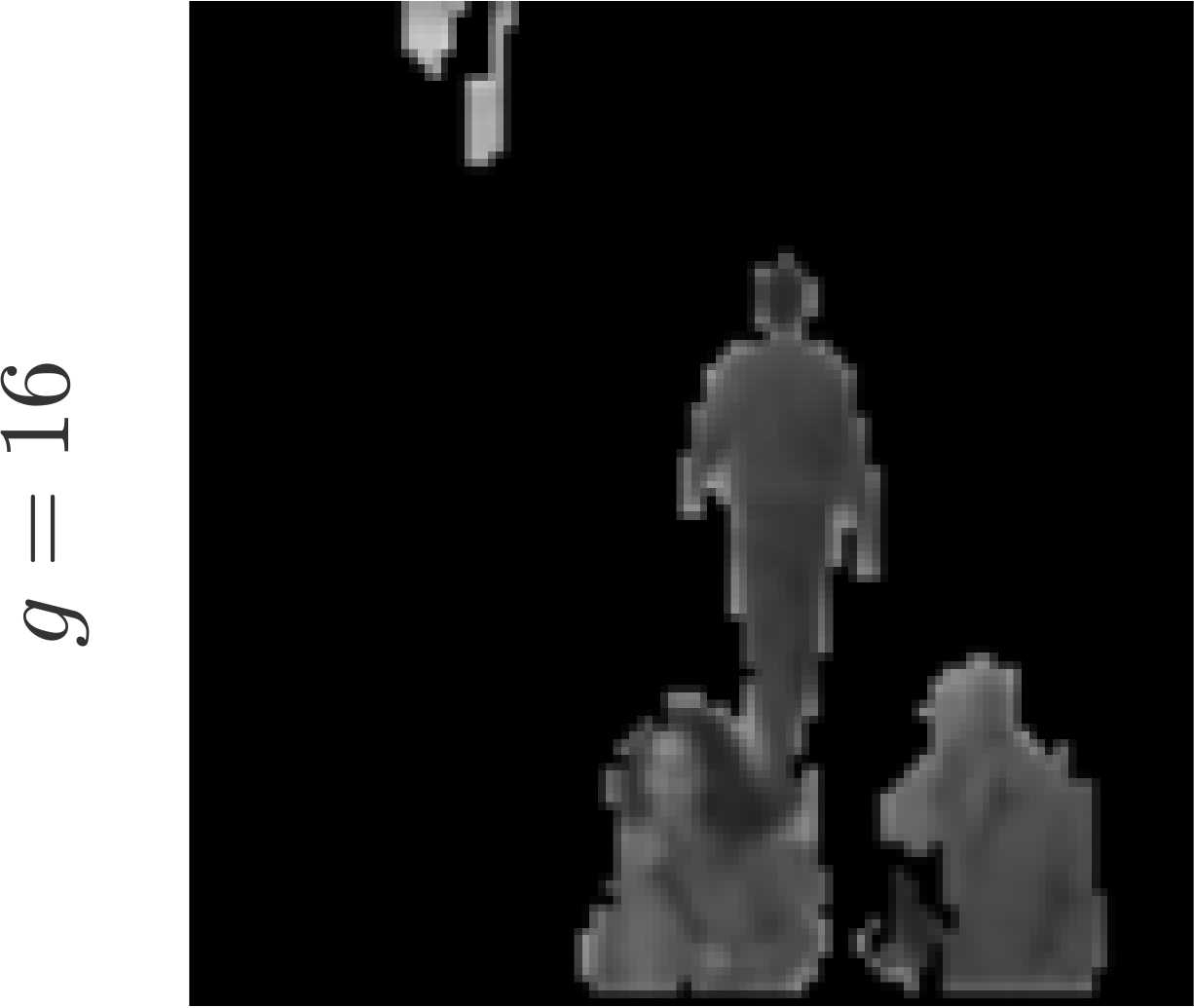}
	\includegraphics[width = 0.15\linewidth]{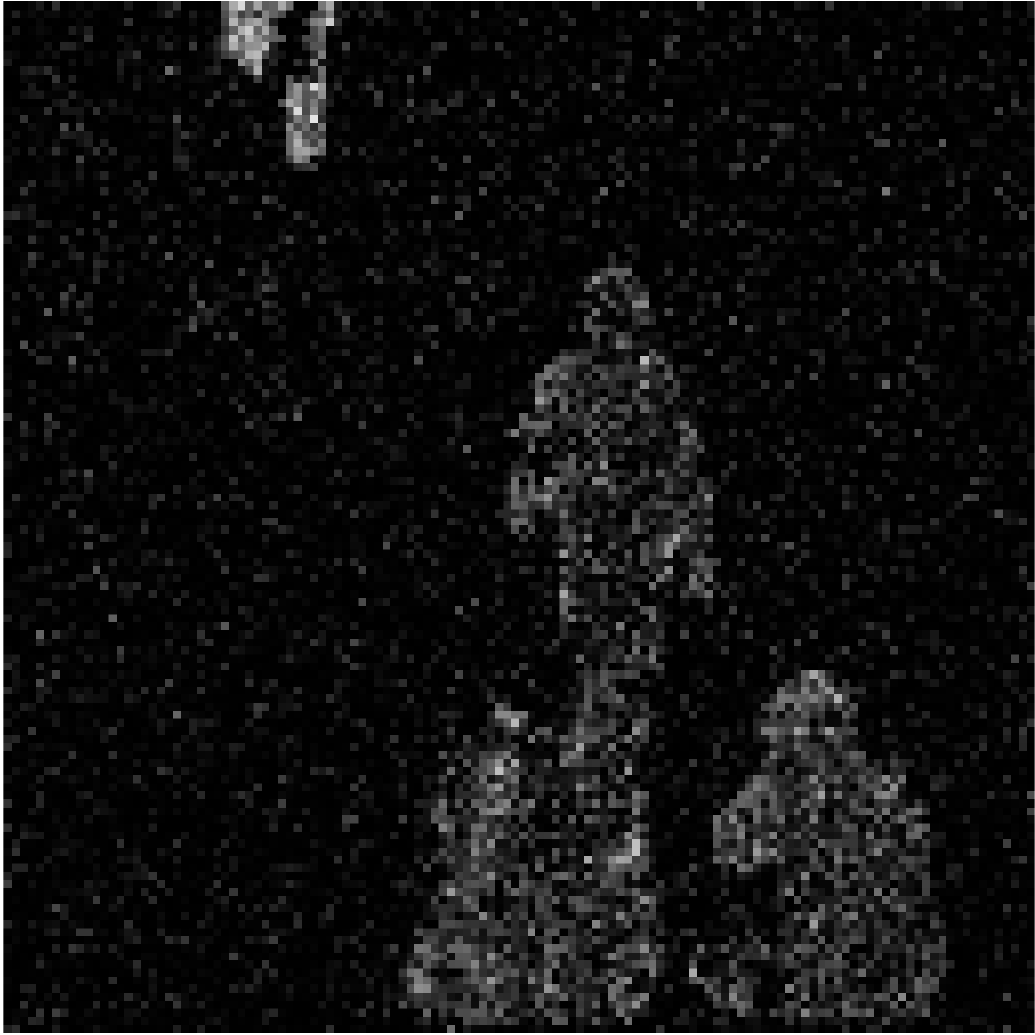}
	\includegraphics[width = 0.15\linewidth]{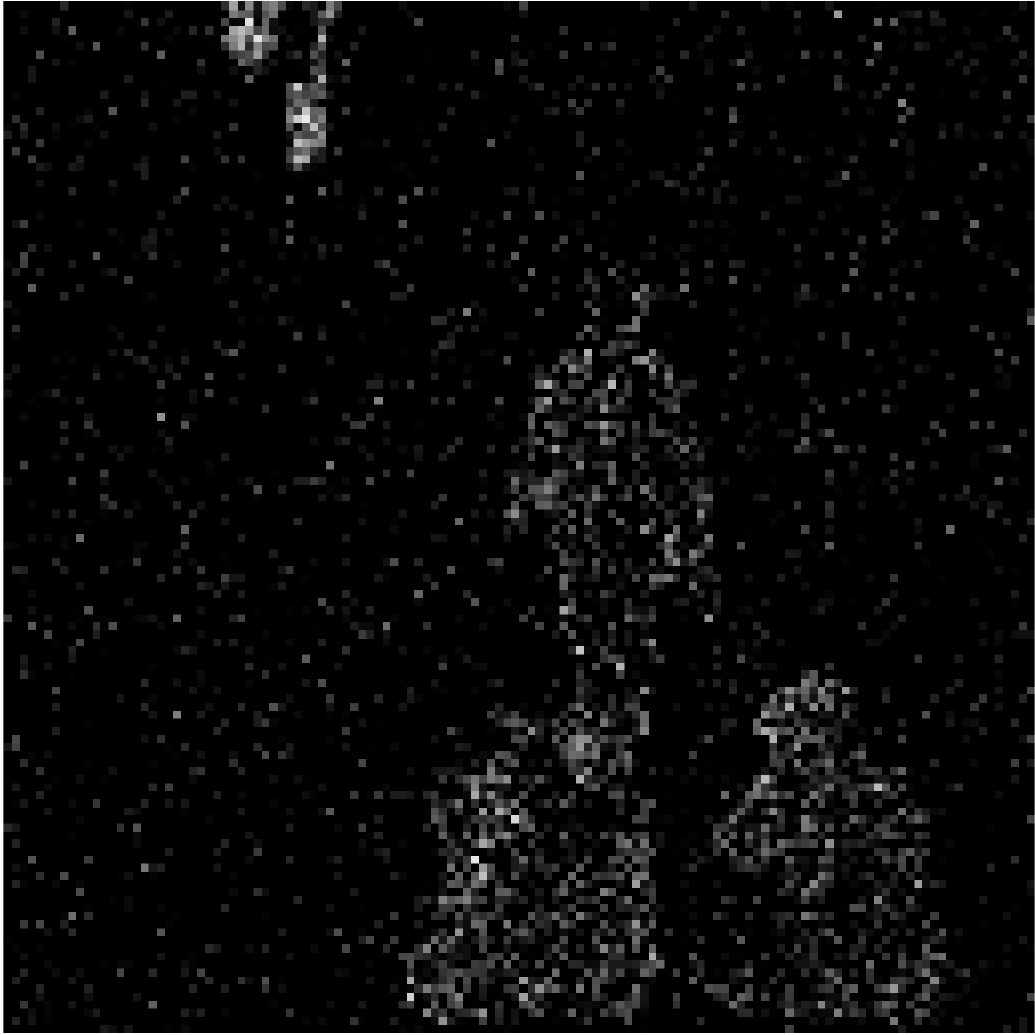}
	\includegraphics[width = 0.15\linewidth]{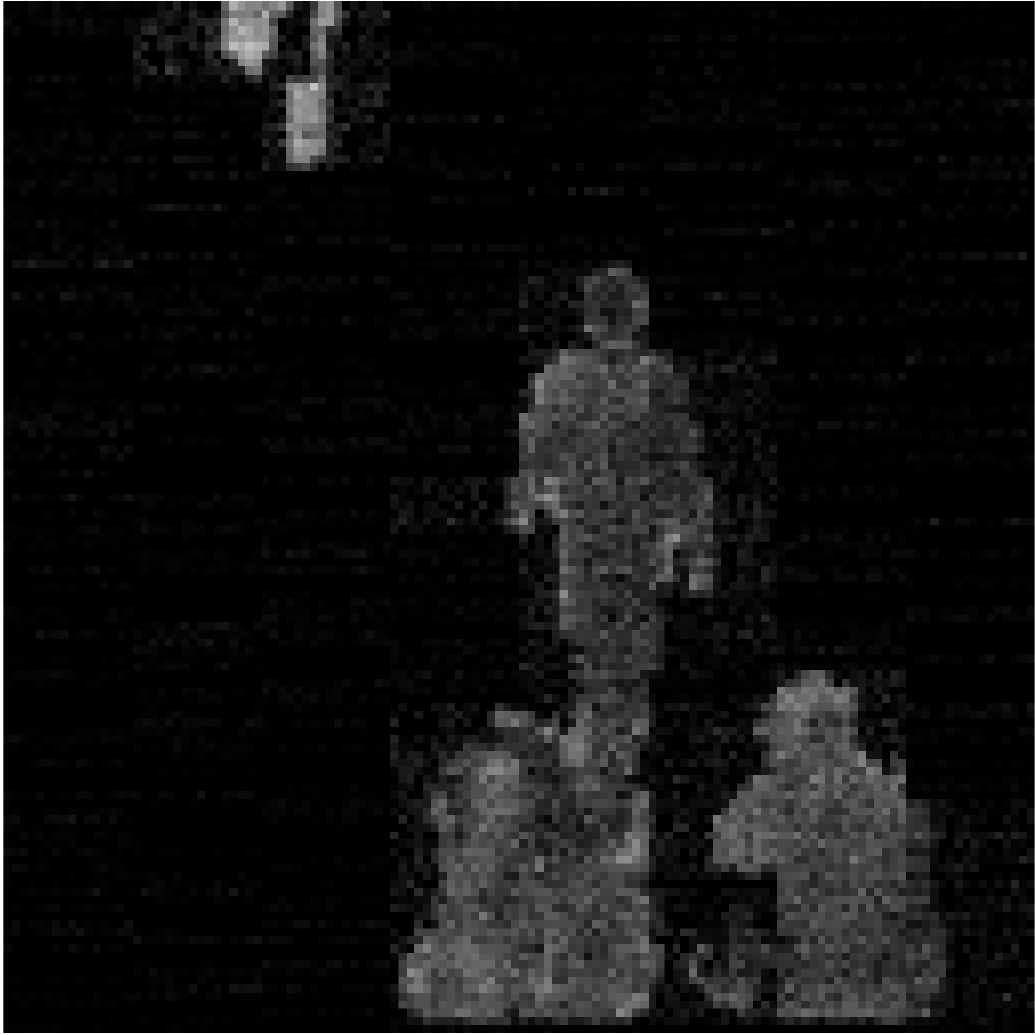}
	\includegraphics[width = 0.15\linewidth]{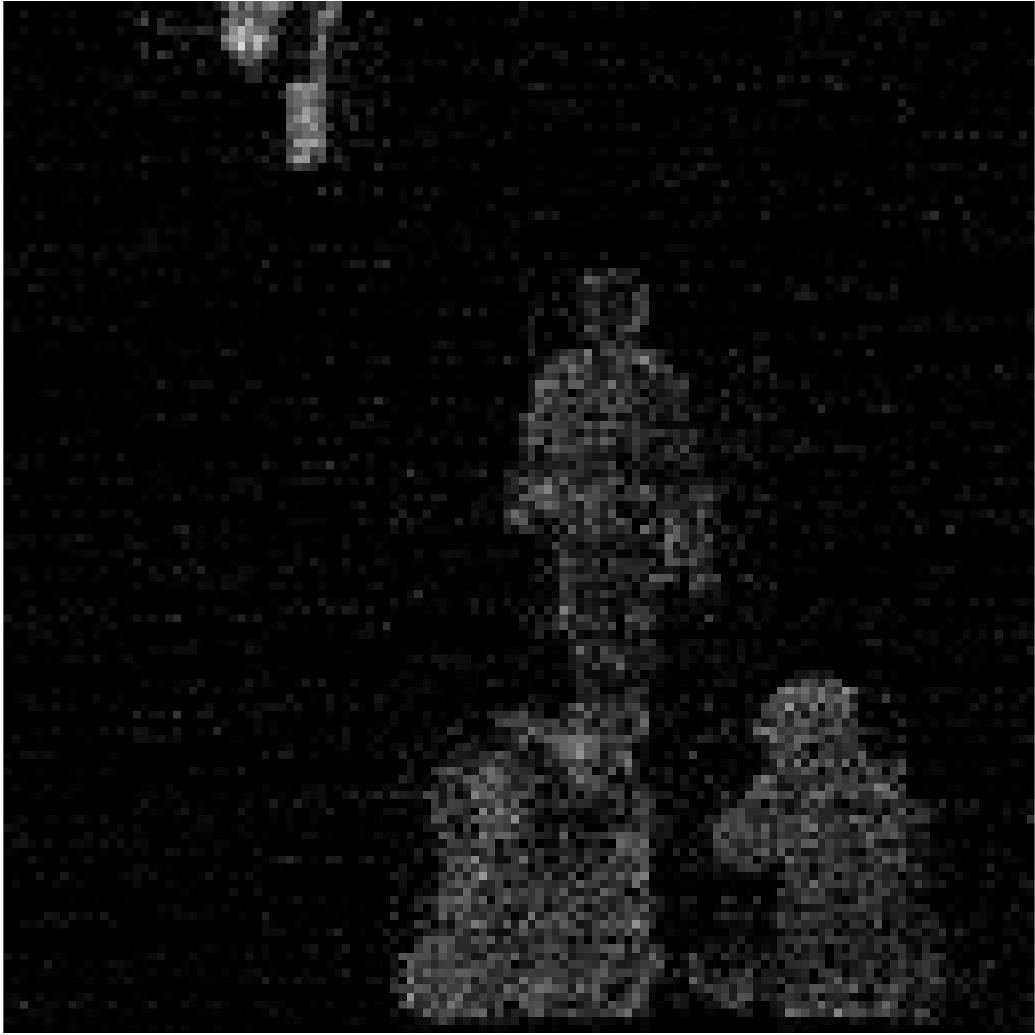} \vspace{-2mm} \\
	\begin{minipage}{0.15\linewidth}
	\centering  \small{$\phantom{a}$}
	\end{minipage}	
	\begin{minipage}{0.15\linewidth}
	\centering  \small{$20.7 ~\rm{dB}$}
	\end{minipage}
	\begin{minipage}{0.15\linewidth}
	\centering \small{$18.5 ~\rm{dB}$}
	\end{minipage}
	\begin{minipage}{0.15\linewidth}
	\centering \small{$28.0 ~\rm{dB}$}
	\end{minipage}
	\begin{minipage}{0.15\linewidth}
	\centering \small{$22.2 ~\rm{dB}$}
	\end{minipage} 

\caption{Results from real data. Representative examples of subtracted frame recovery from compressed measurements. Here,  $\numsam = \lceil 0.3 \cdot \dim \rceil $ measurements are observed for $\dim = 2^{16}$. From top to bottom, each line corresponds to block sparse model $\model$ with groups of consecutive indices, where $g = 4$, $g = 8$, and $g = 16$, respectively. One can observe that one obtains worse recovery as the group size increases; thus a model with groups $g = 4$ is a good choice for this case. } {\label{fig:003}}
\end{figure*}

\begin{figure*}[b]
\centering
	\begin{minipage}{0.15\linewidth}
	\centering \small{Original}
	\end{minipage}
	\begin{minipage}{0.15\linewidth}
	\centering \small{$\ell_1$ + Exp.}
	\end{minipage}
	\begin{minipage}{0.15\linewidth}
	\centering \small{$\ell_1$ + Gaus.}
	\end{minipage}
	\begin{minipage}{0.15\linewidth}
	\centering \small{$\ell_{2,1}$ + Exp.}
	\end{minipage}
	\begin{minipage}{0.15\linewidth}
	\centering \small{$\ell_{2,1}$ + Gaus.}
	\end{minipage}\\
	\includegraphics[width = 0.18\linewidth]{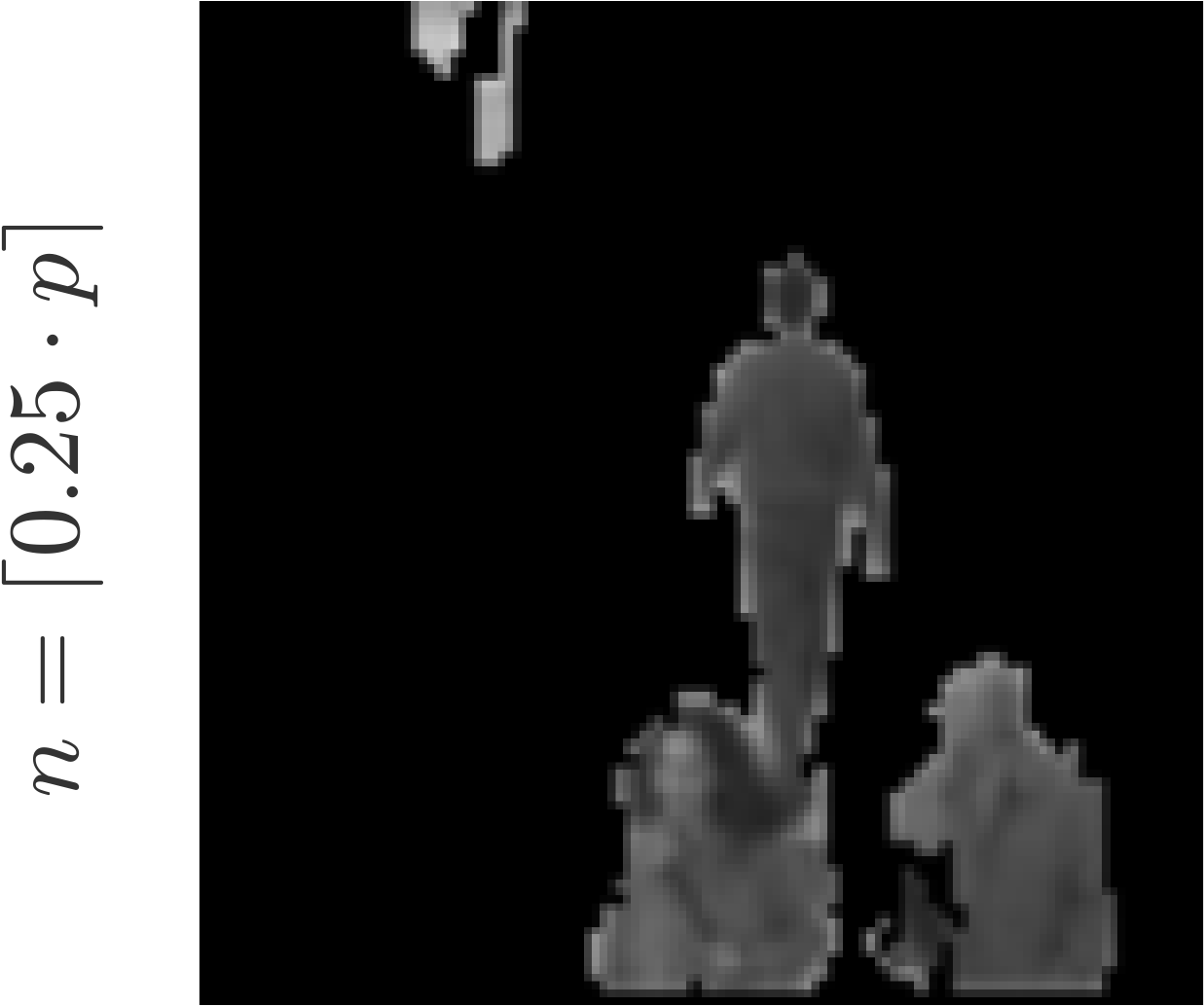}
	\includegraphics[width = 0.15\linewidth]{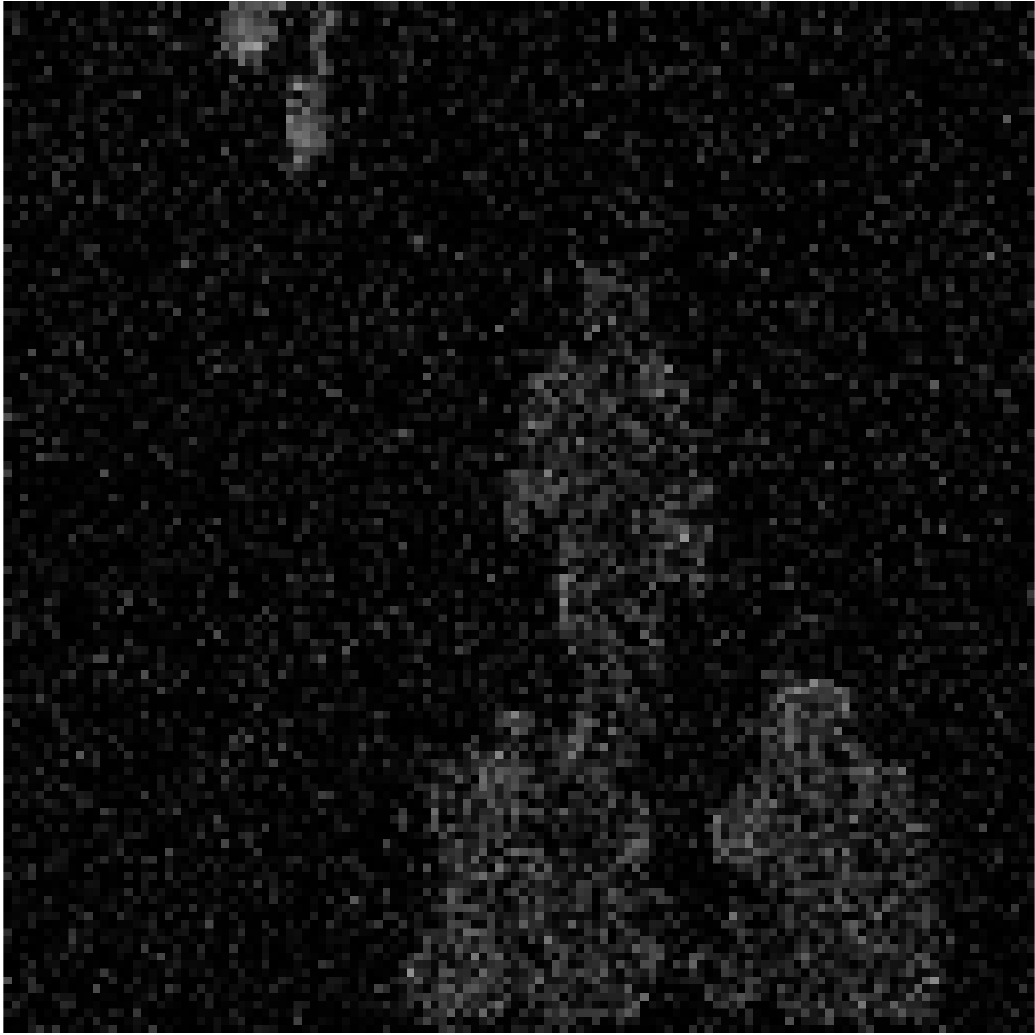}
	\includegraphics[width = 0.15\linewidth]{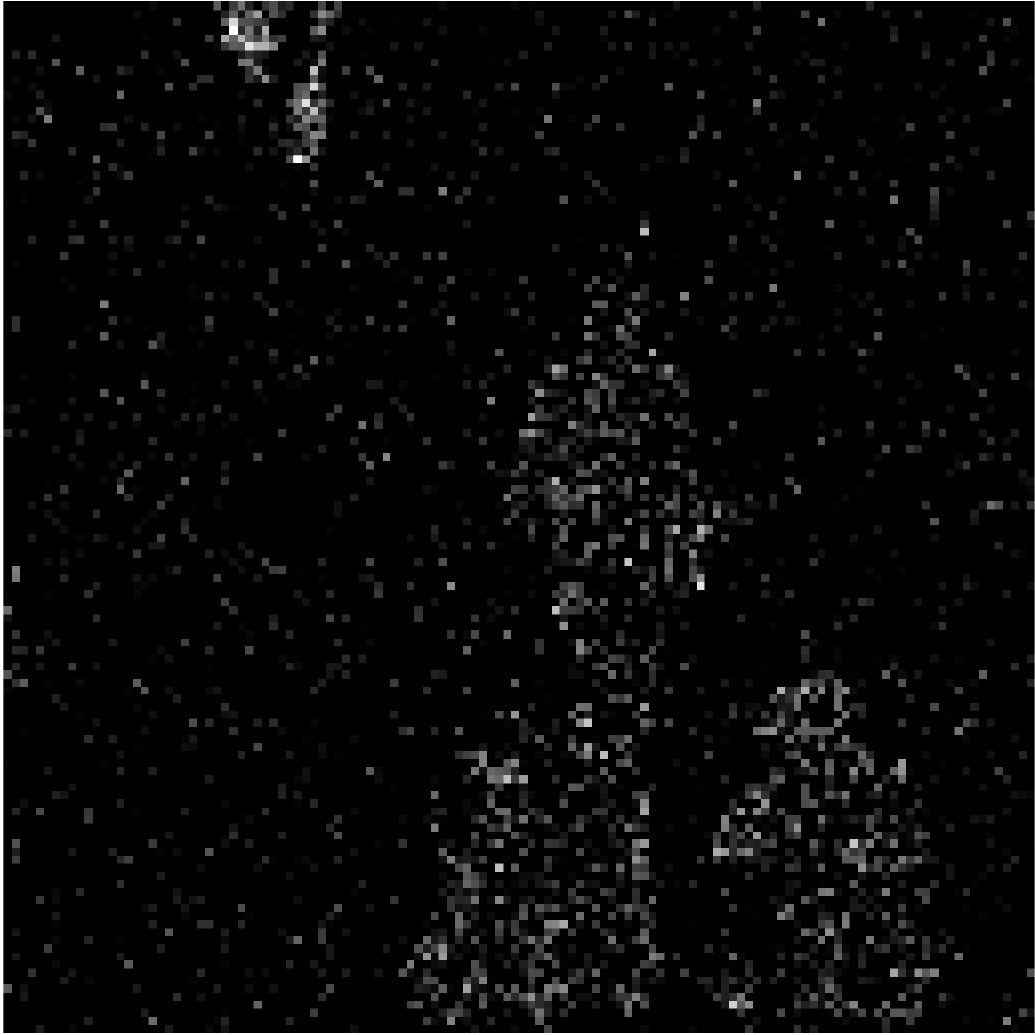}
	\includegraphics[width = 0.15\linewidth]{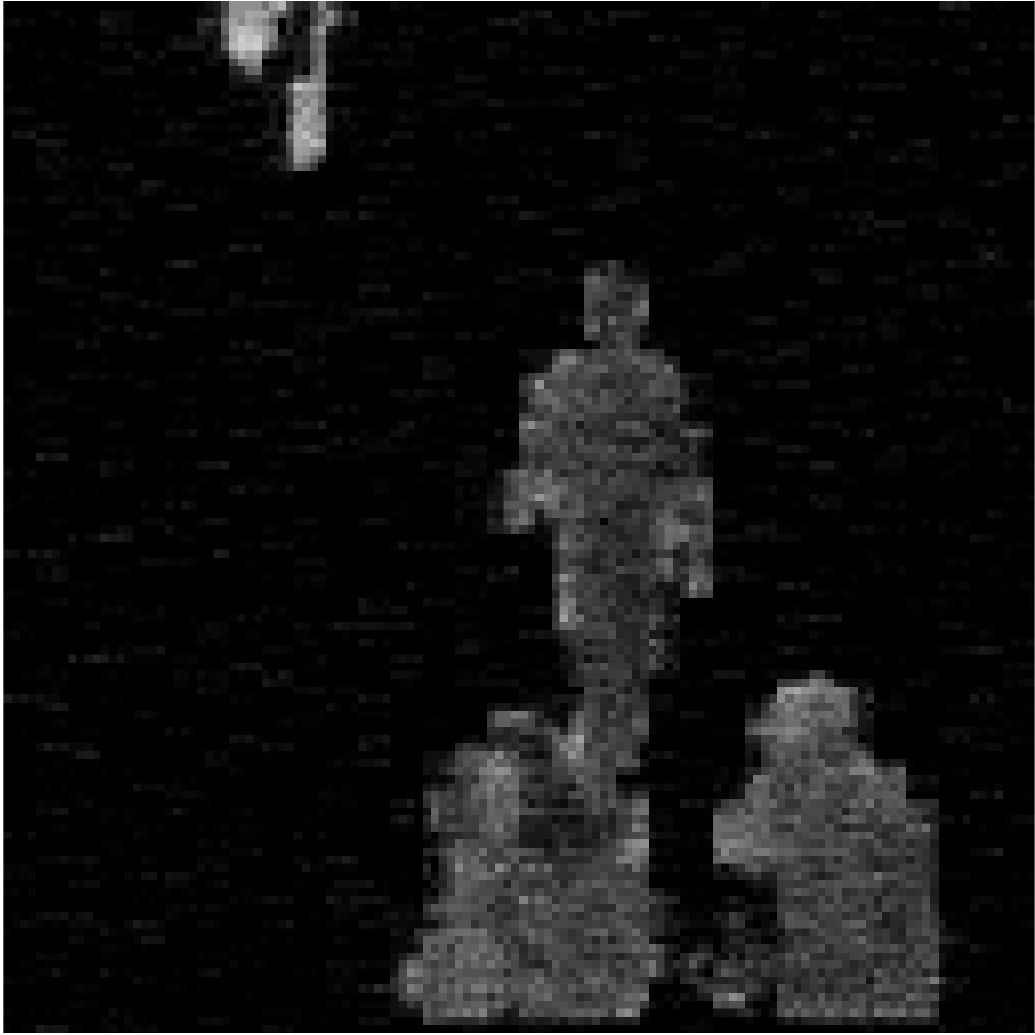}
	\includegraphics[width = 0.15\linewidth]{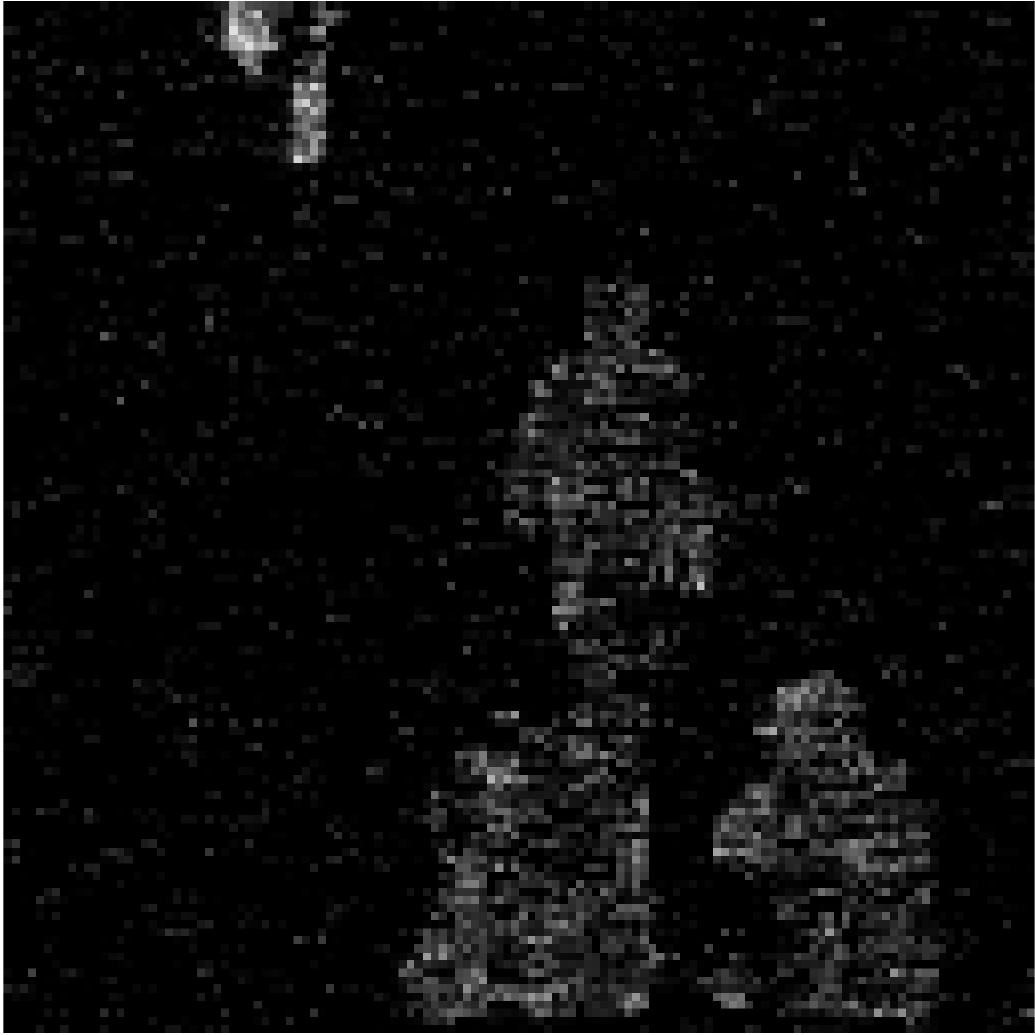} \vspace{-2mm} \\	
	\begin{minipage}{0.15\linewidth}
	\centering  \small{$\phantom{a}$}
	\end{minipage}	
	\begin{minipage}{0.15\linewidth}
	\centering  \small{$19.8 ~\rm{dB}$}
	\end{minipage}
	\begin{minipage}{0.15\linewidth}
	\centering \small{$17.8 ~\rm{dB}$}
	\end{minipage}
	\begin{minipage}{0.15\linewidth}
	\centering \small{$29.0 ~\rm{dB}$}
	\end{minipage}
	\begin{minipage}{0.15\linewidth}
	\centering \small{$20.3 ~\rm{dB}$}
	\end{minipage}\\
	\includegraphics[width = 0.18\linewidth]{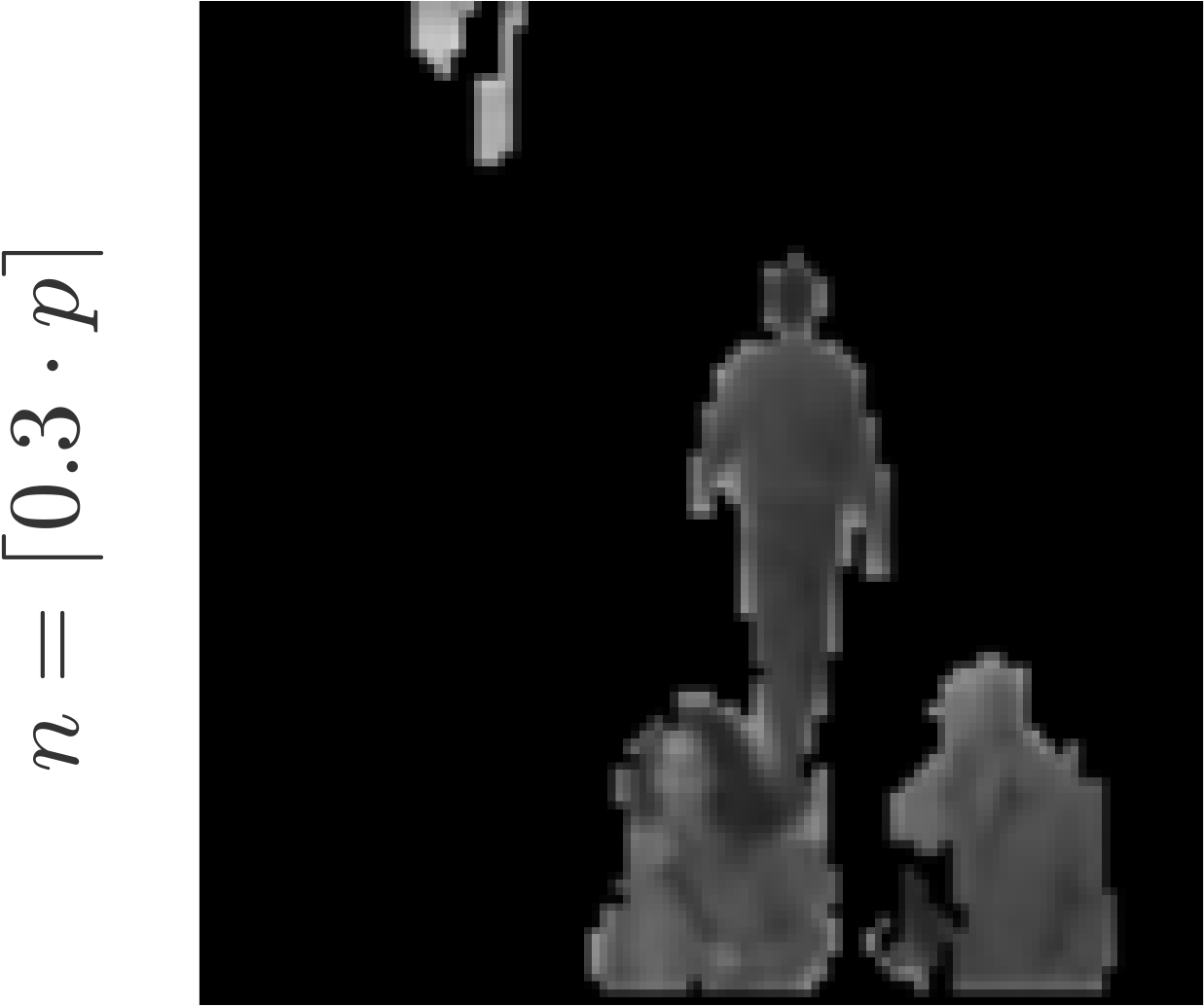}
	\includegraphics[width = 0.15\linewidth]{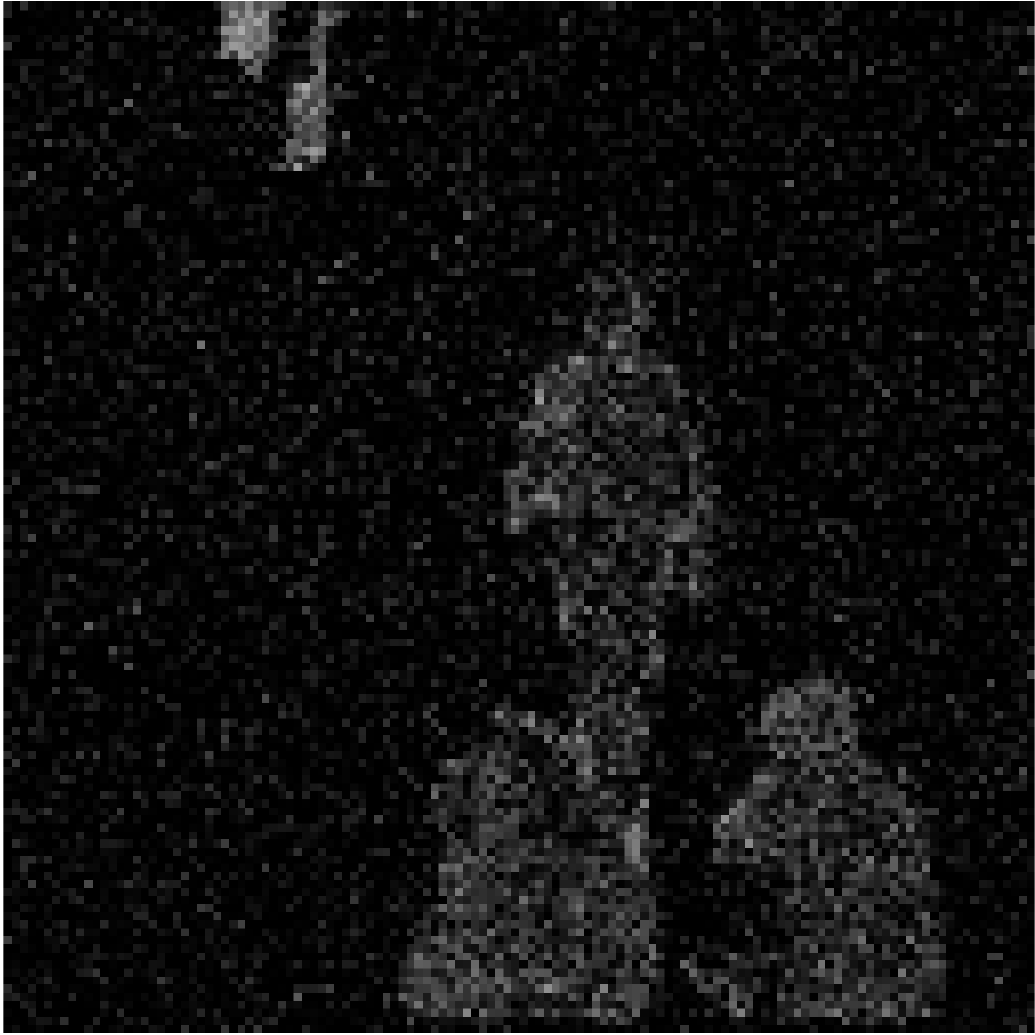}
	\includegraphics[width = 0.15\linewidth]{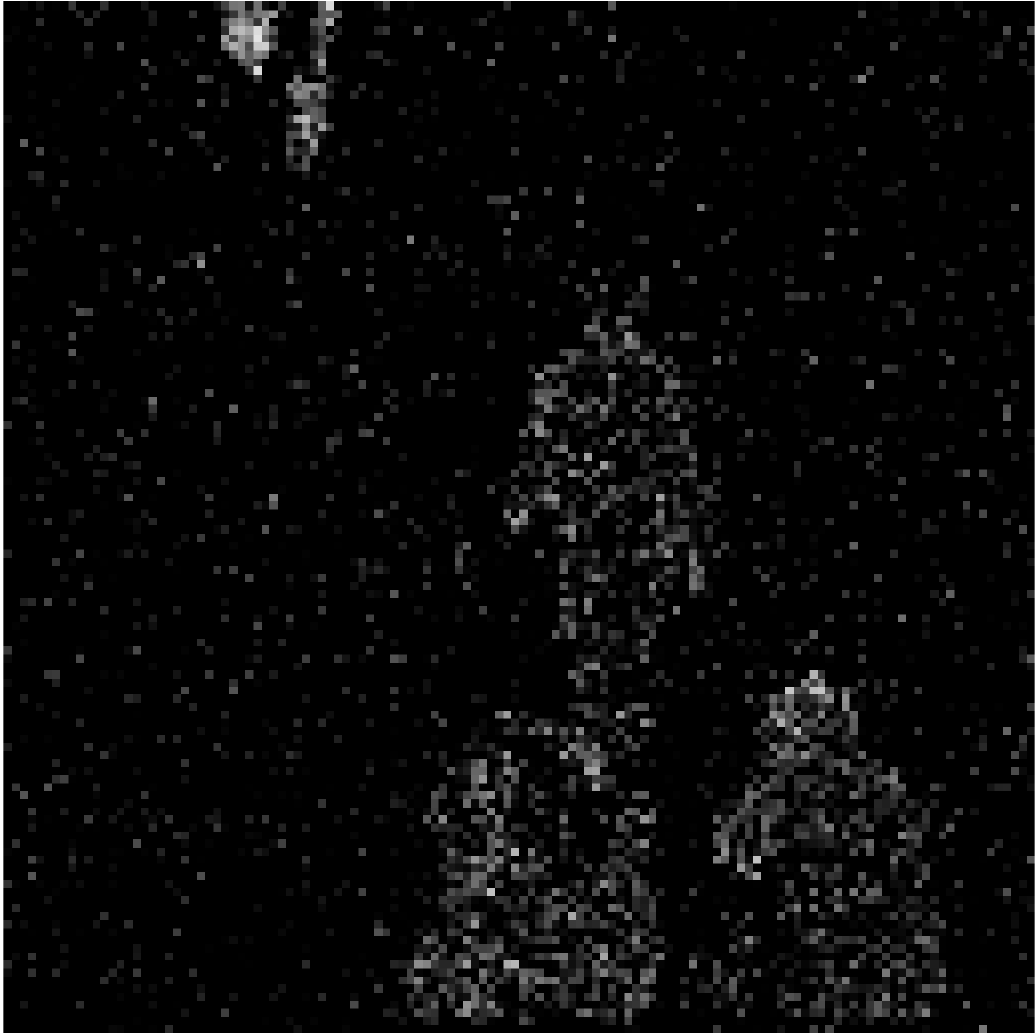}
	\includegraphics[width = 0.15\linewidth]{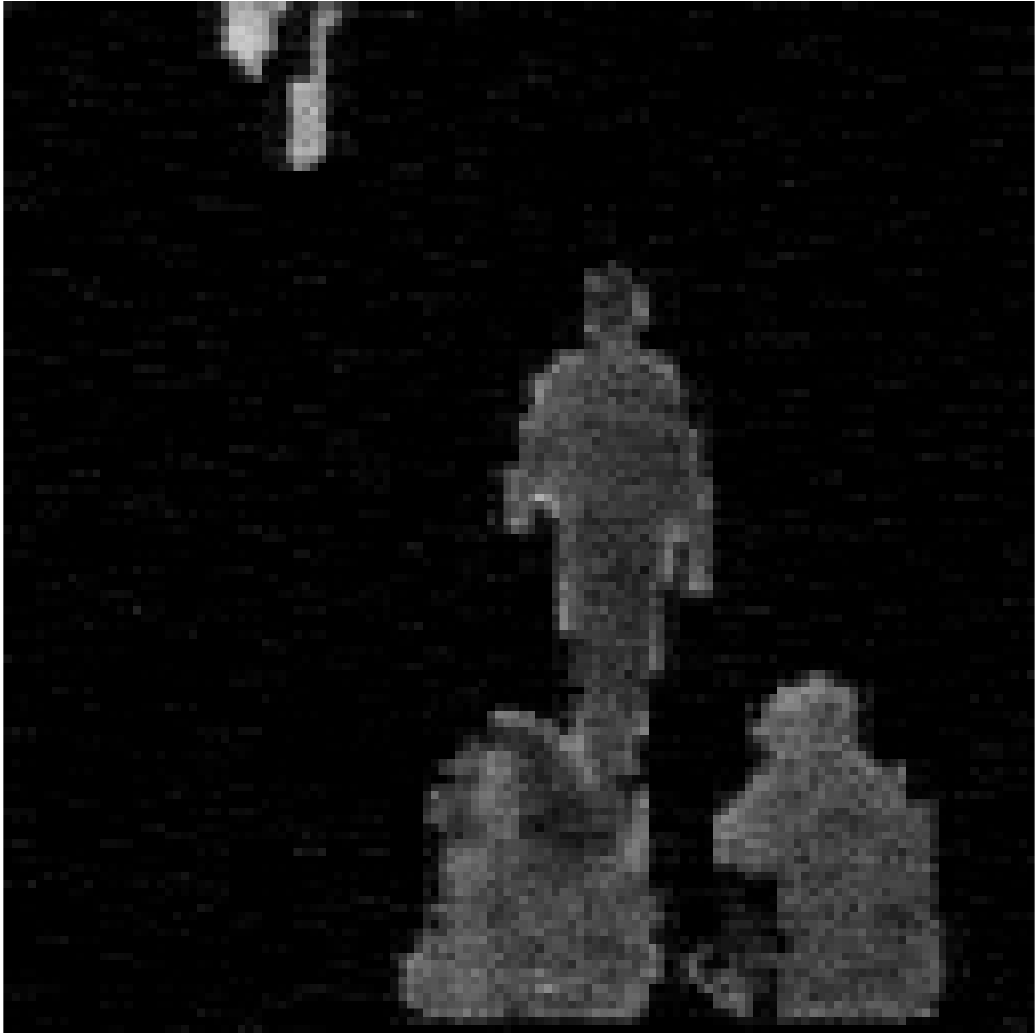}
	\includegraphics[width = 0.15\linewidth]{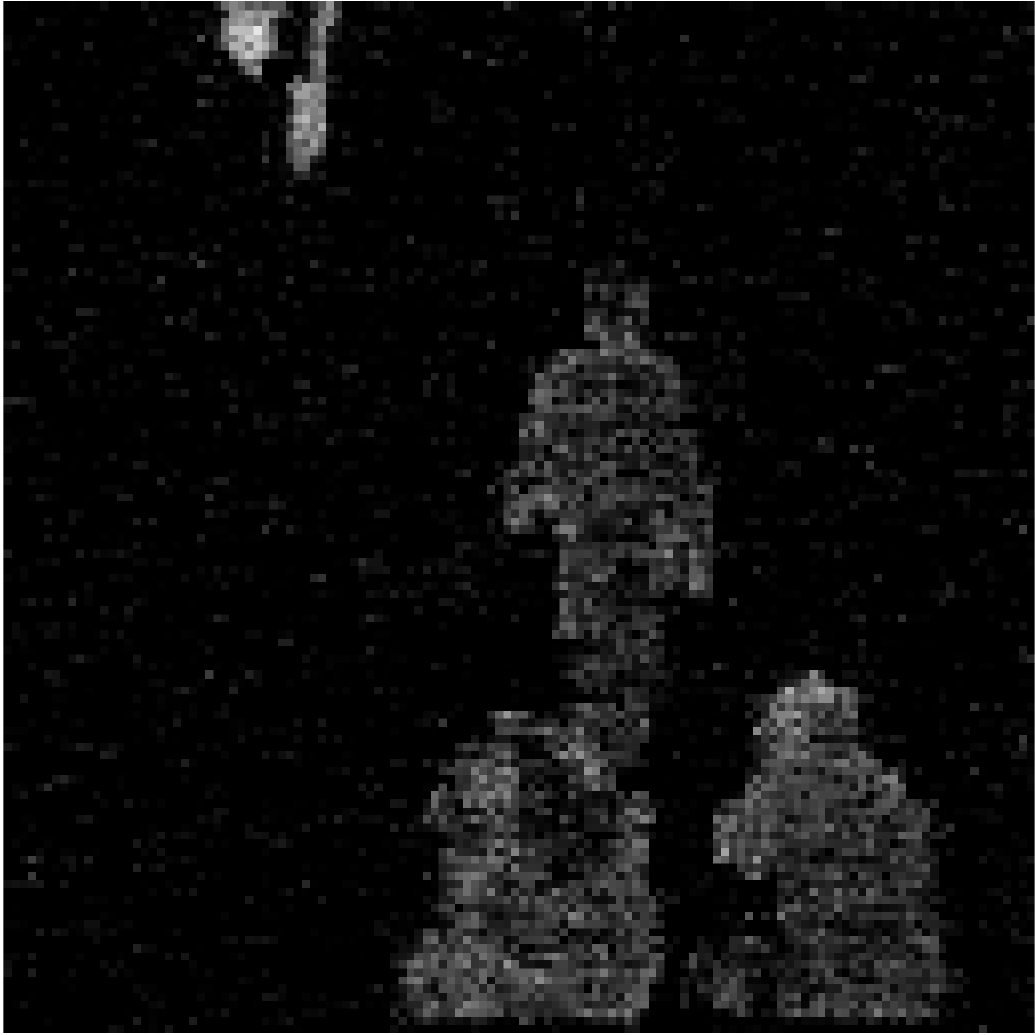} \vspace{-2mm} \\
	\begin{minipage}{0.15\linewidth}
	\centering  \small{$\phantom{a}$}
	\end{minipage}	
	\begin{minipage}{0.15\linewidth}
	\centering  \small{$20.6 ~\rm{dB}$}
	\end{minipage}
	\begin{minipage}{0.15\linewidth}
	\centering \small{$18.5 ~\rm{dB}$}
	\end{minipage}
	\begin{minipage}{0.15\linewidth}
	\centering \small{$31.5 ~\rm{dB}$}
	\end{minipage}
	\begin{minipage}{0.15\linewidth}
	\centering \small{$22.8 ~\rm{dB}$}
	\end{minipage}\\
	\includegraphics[width = 0.18\linewidth]{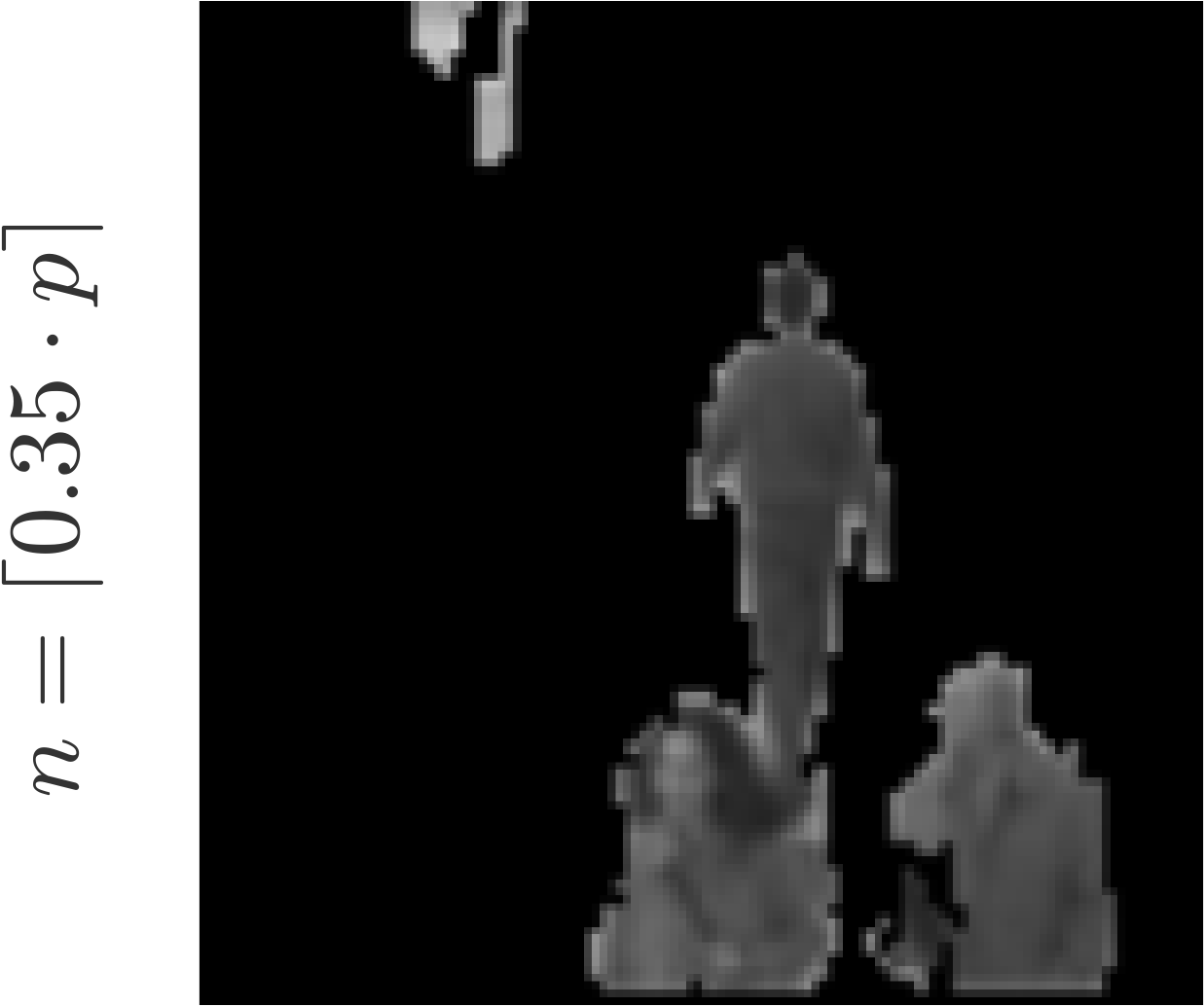}
	\includegraphics[width = 0.15\linewidth]{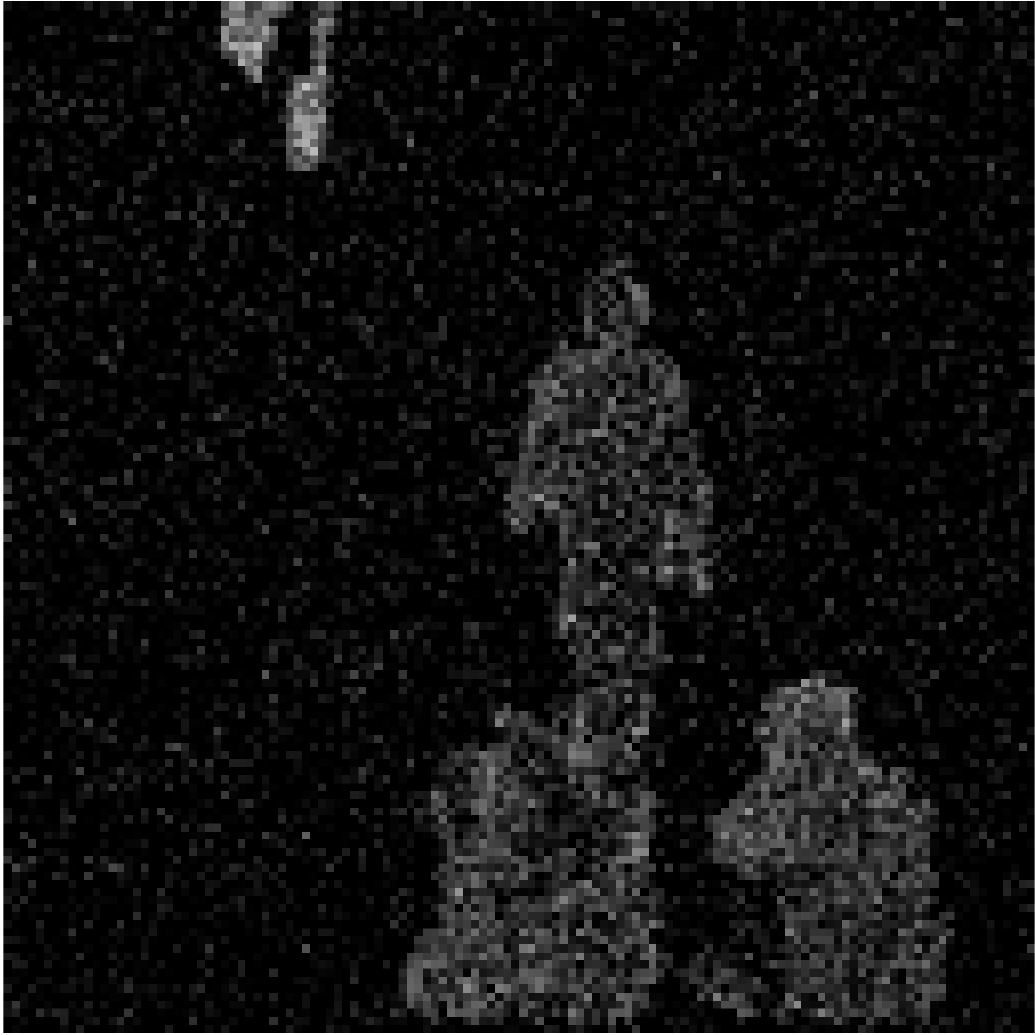}
	\includegraphics[width = 0.15\linewidth]{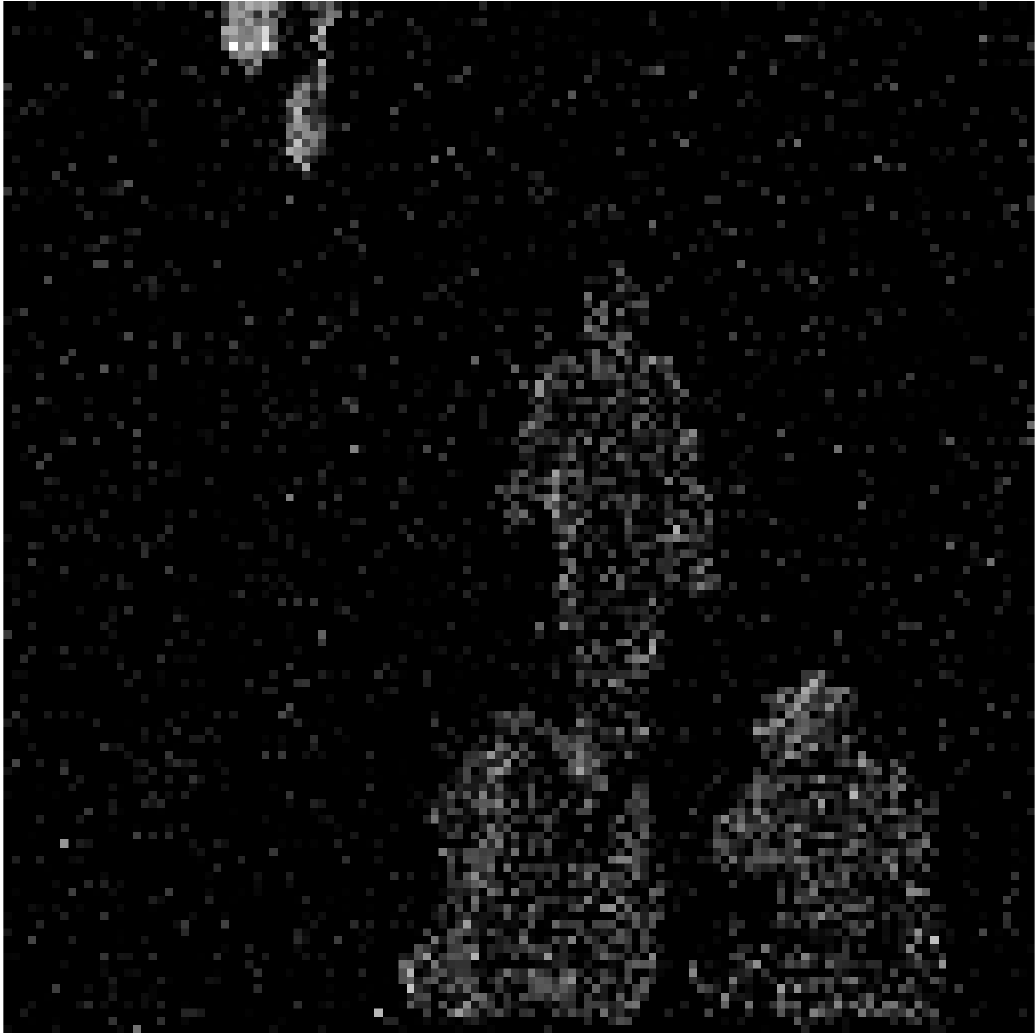}
	\includegraphics[width = 0.15\linewidth]{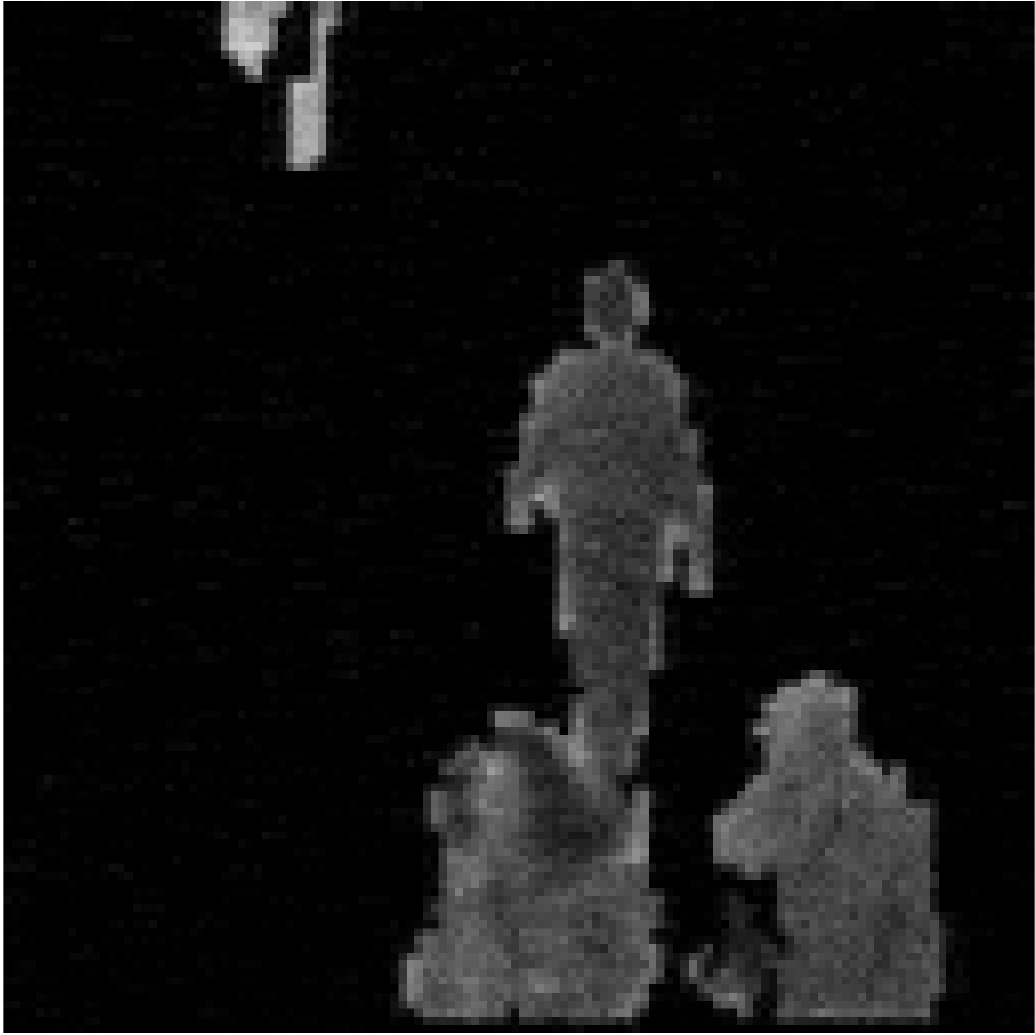}
	\includegraphics[width = 0.15\linewidth]{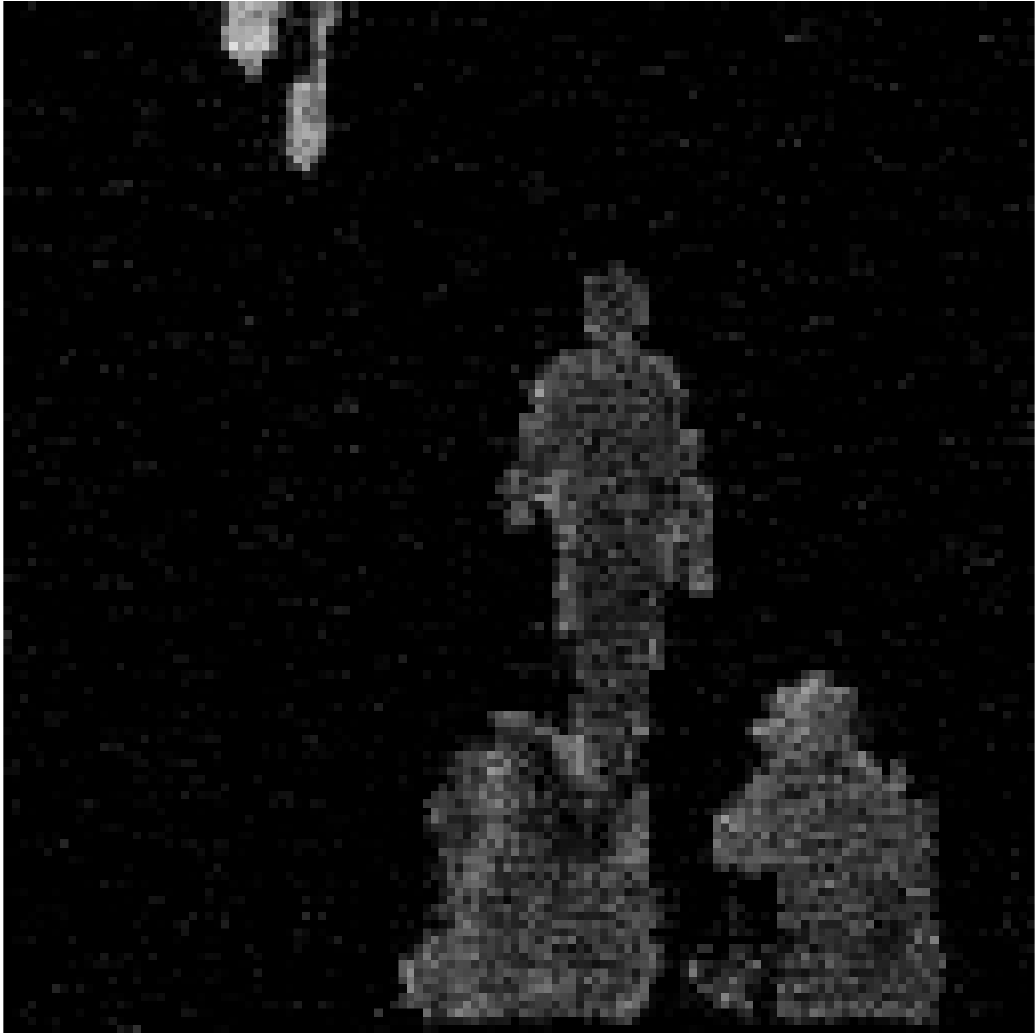} \vspace{-2mm} \\
	\begin{minipage}{0.15\linewidth}
	\centering  \small{$\phantom{a}$}
	\end{minipage}	
	\begin{minipage}{0.15\linewidth}
	\centering  \small{$21.7 ~\rm{dB}$}
	\end{minipage}
	\begin{minipage}{0.15\linewidth}
	\centering \small{$19.3~\rm{dB}$}
	\end{minipage}
	\begin{minipage}{0.15\linewidth}
	\centering \small{$34.9 ~\rm{dB}$}
	\end{minipage}
	\begin{minipage}{0.15\linewidth}
	\centering \small{$25.3 ~\rm{dB}$}
	\end{minipage} \\
	\includegraphics[width = 0.18\linewidth]{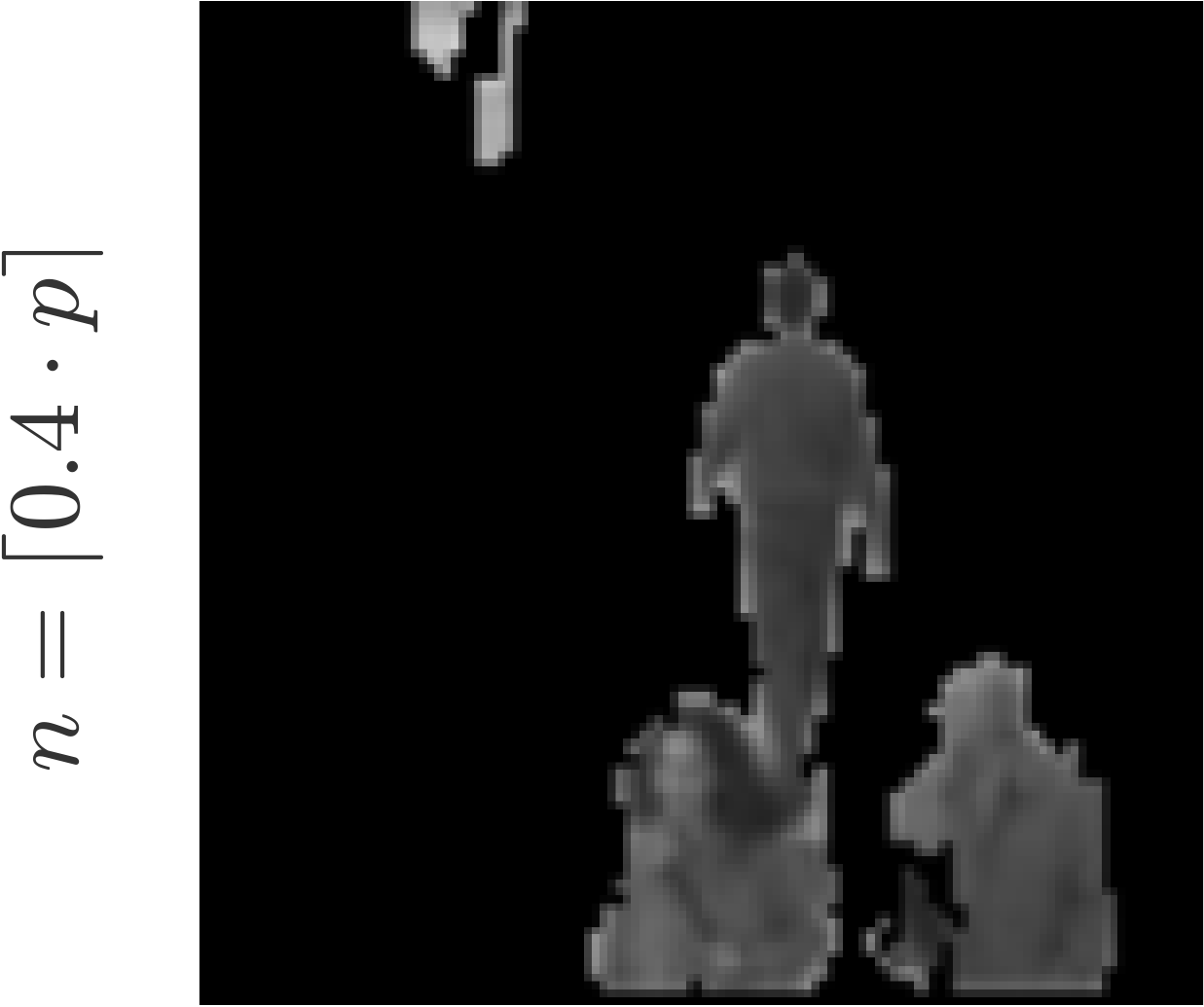}
	\includegraphics[width = 0.15\linewidth]{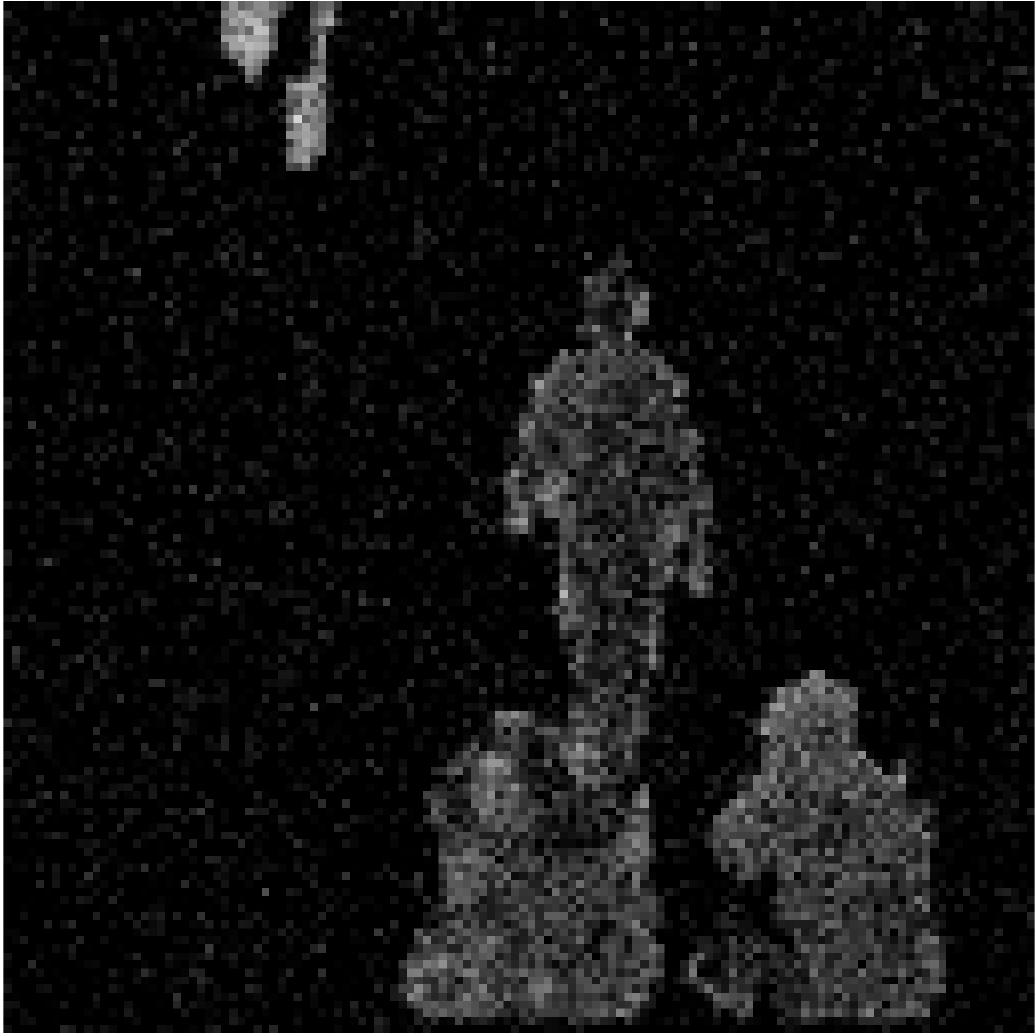}
	\includegraphics[width = 0.15\linewidth]{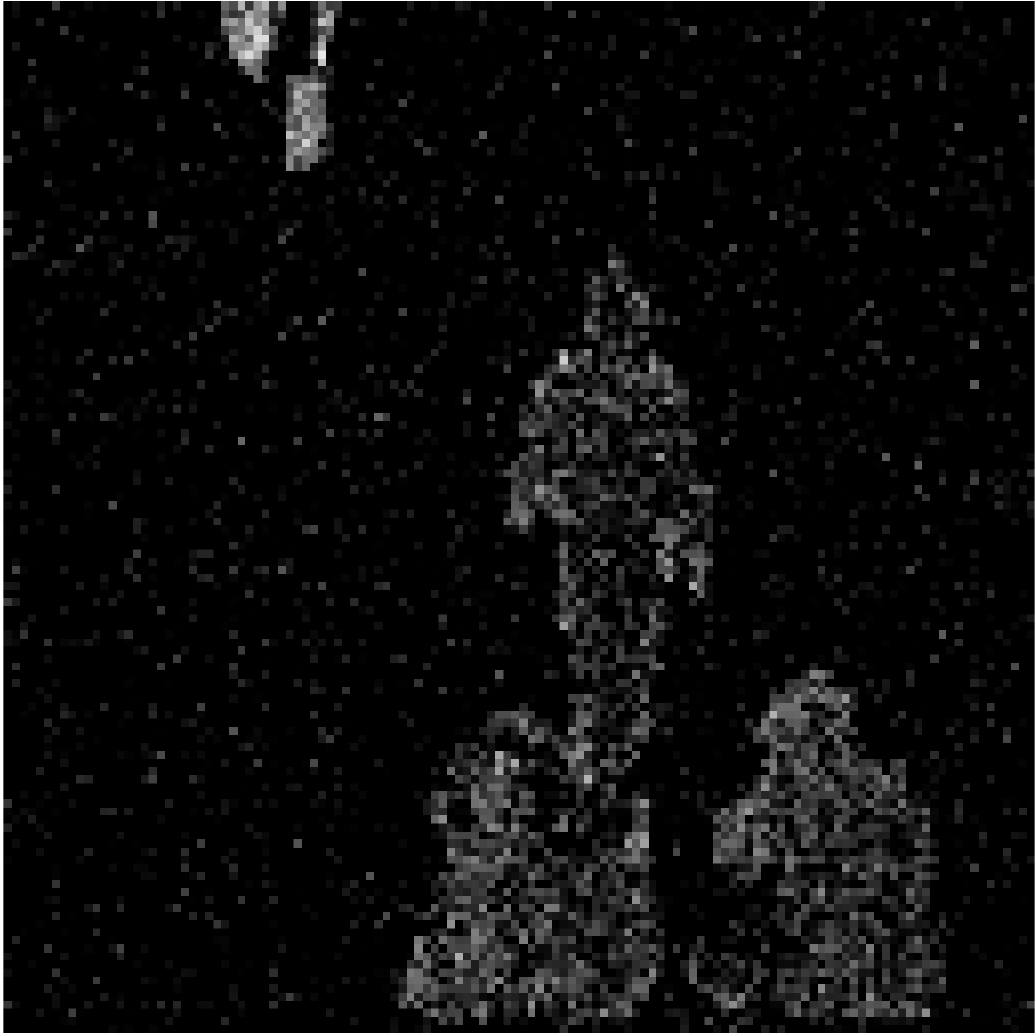}
	\includegraphics[width = 0.15\linewidth]{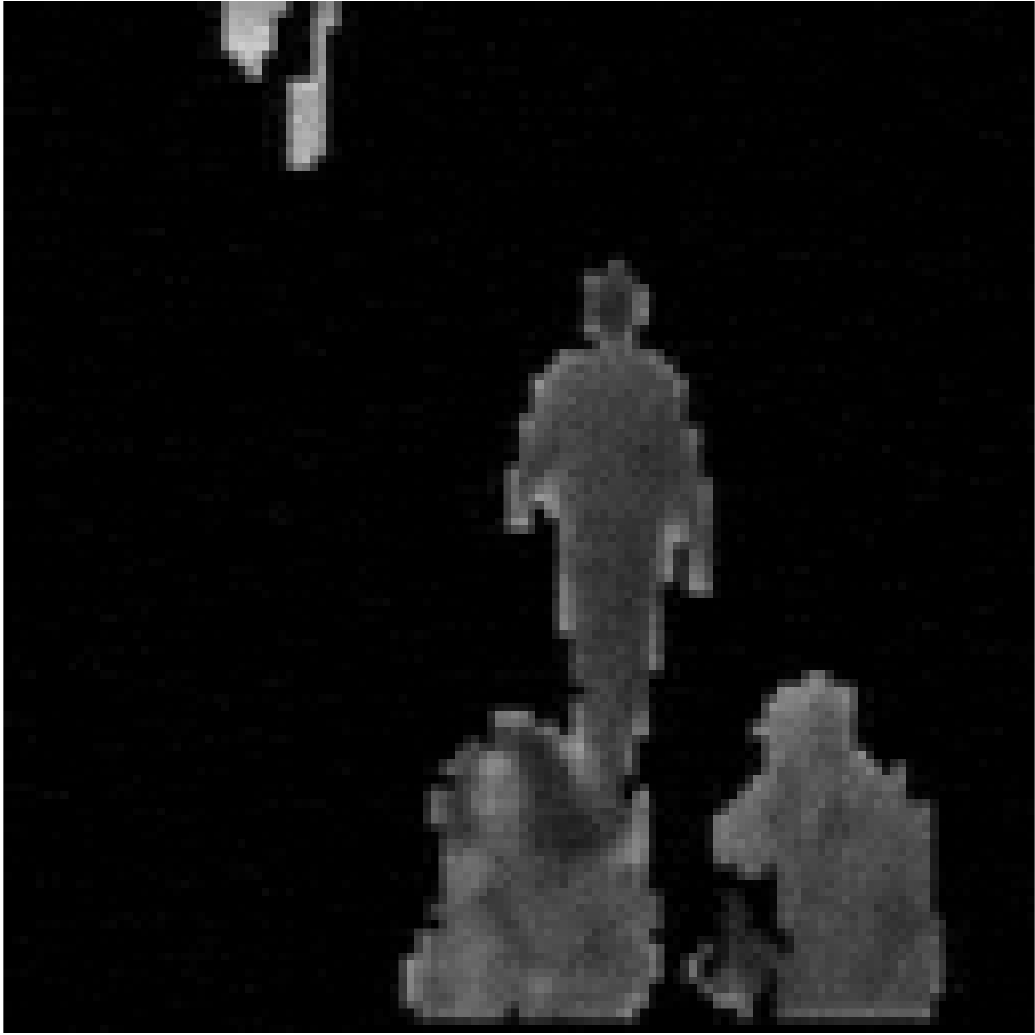}
	\includegraphics[width = 0.15\linewidth]{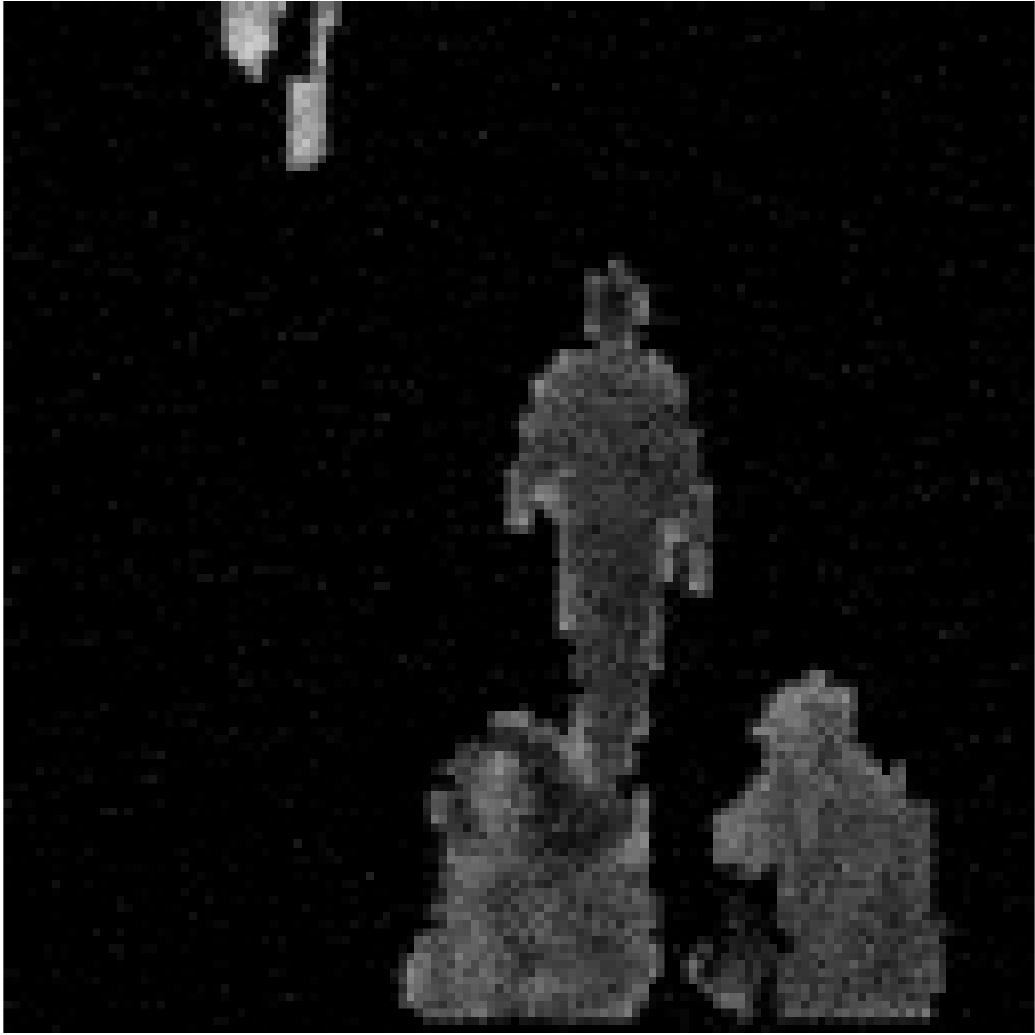} \vspace{-2mm} \\
	\begin{minipage}{0.15\linewidth}
	\centering  \small{$\phantom{a}$}
	\end{minipage}	
	\begin{minipage}{0.15\linewidth}
	\centering  \small{$23.9 ~\rm{dB}$}
	\end{minipage}
	\begin{minipage}{0.15\linewidth}
	\centering \small{$20.4 ~\rm{dB}$}
	\end{minipage}
	\begin{minipage}{0.15\linewidth}
	\centering \small{$39.1 ~\rm{dB}$}
	\end{minipage}
	\begin{minipage}{0.15\linewidth}
	\centering \small{$29.8 ~\rm{dB}$}
	\end{minipage} 

\caption{Results from real data. Representative examples of subtracted frame recovery from compressed measurements. Here, we consider a block sparse model fixed, with $g = 4$ block size per group. From top to bottom, the number of measurements range from $\lceil 0.25 \cdot \dim \rceil $ to $\lceil 0.4 \cdot \dim \rceil $, for $\dim = 2^{16}$. One can observe that one obtains better recovery as the number of measurements increases. } {\label{fig:004}}
\end{figure*}

\end{document}